%% file: subsetFVS.tex
\documentclass[runningheads]{llncs}
\input{preamble}

\begin{document}

\title{Subset Feedback Vertex Set in Chordal and Split Graphs}
\titlerunning{Subset Feedback Vertex Set in Chordal and Split Graphs}

\author{%
  Geevarghese Philip\inst{1}%
  \and%
  Varun Rajan\inst{2}%
  \and%
  Saket Saurabh\inst{3,4}%
  \and%
  Prafullkumar Tale\inst{5}%
  }

  \institute{%
    Chennai Mathematical Institute, Chennai, India.
    \email{gphilip@cmi.ac.in}%
    \and%
    Chennai Mathematical Institute, Chennai, India.
    \email{varunrajan09@gmail.com}%
    \and%
    The Institute of Mathematical Sciences, HBNI, Chennai, India.
    \email{saket@imsc.res.in}%
    \and%
    Department of Informatics, University of Bergen, Bergen,
    Norway.
    \and%
    The Institute of Mathematical Sciences, HBNI, Chennai, India.
    \email{pptale@imsc.res.in}}

\authorrunning{Philip, Rajan, Saurabh and Tale}

%\clearpage

\maketitle

\begin{abstract}
  In the \textsc{Subset Feedback Vertex Set (Subset-FVS)} problem
  the input is a graph $G$, a subset \(T\) of vertices of \(G\)
  called the ``terminal'' vertices, and an integer $k$. The task
  is to determine whether there exists a subset of vertices of
  cardinality at most $k$ which together intersect all cycles
  which pass through the terminals. \textsc{Subset-FVS}
  generalizes several well studied problems including
  \textsc{Feedback Vertex Set} and \textsc{Multiway Cut}. This
  problem is known to be \NP-Complete even in split graphs. Cygan
  et al. proved that \textsc{Subset-FVS} is fixed parameter
  tractable (\FPT) in general graphs when parameterized by $k$
  [SIAM J. Discrete Math (2013)]. In split graphs a simple
  observation reduces the problem to an equivalent instance of the
  $3$-\textsc{Hitting Set} problem with same solution size. This
  directly implies, for \textsc{Subset-FVS} \emph{restricted to
    split graphs}, (i) an \FPT algorithm which solves the problem
  in $\OhStar(2.076^k)$ time \footnote{The \(\OhStar()\) notation
    hides polynomial factors.}% for \textsc{Subset-FVS} in Chordal
  % Graphs
  [Wahlstr\"om, Ph.D. Thesis], and (ii) a kernel of size
  $\mathcal{O}(k^3)$. We improve both these results for
  \textsc{Subset-FVS} on split graphs; we derive (i) a kernel of
  size $\mathcal{O}(k^2)$ which is the best possible unless
  $\NP \subseteq \coNP/{\sf poly}$, and (ii) an algorithm which
  solves the problem in time $\mathcal{O}^*(2^k)$. Our algorithm,
  in fact, solves \textsc{Subset-FVS} on the more general class of
  \emph{chordal graphs}, also in $\mathcal{O}^*(2^k)$ time.
 \end{abstract}

\input{intro}
\input{prelims}

\input{kernel-split}
\input{kernel_lowerbound}
\input{fpt-chordal}
\input{conclusion}
\bibliographystyle{plain}
\bibliography{references}
\end{document}

%% file: intro.tex
%!TEX root = subsetFVS.tex

\section{Introduction}
In a \emph{covering} or \emph{transversal problem} we are given a
universe of elements $U$, a family $\cal F$ ($\cal F$ could be
given implicitly), and an integer $k$, and the objective is to
check whether there exists a subset of $U$ of size at most $k$
which intersects all the elements of $\cal F$. A number of natural
problems on graphs are of this form.  For instance, consider the
classical \textsc{Feedback Vertex Set (FVS)} problem. Here, given
a graph $G$ and a positive integer $k$, the objective is to decide
whether there exists a vertex subset $X$ (a \emph{feedback vertex
  set} of \(G\)) of size at most $k$ which intersects all cycles,
that is, for which the graph $G- X$ is a forest. Other examples
include {\sc Odd Cycle Transversal}, {\sc Directed Feedback Vertex
  Set} and {\sc Vertex Cover (VC)}.  These problems have been
particularly well studied in parameterized
complexity~\cite{cygan2013subset,DBLP:journals/talg/ChitnisCHM15,kawarabayashi2012fixed,lokshtanov2018linear,wahlstrom2014half,DBLP:journals/talg/Thomasse10}.

Recently, a natural generalization of covering problems has
attracted a lot of attention from the point of view of
parameterized complexity. In this generalization, apart from $U$,
$\cal F$ and $k$, we are also given a subset $T$ of $U$ and the
objective is to decide whether there is a subset of $U$ of size at
most $k$ that intersects all those sets in $\cal F$ \emph{which
  contain an element in $T$}. This leads to the \emph{subset
  variant} of classic covering problems; typical examples include
{\sc Subset Feedback Vertex Set (Subset-FVS)}, {\sc Subset
  Directed Feedback Vertex Set} and {\sc Subset Odd Cycle
  Transversal}. These three problems have received considerable
attention and they have all been shown to be {\em fixed-parameter
  tractable} ({\FPT}) with \(k\) as the
parameter~\cite{cygan2013subset,DBLP:journals/talg/ChitnisCHM15,kawarabayashi2012fixed}.

%
%
%The \textsc{Feedback Vertex Set (FVS)} and \textsc{Vertex Cover(VC)} problems are arguably the most studied problems in parameterized complexity. Both problems fall into the category of \emph{vertex deletion problems} where the task is to delete minimum number of vertices to all obstructions.
%For \textsc{VC} and \textsc{FVS} problems, the obstructions are set of all edges and set of all cycles, respectively.
%
%
%We can rounghly associate difficulty of solving vertex deletion problems with ease of stating all obstructions that needs to be killed.
%%The problem harder if it is difficult to 
%We note that a simple single exponent parameterized algorithm for \textsc{VC} was known in $1980$s whereas first single exponent algorithm for \textsc{FVS} appeared in only $2005$\cite{blah}.

In this paper we study the {\sc Subset-FVS} problem when the input
%belongs to a more structured families of graphs.
is a \emph{split} graph or, more generally, a \emph{chordal}
graph.
%\textsc{Subset Feedback Vertex Set (Subset-FVS)} is a generelization of \textsc{FVS}. In this problem, the %obstructions are some specified cycles instead of all cyles in the graph. Because of this, \textsc{Subset-FVS} %is harder to tackle than \textsc{FVS}. In this problem, the input is a graph $G$, subset of vertices $T$, called %terminals, and an integer $k$. The task is to determine whether there exists a subset of vertices of cardinality %at most $k$ which intersect all cycles passing through terminals. 
The {\sc Subset-FVS} problem was introduced by Even et
al.\cite{even20008}, and generalizes several well-studied problems
like \textsc{FVS, VC}, and \textsc{Multiway Cut}
\cite{fomin2014enumerating}.
% Fedor et.al. gave a parameter preserving reduction from \textsc{VC} to \textsc{Subset-FVS} (\cite{fomin2014enumerating} Theorem $2.1$).
The question whether the \textsc{Subset-FVS} problem is fixed
parameter tractable (\FPT) when parameterized by the solution size
was posed independently by Kawarabayashi and the third author in
$2009$. Cygan et al. \cite{cygan2013subset} and Kawarabayashi and
Kobayashi \cite{kawarabayashi2012fixed} independently answered
this question positively in $2011$. Wahlstr{\"o}m
\cite{wahlstrom2014half} gave the first parameterized algorithm
where the dependence on $k$ is $2^{\Oh(k)}$. Lokshtanov et
al. \cite{lokshtanov2018linear} presented a different \FPT\
algorithm which has linear dependence on the input size.  On the
%other hand,
flip side,
%number of vertices and edges of input graph.
Fomin et al. presented a parameter preserving reduction from
\textsc{VC} to \textsc{Subset-FVS}~\cite[Theorem
$2.1$]{fomin2014enumerating}, thus ruling out the existence of
an % which rules out possibility of an
algorithm with \emph{sub-exponential} dependence on $k$ under the
Exponential-Time Hypothesis. Most recently, Hols and
Kratsch~\cite{hols2018randomized} used matroid-based tools to show
that \textsc{Subset-FVS} has a randomized polynomial kernelization
with $\mathcal{O}(k^9)$ vertices. % using matroid-based tools.
%of Kratsch and Wahlstr{\"o}m \cite{kratsch2012representative}. 

%Another evidence that \textsc{Subset-FVS} is harder to tackle than \textsc{FVS} is the fact that on chordal graphs later problem is solvable in polynomial time \cite{blah} but \textsc{Subset-FVS} is \NP-complete.

All the results that we described above hold for arbitrary input
graphs.  The \textsc{Subset-FVS} problem has also been studied
with the input restricted to various families of graphs; in
particular, to chordal graphs and split graphs. Recall
that\footnote{See \autoref{sec:prelims} for formal definitions.}
(i) a graph is \emph{chordal} if it does not contain induced
cycles of length four or larger, (ii) a split graph is one whose
vertex set can be partitioned into a clique and an independent
set, and (iii) every split graph is chordal. The problem remains
\NP-Complete even on split
graphs~\cite{fomin2014enumerating}.  % The stark difference in
% difficulty between \textsc{Subset-FVS} and \textsc{FVS} in this
% graph class is also witnessed by the fact that known upper bounds
% on the number of minimal feedback vertex sets and minimal subset
% feedback vertex sets in split graphs are $n^2$ and $3^{n/3}$
% respectively~\cite{fomin2014enumerating}. It goes without saying
% that {\sc FVS} is polynomial time solvable on chordal graphs while
% \textsc{Subset-FVS} is \NP-Complete on split graphs. 
Golovach et al. designed the first exact exponential time
algorithm for \textsc{Subset-FVS} on chordal graphs in their
pioneering work~\cite{golovach2014subset}; this algorithm runs in
\(\OhStar(1.6708^{n})\) time on a chordal graph with \(n\)
vertices. Chitnis et al. improved this bound to
\(\OhStar(1.6181^{n})\)~\cite{chitnis2013faster}.

In this article we study \textsc{Subset-FVS} on chordal and split
graphs from the point of view of parameterized complexity. For a
given set of vertices $T$, a \emph{$T$-cycle} is a cycle which
contains at least one vertex from $T$. Formally, we study the
following problem on chordal graphs:

\defparproblem{\SFVSC}{A chordal graph $G=(V,E)$, a set of
  \emph{terminal vertices} $T\subseteq{}V$, and an integer
  $k$}{$k$}{Does there exist a set $S\subseteq{}V$ of at most $k$
  vertices of \(G\) such that the subgraph \(G[V\setminus{}S]\)
  contains no $T$-cycle?}

\smallskip
\noindent 
When the input graph in \SFVSC is a \emph{split} graph then we
call it the \textbf{\SFVSS} problem.
%\defparproblem{\SFVSS}{A split graph $G=(V,E)$, a set of
%  \emph{terminal vertices} $T\subseteq{}V$, and an integer
%  $k$}{$k$}{Does there exist a set $X\subseteq{}V$ of at most $k$
%  vertices of \(G\) such that the subgraph \(G[V\setminus{}X]\)
%  contains no $T$-cycle?}

It is a simple observation (see \autoref{lem:triangles}) that in
order to intersect every \(T\)-cycle in a \emph{chordal} graph it
is sufficient---and necessary---to intersect all
$T$-\emph{triangles} in the graph. This yields a
parameter-preserving reduction from \SubsetFVSSplit to
\textsc{3-Hitting Set (3-HS)}.
%On positive side, this formulation in terms of an instance of
%\textsc{$3$-Hitting Set} is useful in various settings.
This, in turn, implies the existence of a polynomial kernel for
\SubsetFVSSplit, of size
$\mathcal{O}(k^3)$~\cite{DBLP:journals/jcss/Abu-Khzam10}, and an
\FPT{} algorithm which solves the problem in time
\(\OhStar(2.076^{k})\)~\cite{wahlstrom2007algorithms}.  Note that
when we formulate \SFVSS in terms of \textsc{$3$-HS} in this
manner, we lose a lot of \emph{structural information} about the
input graph. It is natural to suspect that this lost information
could have been exploited to obtain better algorithms and smaller
kernels for the original problem.  This was most recently
vindicated by the work of Le et al.~\cite{le2018subquadratic} who
designed kernels with a sub-quadratic number of vertices for
several implicit \textsc{$3$-HS} problems on graphs, improving on
long-standing quadratic upper bounds in each case. Our work is in
the same spirit as that of Le et al.: we obtain improved results
for two implicit \textsc{3-Hitting Set} problems---namely:
intersecting all \(T\)-triangles in chordal (respectively, split)
graphs---by a careful analysis of structural properties of the
input graph.

\vspace{0.2cm}

\noindent \textbf{Our results and methods:} Our main result is a
quadratic-size kernel for \SFVSS, with a linear-sized ``clique
side''; more
precisely: 

\begin{restatable}{theorem}{splitkernel}\label{thm:split-kernel}
  There is a polynomial-time algorithm which, given an instance
  $(G; T; k) $ of \SFVSS, returns an instance \((G';T';k')\) of
  \SFVSS such that (i) \((G;T;k)\) is a \YES instance if and only
  if \((G';T';k')\) is a \YES instance, and (ii)
  \(|V(G')|=\Oh(k^{2})\), \(|E(G')|=\Oh(k^{2})\), and
  \(k'\leq{}k\). Moreover, the split graph \(G'\) has a split
  partition \((K',I')\) with \(|K'|\leq{}10k\).
\end{restatable}

\noindent 
Our kernelization algorithm for \SubsetFVSSplit involves
non-trivial applications of the Expansion Lemma, a combinatorial
tool which was central to the design of the quadratic kernel for
{\sc Undirected FVS}~\cite{DBLP:journals/talg/Thomasse10}.  % Our
% kernelization algorithm works in phases.
Given an input graph $(G,T,k)$ and a split partition $(K,I)$ of
$V(G)$, where $K$ is a clique and $I$ is an independent set, we
first reduce the input to an instance \((G; T; k)\) where the
terminal set \(T\) is exactly the independent set \(I\) from a
split partition \((K,I)\) of \(G\). Then we show that if a
(non-terminal) vertex \(v\in{}K\) has at least \(k+1\) neighbours
in \(I\) then we can either include \(v\) in a solution, or safely
delete an edge incident with \(v\); we use the Expansion Lemma to
identify such an {\em irrelevant edge} incident to $v\in K$. This
leads to an instance where each \(v\in{}K\) has at most \(k\)
neighbours in \(I\). We now apply the Expansion Lemma to this
instance to bound the number of vertices in \(K\) by \(10k\); this
gives the bound of \(\Oh(k^{2})\) on \(|I|\).

We complement this upper bound with a matching lower bound on the
bit size of any kernel for this problem:
\begin{restatable}{theorem}{lowerbound}\label{thm:lowerbound}
  For any $\epsilon > 0$, the \SFVSS problem parameterized by the
  solution size does not admit a polynomial kernel of size
  $\Oh(k^{2 - \epsilon})$ bits, unless $\NP \subseteq \coNP/poly$.
\end{restatable}

Our third result is an improved \FPT\ algorithm for
\SubsetFVSChordal which, in addition, runs in time \emph{linear}
in the size of the input graph: % improves over the running time of
% \textsc{$3$-HS}.
\begin{restatable}{theorem}{chordalfpt}\label{thm:chordalfpt}  
\SubsetFVSChordal admits an algorithm with running time $\mathcal{O}(2^k (n + m))$. Here $n, m$ are the  number of vertices and  the edges of the input graph $G$, respectively.
% 
%There exists an algorithm which given an instance $(G, T, k)$ of \SubsetFVSChordal\ runs in time $\mathcal{O}(2^k (n + m))$ and decides whether input is \yes\ of \no\ instance. Here $n, m$ are number of vertices and edges in input graph $G$. 
\end{restatable}
\noindent 
We obtain this improvement by designing a branching strategy based
on a careful analysis of the \emph{clique-tree} structure of the
input chordal graph.  % The structure provided by clique-tree
% is essential for our improvement.  We crucially use the
% properties of \emphlique-tree} defined for given chordal graphs
% to obtain running time %which is better than $\textsc{3-HS}$.
% We also note that in running time, dependency on the number of
% edges and the number of vertices of input graph is linear.

%\todo{Write the FPT and kernel bounds as two separate theorems.}

\vspace{0.2cm}

\noindent \textbf{Organization of the rest of the paper:} We state
various definitions and preliminary results in
\autoref{sec:prelims}. We prove the quadratic kernel upper bound
for \SFVSS
% (\autoref{thm:split-kernel})
in \autoref{sec:kernelbound} and the corresponding lower bound
%(\autoref{thm:lowerbound})
in \autoref{sec:lowerbound}. We derive the \(\OhStar(2^{k})\)
algorithm for \SFVSC in \autoref{sec:fpt-algorithm}, and conclude
in \autoref{sec:conclusion}.

%%% Local Variables:
%%% mode: latex
%%% TeX-master: "subsetFVS"
%%% End:

%% file: prelims.tex
\section{Preliminaries}\label{sec:prelims} 

\subsection{Graphs} All our graphs are finite, undirected, and
simple. We mostly conform to the graph-theoretic notation and
terminology from the book of Diestel~\cite{diestel2017graph}. We
describe the few differences from Diestel and some notation and
terms of our own, and reproduce some definitions from Diestel for
ease of reference. Let \(G\) be a graph. For a vertex $v$ in
$V(G)$ we use $N_G(v)$ to denote the \emph{open neighbourhood}
$\{u \in V(G) \mid vu \in E(G)\}$ of \(v\), and \(N_{G}[v]\) to
denote its closed neighbourhood \(N(v)\cup\{v\}\). We drop the
subscript $G$ when there is no ambiguity.  The \emph{length} of a
path or cycle is the number of \emph{edges} in the path (or
cycle).  An edge $uv$ is a \emph{bridge}
if it is not contained in any cycle of \(G\). An edge \(e\) in
\(G\) is a \emph{chord} of a cycle \(C\) if (i) both the
endvertices of \(e\) are in \(C\), and (ii) edge \(e\) is not in
\(C\). An \emph{induced cycle} is a cycle which has no chord. A
vertex \(v\) of degree exactly one in a tree \(T\) is a
\emph{leaf} of the tree, \emph{unless} \(v\) is the designated
\emph{root} vertex (if one exists) of \(T\). A \emph{matching}
\(M\) in \(G\) is a set of edges no two of which share an
endvertex. The endvertices of the edges in \(M\) are
\emph{saturated} by \(M\). \(M\) is \emph{between} vertex sets
\(X,Y\) if \(X\cap{}Y=\emptyset\) and each edge in \(M\) has one
end each in \(X\) and \(Y\), respectively. A \emph{maximum
  matching} of \(G\) is a matching of the largest size in \(G\).

Let \(S\subseteq{}V(G)\) and \(F\subseteq{}E(G)\) be a vertex
subset and an edge subset of \(G\), respectively. We use (i)
$G[S]$ to denote the subgraph of \(G\) \emph{induced} by \(S\),
(ii) $G - S$ to denote the graph \(G[V\setminus{}S]\), and (ii)
$G - F$ to denote the graph \((V(G),(E(G)\setminus{}F))\). A
\emph{triangle} is a cycle of length three. Set \(S\) is a
\emph{feedback vertex set} (FVS) of \(G\) if \(G-S\) is a forest.
A path \(P\) (or cycle \(C\)) \emph{passes through} \(S\) if \(P\)
(or \(C\)) contains a vertex from \(S\). Let \(T\subseteq{}V(G)\)
be a specified set of vertices called \emph{terminal} vertices (or
\emph{terminals}). A \emph{\(T\)-cycle} (\(T\)-triangle) is a
cycle (triangle) which passes through \(T\). Graph \(G\) is a
\emph{\(T\)-forest} if it contains no \(T\)-cycle. Vertex set
\(S\) is a \emph{subset feedback vertex set} (subset-FVS) of \(G\)
with respect to terminal set \(T\) if the graph \(G-S\) is a
\(T\)-forest. Note that \(S\) may contain vertices from \(T\), and
that \(G-S\) need not be a forest. Set \(S\) is a \emph{subset
  triangle hitting set} (subset-THS) of \(G\) with respect to
terminal set \(T\) if \(G-S\) contains no \(T\)-triangle. More
generally, we say that a vertex \(v\) \emph{hits} a cycle \(C\) if
\(C\) contains \(v\). Vertex set \(S\) \emph{hits} a set
\(\mathcal{C}\) of cycles if for each cycle \(C\in\mathcal{C}\)
there is a vertex \(v\in{}S\) which hits \(C\).  We elide the
phrase ``with respect to \(T\)'' when there is no ambiguity.

\(K_{n}\) is the complete graph on \(n\) vertices. A subset
\(S\subseteq{}V(G)\) of vertices of graph \(G\) is a \emph{clique}
if its vertices are all pairwise adjacent, and is an
\emph{independent set} if they are all pairwise non-adjacent. A
clique \(C\) in \(G\) is a \emph{maximal clique} if \(C\) is not a
\emph{proper} subset of some clique in \(G\). A vertex \(v\) of
\(G\) is a \emph{simplicial vertex} (or is simplicial) in \(G\) if
\(N[v]\) is a clique.  In this case we say that \(N[v]\) is a
\emph{simplicial clique} in \(G\) and that \(v\) is a simplicial
vertex \emph{of} \(N[v]\).
\begin{fact}[\cite{blair1993introduction},
  Lemma~3]\label{fac:simplicial-vertex-characterization}
  Vertex \(v\) is simplicial in graph \(G\) if and only \(v\)
  belongs to precisely one maximal clique of \(G\), namely the set
  \(N[v]\).
\end{fact}

\subsection{Chordal Graphs and Clique Trees} A graph \(G\) is
\emph{chordal} (or \emph{triangulated}) if every induced cycle in
\(G\) is a triangle; equivalently, if every cycle of length at
least four has a chord.  If \(G\) is a chordal graph
then~\cite{golumbic2004algorithmic}: (i) every induced subgraph of
\(G\) is chordal; (ii) \(G\) has a simplicial vertex, and if \(G\)
is not a complete graph then \(G\) has two non-adjacent simplicial
vertices. Whether a graph \(H\) is chordal or not can be found in
time \(\Oh(|V(H)|+|E(H)|)\), and if \(H\) is chordal then a
simplicial vertex of \(H\) can be found within the same time
bound~\cite{golumbic2004algorithmic}.

The number of maximal cliques in a chordal graph $G$ is at most
$|V(G)|$ and they can all be enumerated in time
$\Oh(|V(G)|+|E(G)|)$ (\cite[Theorem
4.17]{golumbic2004algorithmic}). Let \(\mathcal{C}(G)\) be the set
of maximal cliques of a chordal graph \(G\). A \emph{clique
  tree}\footnote{We refer readers to monograph of J. Blair and
  B. Peyton for an introduction to clique trees
  \cite{blair1993introduction}.} of \(G\) is a graph
\(\mathcal{T}_{G}\) with the following properties:
\begin{enumerate}
\item The vertex set of \(\mathcal{T}_{G}\) is the set
  \(\mathcal{C}(G)\).
\item \(\mathcal{T}_{G}\) is a tree.
\item (Clique Intersection Property) Let \(C_{1},C_{2}\) be any two maximal cliques of \(G\), and
  let \(C'=C_{1}\cap{}C_{2}\). If \(C\) is a maximal clique of
  \(G\) which lies on the path from \(C_{1}\) to \(C_{2}\) on
  \(\mathcal{T}_{G}\), then it is the case that
  \(C'\subseteq{}C\).
\end{enumerate}

\begin{fact}\label{fac:clique-tree-properties}
  ~
  \begin{enumerate}
  \item A connected graph \(G\) is chordal if and only if it has a clique
    tree~\cite[Theorem~3.1]{blair1993introduction}.
  \item\label{fac:clique-tree-in-polytime} A clique tree
    \(\mathcal{T}_{G}\) of a chordal graph \(G\) can be computed
    in \(\Oh(|V(G)|+|E(G)|)\)
    time~\cite[Theorem~12]{galinier1995chordal}.
  \item\label{fac:subtree} Let \(G\) be a connected chordal graph
    and \(\mathcal{T}_{G}\) a clique tree of \(G\).  For each
    vertex \(v\) of \(G\), the set of all nodes of
    \(\mathcal{T}_{G}\) which contain \(v\) form a connected
    subgraph (a subtree) of
    \(\mathcal{T}_{G}\)~\cite[Theorem~3.2]{blair1993introduction}.
  \end{enumerate}
\end{fact}

We also need
\begin{observation}\label{obs:leaves-are-simplicial}
  Let \(G\) be a connected graph with at least two vertices and
  let \(\mathcal{T}_{G}\) be a clique tree of \(G\). If \(C\) is a
  leaf node of \(\mathcal{T}_{G}\) then \(C\) is a simplicial
  clique in \(G\).
\end{observation}
\begin{proof}
  Let \(C\) be a leaf node of \(\mathcal{T}_{G}\) and let \(C'\)
  be the unique neighbour of \(C\) in \(\mathcal{T}_{G}\). Since
  \(C\) is a maximal clique we have that
  \(C\setminus{}C'\neq\emptyset\). Pick a vertex
  \(v\in{}C\setminus{}C'\). Since \(C\) is a clique we have that
  \(C\subseteq{}N[v]\).

  If there is a vertex \(u\in{}(N[v]\setminus{}C)\) then let
  \(C''\) be a maximal clique which contains the set
  \(\{u,v\}\). Then \(C''\) is distinct from \(C\) and \(C'\). Now
  (i) \(v\in(C\cap{}C'')\), (ii) \(v\notin{}C'\), and (iii) \(C'\)
  is a maximal clique which lies on the path from \(C\) to \(C''\)
  in \(\mathcal{T}_{G}\). This contradicts the Clique Intersection
  Property of \(\mathcal{T}_{G}\). Hence we get that
  \(N[v]\subseteq{}C\) as well. Thus \(C=N[v]\) is a simplicial
  clique in \(G\).
\end{proof}

\subsection{Split Graphs} A graph \(G\) is a \emph{split graph}
if its vertex set can be partitioned into a clique and an independent set in \(G\). Such a partition is called a \emph{split partition} of \(G\). We use \((K,I)\) to denote the split partitions of the graphs.
We refer to vertex sets \(K\) and \(I\) as the \emph{clique side} and \emph{independent side}, respectively, of graph \(G\).
We say that an edge \(uv\) in \(G[K]\) is \emph{highlighted} if there is a vertex \(x\) in \(I\) such that the vertices \(\{x,u,v\}\) induce a triangle in \(G\).
Every split graph is a chordal graph as
well~\cite{hammer1977split}. Given a graph \(G\) as its adjacency
list we can (i) check if \(G\) is a split graph, and if it is,
(ii) find a partition \(V(G)=K\uplus{}I\) into a clique \(K\) and
independent set \(I\), both in time
\(\Oh(|V(G)|+|E(G)|)\)~\cite{HammerSimeone1981}.

If a chordal graph \(G\) contains a cycle then it contains an
induced cycle, and every induced cycle in \(G\) is---by
definition---a triangle. Conversely, if \(G\) contains a triangle
then it trivially contains a cycle.  Thus a chordal graph contains
a cycle if and only if it contains a triangle. This carries over
to subset feedback vertex sets in chordal (and therefore split)
graphs in a natural way.

\begin{lemma}\label{lem:triangles}
  Let \(G\) be a chordal graph and let \(T\subseteq{}V(G)\) be a
  specified set of terminal vertices. A vertex subset
  \(S\subseteq{}V(G)\) is a subset-FVS of \(G\) with respect to
  \(T\) if and only if the graph \(G-S\) contains no
  \(T\)-triangles.
\end{lemma}
\begin{proof}
  If \(S\) is a subset-FVS of \(G\) with respect to \(T\)
  then---by definition---the graph \(G-S\) contains no
  \(T\)-triangles. For the reverse direction, let \(S\) be a
  subset-THS of \(G\) with respect to \(T\), and let
  \(H=G-S\). Since \(H\) is an induced subgraph of \(G\) we get
  that \(H\) is chordal. Assume for the sake of contradiction that
  there is a \(T\)-cycle in \(H\). Let \(C\) be a \(T\)-cycle in
  \(H\) of the smallest length \(\ell\). Then \(\ell\geq{}4\) and
  we get that \(C\) has a chord. Now this chord is part of two
  cycles \(C',C''\) such that (i) each of \(\{C',C''\}\) has
  length strictly smaller than \(C\), and (ii) the union of the
  vertex sets of \(C'\) and \(C''\) is the vertex set of
  \(C\). Thus at least one of \(\{C',C''\}\), say \(C'\), is a
  \(T\)-cycle \(H\) of length strictly smaller than \(\ell\), a
  contradiction.
\end{proof}

We use \((G;T;k)\) to denote an instance of \SFVSC or \SFVSS where \(G\) is the input graph, \(T\) is the specified set of terminals, and \(k\) is the parameter. 

\begin{corollary}\label{cor:hit-only-triangles}
  An instance \((G;T;k)\) of \SFVSC (or of \SFVSS) is a \YES
  instance if and only if there is a vertex subset
  \(S\subseteq{}V(G)\) of size at most \(k\) such that \(S\) is a
  \(T\)-THS of \(G\).
\end{corollary}

\begin{lemma}\label{lem:bridge-deletion}
  If \(G\) is a chordal graph and \(uv\) is a bridge in \(G\) then
  the graph \(G-\{uv\}\) is also chordal. If \(H\) is a
  \emph{split} graph with at least three vertices on the clique
  side and \(uv\) is an edge with exactly one end in the clique
  side of \(H\) then the graph \(H-\{uv\}\) is also split. If
  \(uv\) is a \emph{bridge} in such a split graph \(H\) then the
  graph \(H-\{uv\}\) is also split.
\end{lemma}
\begin{proof}
  Let \(G\) be a chordal graph. If an edge \(xy\) is a chord of
  some cycle \(C\) then \(xy\) lies in a cycle \(C'\) whose
  vertex set is a strict subset of the vertex set of \(C\). So we
  get, since \(uv\) is a bridge in \(G\), that \(uv\) is not a
  chord of any cycle in \(G\).

  Suppose \(G-\{uv\}\) is not chordal. Then it contains an induced
  cycle \(C\) of length four or more. Since edge \(uv\) is not
  present in graph \(G-\{uv\}\) it is not present in cycle \(C\)
  either. Thus cycle \(C\) is present in graph \(G\). Since \(uv\)
  is not a chord of cycle \(C\) we get that \(C\) is an induced
  cycle of length at least four in \(G\), which contradicts the
  chordality of \(G\). 
  
  Now let \(H\) be a split graph with at least three vertices on
  its clique side, and let \((K,I)\) be a split partition of
  \(H\). Since every edge in \(H[K]\) is part of a cycle, any
  bridge \(uv\) of \(H\) has one end in each of \(K, I\). The sets
  \(K\) and \(I\) remain a clique and an independent set,
  respectively, once we delete any such edge \(uv\). Hence
  \(H-\{uv\}\) is a split graph.
\end{proof}

Since chordal graphs and split graphs are (respectively) closed
under taking induced subgraphs~\cite{golumbic2004algorithmic} we
get

\begin{corollary}\label{cor:bridge-vertex-deletion-ok}
  Let \(H\) be a graph obtained from a graph \(G\) by repeatedly
  deleting vertices and/or bridges of the remaining graph at each
  stage.
  \begin{enumerate}
  \item If \(G\) is a chordal graph then so is \(H\).
  \item If \(G\) is a split graph and at least three of the clique
    vertices of \(G\) remain in \(H\) then \(H\) is a split graph.
  \end{enumerate}
\end{corollary}

Let $(G; T; k)$ be an instance of \SFVSC. A subset
\(S\subseteq{}V(G)\) of vertices of \(G\) is a \emph{solution} of
this instance if \(S\) is a \(T\)-FVS (equivalently, a \(T\)-THS)
of \(G\).

\subsection{Parameterized Algorithms and Kernelization}

We give a quick overview of the main concepts from parameterized
complexity relevant to our work; we refer interested readers to
\cite{CFKLMPPS15} for a detailed exposition on the subject. An
instance of a \emph{parameterized problem} is of the form
\((I;k)\) where \(I\) is an instance of a (classical) decision
problem---whose answer is one of \YES/\NO---and \(k\in\mathbb{N}\)
is the \emph{parameter}; \((I;k)\) is defined to be a \YES
(respectively, \NO) instance if \(I\) is a \YES (respectively,
\NO) instance. A parameterized problem is said to be \emph{fixed
  parameter tractable} (\FPT) if there exists an algorithm
$\mathcal{A}$, a computable function $f$, and a constant $c$ such
that, given any instance $(I; k)$ of the parameterized problem,
the algorithm $\mathcal{A}$ correctly decides whether $(I; k)$ is
an \yes\ instance or not in time $\Oh(f(k)|I|^c)$. The algorithm
$\mathcal{A}$ is an \FPT\ algorithm; if $c=1$ holds then
\(\mathcal{A}\) is a \emph{linear} \FPT\ algorithm. Instances
\((I;k)\) and \((I';k')\) of a parameterized problem are
\emph{equivalent} if \((I;k)\) is a \YES instance if and only if
\((I';k')\) is a \YES instance.  A \emph{kernel} for a
parameterized problem is an algorithm $\mathcal{B}$ that, given an
instance $(I,k)$ of the problem, works in time $\Oh((|I| + k)^c)$
and returns an equivalent instance $(I', k')$ of the same
problem. We require that $k'$ is upper bounded by some computable
function of $k$. If there exists a computable function $g$ such
that size of an output obtained by algorithm $\mathcal{B}$ for
$(I, k)$ is at most $g(k)$, we say that problem admits a kernel of
size $g(k)$.

Let \(H\) be a triangle on the vertex set \(\{x,y,z\}\). Then
\(I_{\YES}=(H;\{x\};1)\) and \(I_{\NO}=(H;\{x\};0)\) are
constant-size trivial \YES and \NO instances, respectively, of
both \SFVSS and \SFVSC.
Our algorithms make use of \emph{reduction rules} which transform
one instance of a problem to another instance of the same problem.
We use \((G;T;k)\) to represent the instance given as input to
each reduction rule, and \((G';T';k')\) to represent the
(modified) instance output by the rule. We say that a reduction
rule is \emph{safe} if for every input instance \((G;T;k)\) the
rule outputs an \emph{equivalent} instance \((G';T';k')\).  We
update \(G\gets{}G',T\gets{}T',k\gets{}k'\) to get the input
instance \((G;T;k)\) for further processing.  We say that an
instance \((G;T;k)\) is \emph{reduced} with respect to a reduction
rule \texttt{RR} if none of the conditions of rule \texttt{RR}
apply to \((G;T;k)\). Equivalently: instance \((G;T;k)\) is
\emph{reduced} with respect to reduction rule \texttt{RR} if, when
given the instance \((G;T;k)\) as input, rule \texttt{RR} produces
as output an instance \((G';T';k')\) which is identical to
\((G;T;k)\).  We say that a reduction rule \texttt{RR}
\emph{applies} to an instance \((G;T;k)\) if \((G;T;k)\) is
\emph{not} reduced with respect to \texttt{RR}.

\subsection{Expansion Lemmas}

Let \(t\) be a positive integer and $G$ a bipartite graph with
vertex bipartition $(P, Q)$. A set of edges $M \subseteq E(G)$ is
called a \emph{$t$-expansion of $P$ into $Q$} if (i) every vertex
of $P$ is incident with exactly $t$ edges of $M$, and (ii) the
number of vertices in \(Q\) which are incident with at least one
edge in $M$ is exactly $t|P|$. We say that \(M\) \emph{saturates}
the endvertices of its edges. Note that the set \(Q\) may contain
vertices which are \emph{not} saturated by \(M\). We need the
following generalizations of Hall's Matching Theorem known as
\emph{expansion lemmas}:

\begin{lemma}[\cite{CFKLMPPS15}
  Lemma~2.18] \label{lem:expansion-lemma} Let $t$ be a positive
  integer and $G$ be a bipartite graph with vertex bipartition
  $(P,Q)$ such that $|Q| \geq t |P|$ and there are no isolated
  vertices in $Q$.  Then there exist nonempty vertex sets
  $X \subseteq P$ and $Y \subseteq Q$ such that (i) $X$ has a
  $t$-expansion into $Y$, and (ii) no vertex in $Y$ has a
  neighbour outside $X$. Furthermore two such sets $X$ and $Y$ can
  be found in time polynomial in the size of $G$.
\end{lemma}

\begin{lemma}[\cite{fomin2016hitting}] \label{lem:expansion-matching}
  Let $t$ be a positive integer and $G$ be a bipartite graph with
  vertex bipartition $(P,Q)$ such that $|Q| > \ell t$ , where
  $\ell$ is the size of a maximum matching in $G$, and there are
  no isolated vertices in $Q$. Then there exist nonempty vertex
  sets $X \subseteq P$ and $Y \subseteq Q$ such that (i) $X$ has a
  $t$-expansion into $Y$, and (ii) no vertex in $Y$ has a
  neighbour outside $X$. Furthermore two such sets $X$ and $Y$ can
  be found in time polynomial in the size of $G$.
\end{lemma}

We need sets $X, Y$ of \autoref{lem:expansion-matching} with an
additional property:

\begin{lemma}\label{lem:extra-expansion-vertex} If the premises of
  \autoref{lem:expansion-matching} are satisfied then we can find,
  in polynomial time, sets $X, Y$ of the kind described in
  \autoref{lem:expansion-matching} and a vertex $w \in Y $ such
  that there exists a \(t\)-expansion \(M\) from \(X\) into \(Y\)
  which does not saturate \(w\).
\end{lemma}
\begin{proof}
  Let $X, Y$ be sets of the kind guaranteed to exist by
  \autoref{lem:expansion-matching}, and let \(M\) be a
  \(t\)-expansion from \(X\) into \(Y\). If $|Y| > t|X|$ then
  there must exist a vertex which is not saturated by \(M\). We
  can find such a vertex \(w\in{}Y\) by, for instance, deleting
  each vertex \(y\in{}Y\) in turn and testing if the resulting
  graph has \(t\)-expansion from \(X\) into
  \(Y\setminus{}y\).
  Thus it is enough to show that we can find, in polynomial time,
  such a pair \(X,Y\) for which $|Y| > t|X|$ holds. We give a
  proof by algorithm. We start by setting \(X=Y=\emptyset\). It
  holds vacuously that (i) there is a \(t\)-expansion from \(X\)
  into \(Y\), and (ii) no vertex in \(Y\) has a neighbour outside
  \(X\), and trivially that \(|Y|=t|X|\).
  \begin{enumerate}
  \item Find sets $X'\subseteq P $ and $Y'\subseteq Q $ as
    guaranteed to exist by \autoref{lem:expansion-matching}. Let
    \(M'\) be a \(t\)-expansion from \(X'\) into \(Y'\). If
    \(|Y'|>t|X'|\) then return
    \(((X\cup{}X)',(Y\cup{}Y'))\). Otherwise, if there is a vertex
    \(w\in{}Q\setminus{}Y'\) which has no neighbour in
    \(P\setminus{}X'\) then return
    \(((X\cup{}X'),(Y\cup{}Y'\cup\{w\}))\).
  \item At this point we have \(|X'|<|P|\) and
    \(|Y'|=t|X'|\). From above we get that there is
    \(t\)-expansion, say \(M\), from \(X\) into \(Y\).  Since
    \(X\cap{}X'=\emptyset=Y\cap{}Y'\) we get that \(M\cup{}M'\) is
    a \(t\)-expansion from \(X\cup{}X'\) into \(Y\cup{}Y'\). Set
    $\hat{X}\gets{}X\cup X', \hat{Y}\gets{}Y\cup Y'$. Then (i)
    there is a \(t\)-expansion from \(\hat{X}\) into \(\hat{Y}\),
    (ii) no vertex in \(\hat{Y}\) has a neighbour outside
    \(\hat{X}\), and (iii) \(|\hat{Y}|=t|\hat{X}|\).

  \item Let \(\hat{P}=(P-\hat{X}),\hat{Q}=(Q-\hat{Y})\). Consider the
    subgraph \(\hat{G}=G[\hat{P}\cup{}\hat{Q}]\) and its vertex
    bipartition \((\hat{P},\hat{Q})\). Since \(t\geq{}1\) and the
    vertices in \(\hat{X}\cup{}\hat{Y}\) are saturated by a
    \(t\)-expansion from \(\hat{X}\) into \(\hat{Y}\) we get that
    the subgraph \(G[\hat{X}\cup{}\hat{Y}]\) contains a
    \emph{matching} of size \(|\hat{X}|\). Since the subgraph
    \(\hat{G}\) of \(G\) contains none of the vertices saturated
    by this matching we get that the size \(\hat{\ell}\) of a
    maximum matching in \(\hat{G}\) satisfies
    \(\hat{\ell}\leq{}\ell-|\hat{X}|\).

    Since every vertex in \(\hat{Q}\) has at least one neighbour
    in \(\hat{P}\) (otherwise we would have returned in step (1))
    we get that there are no isolated vertices in the set
    \(\hat{Q}\) in graph \(\hat{G}\). Since
    $|\hat{Y}| = t|\hat{X}|$ and \(|Q|>\ell{}t\) we have that
    \(|\hat{Q}|=|Q|-t|\hat{X}|>\ell{}t-t|\hat{X}|=t(\ell{}-|\hat{X}|)\geq{}t\hat{\ell}\). Thus
    graph \(\hat{G}\) and its vertex bipartition
    \((\hat{P},\hat{Q})\) satisfy the premises of
    \autoref{lem:expansion-matching}. Set
    \(G\gets\hat{G},P\gets\hat{P},Q\gets\hat{Q},X\gets\hat{X},Y\gets\hat{Y}\)
    and go to step (1).
  \end{enumerate}

  \noindent\textbf{Correctness.} Note that before step (1) is
  executed it is always the case that (i) there is a
  \(t\)-expansion from \(X\) into \(Y\), (ii) no vertex in \(Y\)
  has a neighbour outside \(X\), and (iii) \(|Y|=t|X|\). So we get
  that \emph{if} the algorithm terminates (which it does only at
  step (1)) it returns a correct pair of vertex subsets.

  The graph \(G\) from the premise of
  \autoref{lem:expansion-matching} has a vertex bipartition
  \((P,Q)\) with \((|P|>0,|Q|>0,|Q|>\ell{}t)\), and the sets
  \(\hat{X},\hat{Y}\) in steps (2) and (3) satisfy
  \(0<|\hat{X}|<|P|\) and \(|\hat{Y}|=t|\hat{X}|\).  So the sets
  \(\hat{P},\hat{Q}\) of step (3) satisfy
  \(|\hat{P}|>0,|\hat{Q}|>0,|\hat{Q}|>\hat{\ell{}}t\). Thus the graph
  \(\hat{G}\) computed in step (3) has strictly fewer vertices
  than the graph \(G\) passed in to the previous step (1). Since
  we update \(G\gets\hat{G}\) before looping back to step (1), we
  get that the algorithm terminates in polnomially many steps.
\end{proof}

%% file: kernel-split.tex
\section{Kernel Bounds for {\SubsetFVSSplit}}
\label{sec:kernelbound}

In this section we show that \SFVSS has a quadratic-size kernel
with a linear number of vertices on the clique side. % We show also
% that the kernel size is tight unless \caveat. \todo{We don't, in
%   this section.}

\splitkernel*

% \begin{theorem}\label{thm:split-kernel}
%   \SFVSS has a quadratic-size kernel. More precisely: There is a
%   polynomial-time 
%   algorithm which, given an instance $(G; T; k) $ of \SFVSS,
%   returns an instance \((G';T';k')\) of \SFVSS such that (i)
%   \((G;T;k)\) is \YES instance if and only if \((G';T';k')\) is
%   \YES instance, and (ii) \(|V(G')|=\Oh(k^{2})\),
%   \(|E(G')|=\Oh(k^{2})\), and \(k'\leq{}k\). Moreover if \((K',I')\) is the split partition
%   of split graph \(G'\) then \(|K'|\leq{}10k\).
% \end{theorem}

Our algorithm works as follows. We first reduce the input to an
instance \((G; T; k)\) where the terminal set \(T\) is exactly the
independent set \(I\) from a split partition \((K,I)\) of
\(G\). Then we show that if a (non-terminal) vertex \(v\in{}K\)
has at least \(k+1\) neighbours in \(I\) then we can include \(v\)
in a solution or safely delete one edge incident with \(v\); this
leads to an instance where each \(v\in{}K\) has at most \(k\)
neighbours in \(I\). We apply the expansion lemma
(\autoref{lem:expansion-matching}) to this instance to bound the
number of vertices in \(K\) by \(10k\); this gives the bound of
\(\Oh(k^{2})\) on the number of vertices in \(I\).

We now describe the reduction rules. Recall that we use
\((G;T;k)\) and \((G';T';k')\) to represent the input and output
instances of a reduction rule, respectively.  We always apply the
\emph{first} rule---in the order in which they are described
below---which applies to an instance.  Thus we apply a rule to an
instance \emph{only if} the instance is reduced with respect to
all previously specified reduction rules.

Recall that a split graph may have more than one split
partition. To keep our presentation short we need to be able to
refer to one split partition which ``survives'' throughout the
application of these rules. Towards this we fix an arbitrary split
partition \((K^{\star},I^{\star})\) of the original input
graph. Whenever we say ``the split partition \((K,I)\) of graph
\(G\)'' we mean the ordered pair
\(((K^{\star}\cap{}V(G)),(I^{\star}\cap{}V(G)))\). The only ways
in which our reduction rules modify the graph are: (i) delete a
vertex, or (ii) delete an edge of the form
\(uv\;;\;u\in{}K^{\star},v\in{}I^{\star}\). So
\(((K^{\star}\cap{}V(G)),(I^{\star}\cap{}V(G)))\) remains a split
partition of the ``current'' graph \(G\) at each stage during the
algorithm.

Our first reduction rule deals with some easy instances.
%\todo{We
%  refer to the sub-parts of this rule in our proofs, so if we
%  change the numbering of these sub-parts then we should change
%  those references as well.}
\begin{reductionrule}\label{rr:split-kernel-yes-no}
  Recall that \((K,I)\) is the split partition of graph
  \(G\). Apply the first condition which matches \((G;T;k)\):
  \begin{enumerate}
  \item If \(T=\emptyset\) then output \(I_{\YES}\) and stop.
  \item If \(k<0\), or if \(k=0\) \emph{and} there is a
    \(T\)-triangle in \(G\), then output \(I_{\NO}\) and stop.
  \item If there is no \(T\)-triangle in \(G\) then output
    \(I_{\YES}\) and stop.
  \item If \(|K|\leq{}k+1\) then output \(I_{\YES}\) and stop.
  \item If \(|K|=k+2\) and there is an edge \(uv\) in \(G[K]\)
    which is \emph{not} highlighted then output \(I_{\YES}\) and
    stop.
  \end{enumerate}
\end{reductionrule}

\begin{observation}\label{obs:clique-side-at-least-3}
  If \((G;T;k)\) is reduced with respect to
  \autoref{rr:split-kernel-yes-no} then the clique side \(K\) of
  \(G\) has size at least three.
\end{observation}
\begin{proof}
  The second and third parts of the rule ensure that
  \(k\geq{}1\). The fourth part now implies \(|K|\geq{}3\).
\end{proof}

\begin{lemma}\label{lem:split-kernel-yes-no-rr-safe}
  \autoref{rr:split-kernel-yes-no} is safe.
\end{lemma}
\begin{proof}
  We analyze each part separately.
  \begin{enumerate}
  \item If \(T=\emptyset\) then there are no \(T\)-cycles in
    \(G\). So \((G;T;k)\) is (vacuously) a \YES instance, as is
    \(I_{\YES}\).
  \item If \(k<0\) then---since there does not exist a vertex
    subset \(S\) of negative size---\((G;T;k)\) is a \NO
    instance. If the second part of this condition holds then
    \((G;T;k)\) is clearly a \NO instance. Thus in both cases
    \((G;T;k)\) is a \NO instance, as is \(I_{\NO}\).
  \item This condition applies to \((G;T;k)\) only if the previous
    one does \emph{not} apply. Thus \(k\geq{}0\) and so
    \((G;T;k)\) is a \YES instance, as is \(I_{\YES}\).
  \item Let \(S\subseteq{}K\) be an arbitrary \(k\)-sized set of
    vertices on the clique side. Deleting \(S\) from \(G\) gives a
    graph \(H\) with at most one vertex on the clique side. Since
    there are no edges among vertices on the independent side in
    \(H\) we get that \(H\) has no triangles; \(S\) is a solution
    of \((G;T;k)\) of size at most \(k\). Thus \((G;T;k)\) is a
    \YES instance, as is \(I_{\YES}\).
  \item Let \(S=K\setminus{}\{u,v\}\). Then \(|S|=k\). Since
    vertices \(\{u,v\}\) have no common neighbour on the
    independent side of \(G\), we get that \(G-S\) contains no
    triangles. Thus \((G;T;k)\) is a \YES instance, as is
    \(I_{\YES}\).%\qedhere
  \end{enumerate}
\end{proof}

%\todo{Delete Observation~\ref{obs:bridge-deletion} and add the following after the paragraph in Prelims on split graphs.}

Each remaining rule deletes a vertex or an edge from the graph. We
use the next two observations in our proofs of safeness.

\begin{observation}\label{obs:deleting-vertices-bridges-fwd-safe}
  Let \((G;T;k)\) be a \YES instance of \SFVSS which is reduced
  with respect to \autoref{rr:split-kernel-yes-no}.
  \begin{enumerate}
  \item Let \(G'\) be a graph obtained from \(G\) by deleting a
    vertex \(v\in{}V(G)\), and let \(T'=T\setminus\{v\}\). Then
    \((G';T';k'=k)\) is a \YES instance. If \((G;T;k)\) has a
    solution \(S\) of size at most \(k\) with \(v\in{}S\) then
    \((G';T';k'=k-1)\) is a \YES instance.
  \item Let \(G'\) be a graph obtained from \(G\) by deleting an
    edge which has exactly one of its endvertices in the clique
    side of \(G\). Then \((G';T'=T;k'=k)\) is a \YES instance.
  \end{enumerate}
\end{observation}
\begin{proof}
  First we consider the case \(G'=G-\{v\}\). From
  \autoref{cor:bridge-vertex-deletion-ok} we get that graph
  \(G'\) is a split graph. Since
  \(T'=T\setminus\{v\}\subseteq{}V(G')\) we get that both
  \((G';T';k)\) and \((G';T';k-1)\) are instances of \SFVSS.

  Suppose \((G;T;k)\) is a \YES instance, and let \(S\) be a
  solution of \((G;T;k)\) of size at most \(k\). Then the graph
  \(G-S\) is a split graph with no \(T\)-triangles.  We consider
  two cases:
  \begin{enumerate}
  \item If \(v\in{}S\) then
    \(G'-(S\setminus{}\{v\})=(G-\{v\})-(S\setminus{}\{v\})=G-S\). Hence
    \(S\setminus\{v\}\) is a \(T\)-THS of the split graph \(G'\),
    of size at most \(|S|-1=k-1\). Thus
    \((G';T'=T\setminus\{v\};k'=k-1)\) is a \YES instance, and so
    is \((G';T'=T\setminus\{v\};k'=k)\) .
  \item If \(v\notin{}S\) then the graph
    \(G'-S=(G-\{v\})-S=G-(S\cup{}\{v\})\) is obtained by deleting
    vertex \(v\) from \(G-S\). Thus \(G'-S\) is a split
    graph. Since deleting a vertex cannot create a new
    \(T\)-triangle, we get that \(G'-S\) has no
    \(T\)-triangles. Thus \(S\) is an \(T\)-THS of \(G'\) of size
    at most \(k\), and \((G';T';k'=k)\) is a \YES instance.
  \end{enumerate}
  Now we consider the case \(G'=G-\{xy\}\) where edge \(xy\) has
  one end, say \(x\), in the clique side \(K\) of \(G\) and the
  other end \(y\) in the independent side \(I\).  From
  \autoref{obs:clique-side-at-least-3} we get
  \(|K|\geq{}3\), and then from \autoref{lem:bridge-deletion} we
  get that \(G'\) is a split graph. Once again, since
  \(T'=T\setminus\{v\}\subseteq{}V(G')\) we get that both
  \((G';T';k)\) and \((G';T';k-1)\) are instances of \SFVSS.

  Suppose \((G;T;k)\) is a \YES instance, and let
  \(S\subseteq{}V(G)\) be a solution of \((G;T;k)\) of size at
  most \(k\). Then graph \(G-S\) has no \(T\)-cycle. Since
  deleting an edge cannot introduce a new cycle, we get that the
  graph \(G'-S=(G-S)-\{e\}\) has no \(T\)-cycle either. Thus \(S\)
  is a solution of \((G';T;k)\) of size at most \(k\), and
  \((G';T;k)\) is a \YES instance.
\end{proof}

\begin{observation}\label{obs:adding-solution-vertex-safe}
  Let \((G';T';k')\) be a \YES instance of \SFVSS and let \(S'\)
  be a \(T'\)-THS of \(G'\) of size at most \(k'\). Let \(G\) be a
  \emph{split} graph which can be constructed from graph \(G'\) by
  adding a vertex \(v\) and zero or more edges each incident with
  the new vertex \(v\). Then both \((G;T_{1}=T';k=k'+1)\) and
  \((G;T_{2}=T'\cup\{v\};k=k'+1)\) are \YES instances of \SFVSS,
  and the set \(S'\cup\{v\}\) is a solution of size at most \(k\)
  for both these instances.
\end{observation}
\begin{proof}
  Since \(G\) is a split graph both \((G;T_{1};k)\) and
  \((G;T_{2};k)\) are instances of \SFVSS. Since \(S'\) is a
  \(T'\)-FVS of \(G'\) of size at most \(k'\) we have that graph
  \(G'-S'\) has no \(T'\)-triangle. Let \(S=S'\cup\{v\}\). Then
  \(|S|=|S'|+1\leq{}(k'+1)\) and \(G-S=G'-S'\). Thus graph \(G-S\)
  has no \(T'\)-triangle. Since \(v\notin{}V(G-S)\) we get that
  \(G-S\) has no triangle which contains vertex \(v\). Thus
  \(G-S\) has no \(T'\cup\{v\}\)-triangle either. Hence both
  \((G;T_{1};k)\) and \((G;T_{2};k)\) are \YES instances of
  \SFVSS and the set \(S'\cup\{v\}\) is a solution of size at most
  \(k\) for both these instances.
\end{proof}
\begin{reductionrule}\label{rr:delete-isolates}
  %\todo{This rule applies to chordal graphs as well.}
  If there is
  a vertex \(v\) of degree zero in \(G\) then delete \(v\) from
  \(G\) to get graph \(G'\).  Set
  \(T'\gets{}T\setminus\{v\},k'\gets{}k\). The reduced instance is
  \((G';T';k')\).
\end{reductionrule}

Since adding or deleting vertices of degree zero does not create
or destroy cycles of any kind, we have 
\begin{observation}\label{obs:delete-isolates-rr-safe}
  \autoref{rr:delete-isolates} is safe.
\end{observation}

%\todo{Fix references to the next two rules; I changed their labels.}
\begin{reductionrule}\label{rr:only-T-neighbours-split}
  %\todo{Reduction Rule~\ref{rr:only-T-neighbours} is essentially
  %  the same as this one. Make a call if we should delete one of
  %  them. If we do delete, fix the labels/references. Also fix the  safeness arguments to take care of both kinds of graphs.}
  If there is a \emph{non-terminal} vertex \(v\) in \(G\) which is
  \emph{not} adjacent to a terminal vertex, then delete \(v\) from
  \(G\) to get graph \(G'\).  Set \(T'\gets{}T,k'\gets{}k\). The
  reduced instance is \((G';T';k')\).
\end{reductionrule}

\iffalse
%%%% ATTENTION: This is commented out.
\begin{observation}\label{obs:split-instance-connected}
  If \((G;T;k)\) is reduced with respect to Reduction
  Rules~\ref{rr:split-kernel-yes-no}, \ref{rr:delete-isolates},
  and~\ref{rr:only-T-neighbours-split} then graph \(G\) is
  connected.
  %\todo{This does \emph{not} apply to chordal graphs. }
  %\todo{See if we need/use the assumption of connectedness for the chordal algorithm.}
\end{observation}
\begin{proof}
  Let \((K,I)\) be the split partition of \(G\). From
  Observation~\ref{obs:clique-side-at-least-3} we get
  \(|K|\geq{}3\). So if \(G\) is disconnected then there has to be
  a vertex \(v\) degree zero on its independent side \(I\). But
  this cannot be the case since \((G;T;k)\) is reduced with
  respect to Reduction Rule~\ref{rr:delete-isolates}.
\end{proof}
\fi

\begin{lemma}\label{lem:only-T-neighbours-rr-safe-split}
  \autoref{rr:only-T-neighbours-split} is safe.
\end{lemma}
\begin{proof}
  Let \((G;T;k)\) be an instance given as input to
  \autoref{rr:only-T-neighbours-split} and let \((G';T';k')\) be
  the corresponding instance output by the rule.  Then
  \(G'=G-\{v\}\) where vertex \(v\) is as defined by the rule, and
  \(T'=T,k'=k\). From
  \autoref{obs:deleting-vertices-bridges-fwd-safe} we get that if
  \((G;T;k)\) is a \YES instance then so is \((G';T';k')\).

  %\todo{Use \(T',k'\) even if the values didn't change, to avoid
  %  confusion.}
  % Since split graphs are closed under vertex
  % deletions we have that \(G'\) is a split graph. Thus
  % \((G';T;k)\) is a valid TODODODO instance of \SFVSS.

  % Now suppose \((G;T;k)\) is a \YES instance, and let
  % \(S\subseteq{}V(G)\) be a solution of \((G;T;k)\) of size at
  % most \(k\). Then graph \(G-S\) has no \(T\)-cycle. Since
  % deleting a vertex cannot introduce a new cycle, we get that the
  % graph \(G'-S=G-(S\cup\{v\})\) has no \(T\)-cycle either. Thus
  % \(S\) is a solution of \((G';T;k)\) of size at most \(k\), and
  % \((G';T;k)\) is a \YES instance.

  Now suppose \((G';T';k')\) is a \YES instance, and let
  \(S'\subseteq{}V(G')\) be a solution of \((G';T';k')\) of size
  at most \(k'\). We claim that \(S'\) is a solution of
  \((G;T;k)\) as well. Suppose not; then graph \(G-S'\) has a
  \(T\)-cycle. Since \(T'=T\) and \(G'-S'=G-(S'\cup\{v\})\) does
  not contain any \(T'\)-cycle we get that every \(T\)-cycle in
  \(G-S'\) must contain vertex \(v\).

  Let \(C\) be a \emph{shortest} \(T\)-cycle of \(G-S'\), and let
  \(t\) be a terminal vertex in \(C\). Since---by
  assumption---vertices \(v\) and \(t\) are not adjacent we get
  that \(C\) contains at least four vertices. Since \(C\) is a
  cycle of length at least four in the chordal graph \(G-S'\) we
  get that \(C\) has a chord. This chord is part of two cycles
  \(C',C''\) such that (i) each of \(\{C',C''\}\) has length
  strictly smaller than \(C\), and (ii) the union of the vertex
  sets of \(C'\) and \(C''\) is the vertex set of \(C\). At least
  one of the two cycles \(C',C''\) contains the terminal \(t\);
  assume that \(t\) is in \(C'\). If vertex \(v\) is also in
  \(C'\) then \(C'\) is a \(T\)-cycle of \(G-S'\) which is
  \emph{shorter} than \(C\); this contradicts our assumption about
  \(C\).  If \(v\) is \emph{not} in \(C'\) then \(C'\) is a
  \(T\)-cycle---and hence \(T'\)-cycle---in
  \(G-(S'\cup\{v\})=G'-S'\), which contradicts our assumption that
  \(S'\) is a solution of \((G';T';k')\).

  Thus \(S'\) is a solution of \((G;T;k)\) of size at most
  \(k'=k\), and \((G;T;k)\) is a \YES instance.
\end{proof}

%\todo{Remove from the chordal section?}
\begin{reductionrule}\label{rr:delete-bridges-split}
  If there is a \emph{bridge} \(e\) in \(G\) then delete edge
  \(e\) (\emph{not} its endvertices) to get graph \(G'\).  Set
  \(T'\gets{}T,k'\gets{}k\). The reduced instance is
  \((G';T';k')\).
\end{reductionrule}

\begin{lemma}\label{lem:no-bridges-rr-safe-split}
  \autoref{rr:delete-bridges-split} is safe.
\end{lemma}
\begin{proof}
  Let \((G;T;k)\) be an instance given as input to
  \autoref{rr:delete-bridges-split} and let \((G';T';k')\) be the
  corresponding instance output by the rule.  Then \(G'=G-\{e\}\)
  where \(e\) is a bridge in \(G\), \(T'=T,k'=k\).

  From \autoref{obs:deleting-vertices-bridges-fwd-safe} we get
  that if \((G;T;k)\) is a \YES instance then so is
  \((G';T';k')\).
  
  Now suppose \((G';T';k')\) is a \YES instance, and let
  \(S'\subseteq{}V(G')\) be a solution of \((G';T';k')\) of size
  at most \(k'\).  Observe that since (i) \(e\) is a bridge in
  \(G\), and (ii) deleting vertices does not introduce new cycles,
  edge \(e\), if it exists in graph \(G-S'\), is a bridge in
  \(G-S'\) as well. So \(e\) cannot be in \emph{any} cycle in
  \(G-S'\). Hence if graph \(G-S'\) has a \(T\)-cycle \(C\) then
  \(C\) does not contain edge \(e\), which implies that \(C\) is
  present in the graph \(G'-S'=(G-S')-\{e\}\) as well. But this
  contradicts our assumption that \(S'\) is a solution of
  \((G';T';k')\). Thus there cannot be a \(T\)-cycle in \(G-S'\).
  So \(S'\) is a solution of \((G;T;k)\) of size at most \(k'=k\),
  and \((G;T;k)\) is a \YES instance.
\end{proof}

\begin{lemma}\label{lem:neat-split-instance}
  Let \((G;T;k)\) be an instance of \SFVSS which is reduced with
  respect to
  \Autoref{rr:split-kernel-yes-no,rr:delete-isolates,rr:only-T-neighbours-split,rr:delete-bridges-split}.
  Then
   \begin{enumerate}
   \item Each vertex in \(G\) has degree at least two.
   \item Every vertex in \(G\) is part of some \(T\)-triangle.
   \item If \((G;T;k)\) is a \YES instance then every terminal
     vertex on the clique side of \(G\) is present in \emph{every}
     solution of \((G;T;k)\) of size at most \(k\).
   \end{enumerate}
\end{lemma}
\begin{proof}
  We prove each claim in turn. Let \((K,I)\) be the split
  partition of \(G\).
  \begin{enumerate}[listparindent=1.5em]
  \item Since \((G;T;k)\) is reduced with respect to
    \autoref{rr:delete-isolates} we get that every vertex in \(G\)
    has degree at least one. If vertex \(v\) has degree exactly
    one then the only edge incident on \(v\) is a bridge, which
    cannot exist since \((G;T;k)\) is reduced with respect to
    \autoref{rr:delete-bridges-split}. Thus every
    vertex in \(G\) has degree at least two.
  \item From \autoref{obs:clique-side-at-least-3} we get
    \(|K|\geq{}3\). Consider a vertex \(v\in{}K\). Since
    \((G;T;k)\) is reduced with respect to
    \autoref{rr:only-T-neighbours-split} we get that \(v\) is
    adjacent to at least one terminal vertex \(t\). If \(t\in{}K\)
    then \(v\) is part of a \(T\)-triangle which contains
    \(t\). If \(t\in{}I\) then let \(u\in{}K\;;\;u\neq{}v\) be
    another neighbour of \(t\) in \(K\). Such a neighbour exists
    because \(t\) has degree at least two and every neighbour of
    \(t\) is in \(K\). Since \(uv\) is an edge in \(G[K]\) we get
    that \(\{t,u,v\}\) is a \(T\)-triangle which contains vertex
    \(v\).

    Now suppose \(v\) is a vertex in \(I\). Then \(v\) has at
    least two neighbours \(x,y\in{}K\) for which \(xy\) is an
    edge. If \(v\) is a terminal then it belongs to the
    \(T\)-triangle \(\{v,x,y\}\). If \(v\) is \emph{not} a
    terminal then---since \((G;T;k)\) is reduced with respect to
    \autoref{rr:only-T-neighbours-split}---we get that \(v\) is
    adjacent to at least one terminal vertex, which has to be in
    \(K\). Set \(x\) to be such a terminal neighbour of
    \(v\). Then \(v\) belongs to the \(T\)-triangle \(\{v,x,y\}\).
  \item Suppose not. Let \(S\) be a solution of \((G;T;k)\) of
    size at most \(k\), and let \(t\in{}(K\cap{}T)\setminus{}S\)
    be a terminal vertex on the clique side of \(G\) which is not
    in \(S\). If there are two other vertices \(x,y\) on the
    clique side which are also not in \(S\) then \(\{t,x,y\}\) is
    a \(T\)-triangle in \(G-S\), a contradiction. So we have that
    \(|K\setminus{}S|\leq{}2\). Now since \((G;T;k)\) is reduced
    with respect to \autoref{rr:split-kernel-yes-no} we
    have---part (4) of the rule---that \(|K|\geq{}k+2=|S|+2\),
    from which we get \(|K\setminus{}S|\geq{}2\). Thus
    \(|K\setminus{}S|=2\). Substituting these in the identity
    \(|K|=|K\setminus{}S|+|K\cap{}S|\) we get
    \(|S|\leq{}|K\cap{}S|\) which implies \(|S|=|K\cap{}S|\) and
    \(S=K\cap{}S\).

    Thus we get that \(K\) is of the form \(K=S\cup\{t,x\}\) for
    some vertex \(x\).  Now from part (5) of
    \autoref{rr:split-kernel-yes-no} we get that vertices \(t\)
    and \(x\) have a common neighbour, say \(y\), in set \(I\). So
    the \(T\)-triangle \(\{t,x,y\}\) is present in graph \(G-S\),
    a contradiction. Hence \(t\) must be in \(S\).%\qedhere
\end{enumerate}
\end{proof}

It is thus safe to pick a terminal vertex from the clique side
into the solution.

\begin{reductionrule}\label{rr:pick-clique-terminals}
  If there is a terminal vertex \(t\) on the clique side then
  delete \(t\) to get graph \(G'\). Set
  \(T'\gets{}T\setminus\{t\},k'\gets{}k-1\). The reduced instance
  is \(G';T';k')\).
\end{reductionrule}

\begin{lemma}\label{lem:pick-clique-terminals-rr-safe}
  \autoref{rr:pick-clique-terminals} is safe.
\end{lemma}
\begin{proof}
  Suppose \((G;T;k)\) is a \YES instance. % given as input to
  % Reduction Rule~\ref{rr:pick-clique-terminals} and let
  % \((G';T';k')\) be the corresponding instance output by the rule.
  % Then \(G'=G-\{t\},T'=G-\{t\},k'=k-1\) where \(t\) is a terminal
  % vertex on the clique side in \(G\).
  From \autoref{lem:neat-split-instance} we get that vertex \(t\)
  is present in \emph{every} solution of \((G;T;k)\) of size at
  most \(k\), and so from
  \autoref{obs:deleting-vertices-bridges-fwd-safe} we get that
  \((G';T';k')\) is a \YES instance.

  % Since split graphs are closed under vertex deletion we get that
  % \((G';T';k')\) is a valid instance of \SFVSS.

  % Let \((K,I)\) be the split partition of \(G\).

  % Suppose \((G;T;k)\) is a \YES instance, and let
  % \(S\subseteq{}V(G)\) be a solution of \((G;T;k)\) of size at
  % most \(k\). From Lemma~\ref{lem:neat-split-instance} we get
  % \(t\in{}S\). Let \(S'=S\setminus{}t\). Then
  % \(S'\subseteq{}V(G')\) and \(|S'|\leq{}k-1\). Now
  % \(G'-S'=(G-\{t\})-(S\setminus{}t)=G-S\), and so \(G'-S'\) has no
  % \(T\)-cycle. Thus \(S'\) is a solution of \((G';T';k')\) of size
  % at most \(k'\), and \((G';T';k')\) is a \YES instance.

  % Now suppose \((G';T';k')\) is a \YES instance, and let
  % \(S'\subseteq{}V(G')\) be a solution of \((G';T';k')\) of size
  % at most \(k'\). Let \(S=S'\cup\{t\}\). Then \(S\subseteq{}V(G)\)
  % and \(|S|\leq{}k\). Now
  % \(G-S=G-(S'\cup\{t\})=(G-\{t'\})-S'=G'-S'\), and so \(G-S\) has
  % no \(T\)-cycle. Thus \(S\) is a solution of \((G;T;k)\) of size
  % at most \(k\), and \((G;T;k)\) is a \YES instance.
  If \((G';T';k')\) is a \YES instance then we get from
  \autoref{obs:adding-solution-vertex-safe} that \((G;T;k)\) is a
  \YES instance as well.
\end{proof}

\begin{observation}\label{obs:T-equals-I}
  Let \((G;T;k)\) be reduced with respect to
  \Autoref{rr:split-kernel-yes-no,rr:delete-isolates,rr:only-T-neighbours-split,rr:pick-clique-terminals}. Let
  \((K,I)\) be the split partition of \(G\). Then \(T=I\) and
  every vertex in \(K\) has a neighbour in \(I\).
\end{observation}
\begin{proof}
  Since \((G;T;k)\) is reduced with respect to
  \autoref{rr:pick-clique-terminals} we have that no vertex on the
  clique side of \(G\) is a terminal. Thus
  \(T\subseteq{}I\). Suppose there is a non-terminal vertex
  \(v\in{}I\). From part (2) of
  \autoref{lem:neat-split-instance} we get that \(v\) must be
  adjacent to some terminal vertex. This cannot happen because
  every neighbour of \(v\) is in \(K\) and none of them is a
  terminal. So every vertex in \(I\) is a terminal.  Thus
  \(I\subseteq{}T\), and hence \(T=I\).

  Every vertex in \(K\) is a non-terminal, and every non-terminal
  is adjacent to some terminal vertex.  So every vertex in \(K\)
  must have a neighbour in \(I\).
\end{proof}

Our kernelization algorithm can be thought of having two main
parts: (i) bounding the number of vertices on the clique side by
\(\Oh(k)\), and (ii) bounding the number of independent set
vertices in the neighbourhood of each clique-side vertex by
\(k\). We now describe the second part. We need some more
notation. For a vertex \(v\in{}K\) on the clique side of graph
\(G\) we use (i) \(N_{1}(v)\) for the set of neighbours
\(N(v)\cap{}I\) of \(v\) on the independent side \(I\), and (ii)
\(N_{2}(v)\) to denote the set of all \emph{other} clique
vertices---than \(v\)---which are adjacent to some vertex in
\(N_{1}(v)\); that is,
\(N_{2}(v)=N(N_{1}(v))\setminus\{v\}\). Informally, \(N_{2}(v)\)
is the second neighbourhood of \(v\) ``going via \(I\)''. We use
\(B(v)\) to denote the bipartite graph obtained from
\(G[N_{1}(v)\cup{}N_{2}(v)]\) by deleting every edge with both its
endvertices in \(N_{2}(v)\). Equivalently: Let \(H\) be the
(bipartite) graph obtained by deleting, from \(G\), every edge
which has both its ends on the clique side of \(G\). Then
\(B(v)=H[N_{1}(v)\cup{}N_{2}(v)]\).  We call \(B(v)\) the
bipartite graph \emph{corresponding to} vertex
\(v\in{}K\). %\todo{Add a figure which shows \(N_{}(v),N_{2}(v),B(v)\)}

\paragraph*{Bounding the Independent-side Neighbourhood of  a Vertex on the Clique Side }
%In order to reduce degree of a vertex, say, $u$ in $C$  with $|N(u)\cap I|>k+1 $, we carefully delete edges incident between $u$ and set $I$.
The first reduction rule of this part applies when there is a
vertex \(v\in{}K\) which is part of more than \(k\)
\(T\)-triangles and these \(T\)-triangles are pairwise
vertex-disjoint apart from the one common vertex \(v\). In this
case any solution of size at most \(k\) must contain \(v\), so we
delete \(v\) and reduce \(k\).

% We reduce degree of a vertex, say $u$, in $C$ which is adjacent with more than $k + 1$ vertices in $I$ by carefully deleting edges incident on $u$.
% Recall that for each $u$ in $C$, $N_1(u)$ and $N_2(u)$ denote the sets $N(u) \cap I $ and $N(N_1(u))\setminus\set{u}$ respectively.
% %We start with the following rule.
% We apply Reduction Rule~\ref{rr:max-mat} or \ref{rr:degC} depending on the number of matching edges in bipartite graph induced on edges with one end in $N_1(u)$ and another in $N_2(u)$.

\begin{lemma}\label{lem:large-matching-forces-v}
  Let \(v\in{}K\) be a vertex on the clique side of graph \(G\)
  such that the bipartite graph \(B(v)\) contains a
  \emph{matching} of size at least \(k+1\). Then \emph{every}
  \(T\)-FVS of \(G\) of size at most \(k\) contains \(v\).
\end{lemma}
\begin{proof}
  Suppose not; let \(S\subseteq{}V(G)\;;\;|S|\leq{}k\) be a
  \(T\)-FVS of \(G\) of size at most \(k\) such that
  \(v\notin{}S\). Let \(M=\{x_{1}y_{1},\dotsc,x_{k+1}y_{k+1}\}\)
  be a matching in graph \(B(v)\). Note that every edge in \(M\)
  is present in graph \(G\). Further, each edge
  \(x_{i}y_{i}\in{}M\) (i) has one end in the clique side \(K\) of
  \(G\) and the other end in the independent side \(I\), and (ii)
  forms a triangle \(\{v,x_{i},y_{i}\}\) in \(G\) together with
  vertex \(v\). Since every vertex in \(I\) is a terminal we get
  that each triangle of the form \(\{v,x_{i},y_{i}\}\) is a
  \(T\)-triangle.

  Now since \(S\) is a \(T\)-FVS of \(G\) it has a non-empty
  intersection with every triangle in the set
  \(\{\{v,x_{i},y_{i}\}\;;\;1\leq{}i\leq{}(k+1)\}\). Since
  \(v\notin{}S\) we get that \(S\) has a non-empty intersection
  with every set in the collection
  \(\{\{x_{i},y_{i}\}\;;\;1\leq{}i\leq{}(k+1)\}\). Since the
  vertex pairs \(\{x_{i},y_{i}\}\;;\;1\leq{}i\leq{}(k+1)\) are
  pairwise disjoint we get that \(S\) contains at least \(k+1\)
  distinct vertices from the set
  \(\cup_{i=1}^{k+1}\{x_{i},y_{i}\}\). This contradicts the
  assumption \(|S|\leq{}k\).
\end{proof}

\begin{reductionrule}\label{rr:max-mat}
  If there is a vertex \(v\) on the clique side \(K\) of graph
  \(G\) such that the bipartite graph \(B(v)\) has a
  \emph{matching} of size at least \(k+1\) then delete vertex
  \(v\) from \(G\) to get graph \(G'\). Set
  \(T'\gets{}T,k'\gets{}k-1\). The reduced instance is
  \((G';T';k')\).
\end{reductionrule}
\begin{lemma}\label{lem:max-mat-rr-safe}
  \autoref{rr:max-mat} is safe.
\end{lemma}
\begin{proof}
  Since---\autoref{obs:T-equals-I}---vertex \(v\) is not a
  terminal vertex in \(G\) we get that
  \(T'\subseteq{}V(G')\). Suppose \((G;T;k)\) is a \YES
  instance. Then since \(v\) is in present in \emph{every}
  solution of \((G;T;k)\) of size at most
  \(k\)---\autoref{lem:large-matching-forces-v}---we get from
  \autoref{obs:deleting-vertices-bridges-fwd-safe} that
  \((G';T';k')\) is a \YES instance.

  If \((G';T';k')\) is a \YES instance then from
  \autoref{obs:adding-solution-vertex-safe} we get that
  \((G;T;k)\) is a \YES instance as well.
\end{proof}

% For a given instance $(G; T; k)$,  if there exists a vertex $u$ in $C$ such that there is a matching of size at least $k+1$ between $N_1(u) $ and $N_2(u) $ then include $u$ in solution. The new instance is $(G - u;T \setminus\set{u}; k-1)$.

% If $u$ is not in a solution then the solution must contain at least one vertex from each of the edges in the matching as $u$ forms a triangle with each of these edges.
% Since there are at least $k+1$ such edges, solution can not exclude $u$.

%implying that the solution has size at least $k+1 $, a contradiction. Hence, $u$ must belong to every solution.

%Now, suppose there is a vertex $u$ in $C $ such that $|N_1(u)|>k $ and the maximum matching between $N_1(u) $ and $N_2(u) $ has size at most $k $.

Let \((G;T;k)\) be an instance which is reduced with respect to
\autoref{rr:max-mat}. We show that if there is a vertex
\(v\in{}K\) on the clique side of \(G\) which has more than \(k\)
neighbours in the independent side \(I\), then we can find an
\emph{edge} of the form \(vw\;;\;w\in{}I\) which can safely be
deleted from the graph. We get this by a careful application of
the ``matching'' version (\autoref{lem:expansion-matching}) of the
Expansion Lemma together with
\autoref{lem:extra-expansion-vertex}. Let \(v\) be such a vertex
and let \(P=N_{2}(v),Q=N_{1}(v),t=1\). Then \((P,Q)\) is a
bipartition of the graph \(B(v)\) corresponding to vertex
\(v\). Let \(\ell\leq{}k\) be the size of a maximum matching of
\(B(v)\).  Note that \(|Q|\geq(k+1)>\ell{}t\) and that---by part
(1) of \autoref{lem:neat-split-instance}---there are no isolated
vertices in set \(Q\). Thus \autoref{lem:extra-expansion-vertex}
applies to graph \(B(v)\) together with \(P,Q,t=1\). Since a
\(1\)-expansion from \(X\) into \(Y\) contains a matching between
\(X\) and \(Y\) which saturates \(X\) we get

%\todo{Change \(q\) to \(t\) when referring to the expansion
%  lemma.}
\begin{corollary}\label{cor:large-degree-no-matching-expansion}
  Let \((G;T;k)\) be an instance which is reduced with respect to
  \autoref{rr:max-mat}. Suppose there is a vertex
  \(v\in{}K\) on the clique side of \(G\) which has more than
  \(k\) neighbours in the independent side \(I\). Then we can
  find, in polynomial time, non-empty vertex sets
  \(X\subseteq{}N_{2}(v)\subseteq{}K,Y\subseteq{}N_{1}(v)\subseteq{}I\)
  and a vertex \(w\in{}Y\) such that (i) there is a matching \(M\)
  between \(X\) and \(Y\) which saturates every vertex of \(X\)
  and does \emph{not} saturate \(w\), and (ii)
  \(N_{G}(Y)=X\cup\{v\}\).
\end{corollary}

%\todo{Add a figure or two.}
%\todo{Add some filler text here.}

\begin{lemma}\label{lem:v-or-all-of-X}
  Let \((G;T;k)\) be an instance which is reduced with respect to
  \autoref{rr:max-mat}, and let \(v\in{}K\) be a vertex
  on the clique side which has more than \(k\) neighbours in the
  independent side \(I\). Let
  \(X\subseteq{}K,w\in{}Y\subseteq{}I,M\subseteq{}E(G[X\cup{}Y])\)
  be as guaranteed to exist by
  \autoref{cor:large-degree-no-matching-expansion}. Let
  \(G'=G-\{vw\}\), and let \(S'\) be a \(T\)-THS of \(G'\) of size
  at most \(k\). If \(v\notin{}S'\) then
  \((S'\setminus{}Y)\cup{}X\) is a \(T\)-THS of \(G'\) of size at
  most \(k\).
\end{lemma}
\begin{proof}
  Let \(M=\{x_{1}y_{1},\dotsc,x_{|X|}y_{|X|}\}\) be a matching in
  \(G\) between \(X\) and \(Y\) which saturates all of \(X\) and
  does not saturate \(w\). Since \(vw\notin{}M\) we get that
  matching \(M\) is present in graph \(G'\) as well.  Thus
  \(\{\{v,x_{1},y_{1}\},\dotsc,\{v,x_{|X|},y_{|X|}\}\}\) is a set
  of \(|X|\)-many \(T\)-triangles in \(G'\) which pairwise
  intersect exactly in \(\{v\}\). Let \(S'\) be a \(T\)-THS of
  \(G'\) of size at most \(k\) which does \emph{not} contain
  \(v\). Then \(S'\) contains at least one vertex from each of the
  sets \(\{x_{i},y_{i}\}\;;\;1\leq{}i\leq{}|X|\).  Let
  \(\hat{S}=(S'\setminus{}Y)\cup{}X\). Then we can get \(\hat{S}\)
  from \(S'\) as follows.

  % Then
  % \(\hat{S}\subseteq((S'\setminus\{y_{1},\dotsc,y_{|X|}\})\cup\{x_{1},\dotsc,x_{|X|}\})\).
  
  \begin{itemize}
  \item For each edge \(x_{i}y_{i}\in{}M\),
    \begin{itemize}
    \item  if \(S'\cap{}\{x_{i},y_{i}\}=\{y_{i}\}\) then delete
     \(y_{i}\) from \(S'\) and add \(x_{i}\), and,
   \item if \(S'\cap{}\{x_{i},y_{i}\}=\{x_{i},y_{i}\}\) then
     delete \(y_{i}\) from \(S'\) (and don't add anything).
    \end{itemize}
  \item Delete all of \(Y\setminus\{y_{1},\dotsc,y_{|X|}\}\) from
    \(S'\). 
  \end{itemize}
  Thus to get \(\hat{S}\) from \(S'\) we add at most as many
  vertices as we delete, and so it is the case that
  \(|\hat{S}|\leq{}|S'|\leq{}k\).

  Consider the induced subgraphs \(H'=G'-S'\) and
  \(\hat{H}=G'-\hat{S}\) of \(G'\). Of these \(H'\) has no
  \(T\)-triangles. Every vertex of \(\hat{H}\) which is \emph{not}
  present in \(H'\) belongs to the set \(Y\). So if \(\hat{H}\)
  contains a \(T\)-triangle then each such \(T\)-triangle must
  contain a vertex from \(Y\). Now from
  \autoref{cor:large-degree-no-matching-expansion} we get
  that \(N(Y)\subseteq{}X\cup\{v\}\) and by definition we have
  that no vertex in \(X\) is present in \(\hat{H}\). Thus every
  vertex in \(Y\) has degree at most one in \(\hat{H}\). So no
  vertex in \(Y\) is in any \(T\)-triangle in \(\hat{H}\). Hence
  there are no \(T\)-triangles in \(\hat{H}\). Thus
  \(\hat{S}=(S'\setminus{}Y)\cup{}X\) is a \(T\)-THS of \(G'\) of
  size at most \(k\).

\end{proof}

\begin{reductionrule}\label{rr:degK}
  If there is a vertex \(v\) on the clique side \(K\) of graph
  \(G\) such that \(v\) has more than \(k\) neighbours in the
  independent side \(I\), then find a vertex \(w\in{}I\) as
  described by \autoref{cor:large-degree-no-matching-expansion}
  and delete the edge \(vw\) to get graph \(G'\).  Set
  \(T'\gets{}T,k'\gets{}k\). The reduced instance is
  \((G';T';k')\).
  % For a given instance $(G; T; k)$ if there exists vertex $u$ in
  % $C$ and vertex $v$ in $B$ as described in the preceding
  % paragraph then output the reduced instance $(G - (u,v);T;k)$.
\end{reductionrule}

\begin{lemma}\label{lem:degK-rr-safe}
  \autoref{rr:degK} is safe.
\end{lemma}
\begin{proof}
  If \((G;T;k)\) is a \YES instance then we get from
  \autoref{obs:deleting-vertices-bridges-fwd-safe} that
  \((G'=G-\{vw\};T';k')\) is a \YES instance.

  Now suppose \((G';T';k')\) is a \YES instance. Let \(S'\) be a
  \(T\)-THS of \(G'\) of size at most \(k\). If \(v\in{}S'\) then
  we have \(G-S'=G'-S'\) and in this case \(S'\) is a \(T\)-THS of
  \(G\) as well, of size at most \(k\). If \(v\notin{}S'\) then
  from \autoref{lem:v-or-all-of-X} we get that
  \(\hat{S}=(S'\setminus{}Y)\cup{}X\) is a \(T\)-THS of \(G'\) of
  size at most \(k\). Thus the graph \(H'=G'-\hat{S}\) has no
  \(T\)-triangles. The only difference between graphs \(H'\) and
  \(H=G-\hat{S}\) is that the latter graph has the extra edge
  \(vw\). So if \(H\) contains a \(T\)-triangle then each such
  \(T\)-triangle must contain both the vertices \(\{v,w\}\).

  From \autoref{cor:large-degree-no-matching-expansion} we
  know that \(N(w\in{}Y)\subseteq{}X\cup\{v\}\). Thus vertex \(w\)
  has no neighbours in the graph \(H'=(G-\{vw\})-\hat{S}\), and
  has exactly one neighbour---namely, vertex \(v\)---in the graph
  \(H=G-\hat{S}\). So \(w\) is not part of any triangle in graph
  \(H\). Thus \(\hat{S}=(S'\setminus{}Y)\cup{}X\) is a \(T\)-THS
  of graph \(G\) of size at most \(k\).
\end{proof}

%Suppose $|N_1(u)| > k+1 $ and we delete the edge $(u,v)$ in $ E(G) $ to obtain the graph $G' $.
%It follows from the preceding lemma that the instance $(G',T,k) $ is an \yes{} instance if and only if $(G',T,k) $ has a solution that contains $u $ or all of $A $.
%But, this solution also hits every triangle killed by the deletion of the edge $(u,v) $ and hence, we have the following reduction rule.
%As guaranteed by Lemma \ref{lem:expansion-matching}, we know there exists a vertex, say $v$ in $B$, such that no edge in the $1$-expansion is incident on $v$.

We now show how to bound the number of vertices on the clique side
\(K\) of an instance \((G;T;k)\) which is reduced with respect to
\autoref{rr:degK}. 

\paragraph*{Bounding the Size of the Clique Side}

We partition the clique side $K$ into three parts and bound the
size of each part separately.  To do this we first find a
$3$-approximate solution $\tilde{S}$ to $(G; T; k)$. For this we
initialize \(\tilde{S}\gets{}\emptyset\) and iterate as follows: If there
is a vertex \(v\) in the independent side \(I\) such that \(v\) is
part of a triangle \(\{v,x,y\}\) in the graph \(G-\tilde{S}\)---note that
in this case \(\{x,y\}\subseteq{}K\)---then we set
\(\tilde{S}\gets{}\tilde{S}\cup\{v,x,y\}\). We repeat this till there is no such
vertex \(v\in{}I\) or till \(|\tilde{S}|\) becomes larger than \(3k\),
whichever happens first.

\begin{reductionrule}\label{rr:large-approx-solution-no}
  Let \((G;T;k)\) be an instance which is reduced with respect to
  \autoref{rr:degK} and let \(\tilde{S}\) be the set
  constructed as described above. If \(|\tilde{S}|>3k\) then return
  \(I_{NO}\).
\end{reductionrule}
\begin{lemma}\label{lemma:proof-rr-large-approx-solution-no}
 \autoref{rr:large-approx-solution-no} is safe.
\end{lemma}
\begin{proof}
  By \autoref{lem:neat-split-instance} we have that the
  terminal set \(T\) of graph \(G\) is exactly its independent
  side \(I\). Hence each triangle whose vertex set is added to
  \(\tilde{S}\) by the construction is a \(T\)-triangle, and these
  vertex sets are pairwise disjoint. If \(|\tilde{S}|>3k\) then
  graph \(G\) contains more than \(k\) pairwise vertex-disjoint
  \(T\)-triangle and therefore is a \NO instance, as is
  \(I_{NO}\).
\end{proof}

At this point we have that the cardinality of the approximate
solution $\tilde{S}$ is at most $3k$. We now partition the sets
\(K,I\) into three parts each and bound each part separately (See
\autoref{fig:bound-K1}.):
%\todo{Draw a prettier picture with labels  matching the text.}

\begin{itemize}
\item $K_{\tilde{S}}$ is the set of clique-side vertices included in $\tilde{S}$:
  $K_{\tilde{S}}=K\cap{}\tilde{S}$.
\item $I_{\tilde{S}}$ is the set of independent-side vertices included in
  $\tilde{S}$: $I_{\tilde{S}}= I\cap{}\tilde{S}$.
\item $K_{0}$ is the set of clique-side vertices not in \(\tilde{S}\)
  whose neighbourhoods in the independent side $I$ are all
  contained in $I_{\tilde{S}}$:
  $K_{0}=\set{u\in(K\setminus{}K_{\tilde{S}})\;;\;N(u)\cap{}I\subseteq{}I_{\tilde{S}}}$;
\item $I_{0}$ is the set of independent-side vertices not in
  \(\tilde{S}\) whose neighbourhoods are all contained in
  $K_{\tilde{S}}$:
  $I_{0}=\set{v\in{}I\setminus{}I_{\tilde{S}}\;;\;N(v)\subseteq
    K_{\tilde{S}}}$
\item \(K_{1},I_{1}\) are the remaining vertices in each set:
  \(K_{1}=K\setminus{}(K_{\tilde{S}}\cup{}K_{0})\) and
  \(I_{1}=I\setminus{}(I_{\tilde{S}}\cup{}I_{0})\). 
\item $K_{1}$ is the set of clique-side vertices not in \(\tilde{S}\)
  which have at least one neighbour in $I$ outside of
  $I_{\tilde{S}}\cup{}I_{0}$. Equivalently, it is the set of clique-side
  vertices not in $K_{\tilde{S}}\cup{}K_{0}$:
  \(K_{1}=K\setminus{}(K_{\tilde{S}}\cup{}K_{0})\).
\item $I_{1}$ is the set of independent-side vertices which are
  not in
  $I_{\tilde{S}}{}\cup{}I_{0}$:\(I_{1}=I\setminus{}(I_{\tilde{S}}\cup{}I_{0})\).
  Since $\tilde{S}$ is a solution each vertex in $I_{1}$---being a
  terminal---can have exactly one neighbour in $K_{1}$.
\end{itemize}

We list some simple properties of this partition which we need
later in our proofs. %\todo{Do we actually use all of these?}
\begin{observation}\label{obs:split-approx-partition-properties}
  $|K_{\tilde{S}}|\leq{}2k$ and $|I_{\tilde{S}}|\leq{}k$. Each vertex in \(K_{1}\)
  has (i) no neighbour in \(I_{0}\) and (ii) at least one
  neighbour in \(I_{1}\). Each vertex in \(I_{1}\) has exactly one
  neighbour in \(K_{1}\). The bipartite graph obtained from
  $G[K_{1}\cup{}I_{1}]$ by deleting all the edges in \(G[K_{1}]\)
  is a forest where each connected component is a star.
\end{observation}
\begin{proof}
  It follows directly from the construction that $|K_{\tilde{S}}|\leq{}2k$
  and $|I_{\tilde{S}}|\leq{}k$.

  Let \(v\) be a vertex in \(K_{1}\). Then
  \(v\notin{}K_{\tilde{S}}\) by construction.  If \(v\) has a neighbour
  \(w\in{}I_{0}\) then \(w\in{}I_{0}\) has a neighbour outside of
  \(K_{\tilde{S}}\), a contradiction.

  Since \(I\) is the set of terminals we
  get---\autoref{rr:only-T-neighbours-split}---that vertex
  \(v\) has at least one neighbour in the set \(I\). Since the
  vertices in \(I_{0}\) do not have neighbours outside the set
  \(K_{\tilde{S}}\) we get that \(v\notin{}K_{\tilde{S}}\) does not have a
  neighbour in the set \(I_{0}\).  So if \(v\) has no neighbour in
  \(I_{0}\) then its neighbourhood on the independent side is
  contained in the set \(I_{\tilde{S}}\), which implies that \(v\) is in
  \(K_{0}\) and not in \(K_{\tilde{S}}\), a contradiction.
  
  If a vertex \(v\in{}I_{1}\) has two neighbours \(x,y\) in the
  set \(K_{1}\) then---since \(v\in{}I\) is a terminal---the
  vertices \(\{v,x,y\}\) form a \(T\)-triangle which does not
  intersect the \(T\)-THS \(\tilde{S}=K_{\tilde{S}}\cup{}I_{\tilde{S}}\), a
  contradiction. So each vertex in \(I_{1}\) has exactly one
  neighbour in \(K_{1}\), which implies that the bipartite graph
  obtained from $G[K_{1}\cup{}I_{1}]$ by deleting all the edges in
  \(G[K_{1}]\) is a forest where each connected component is a
  star.
\end{proof}

\begin{figure}
\includegraphics[scale=0.3]{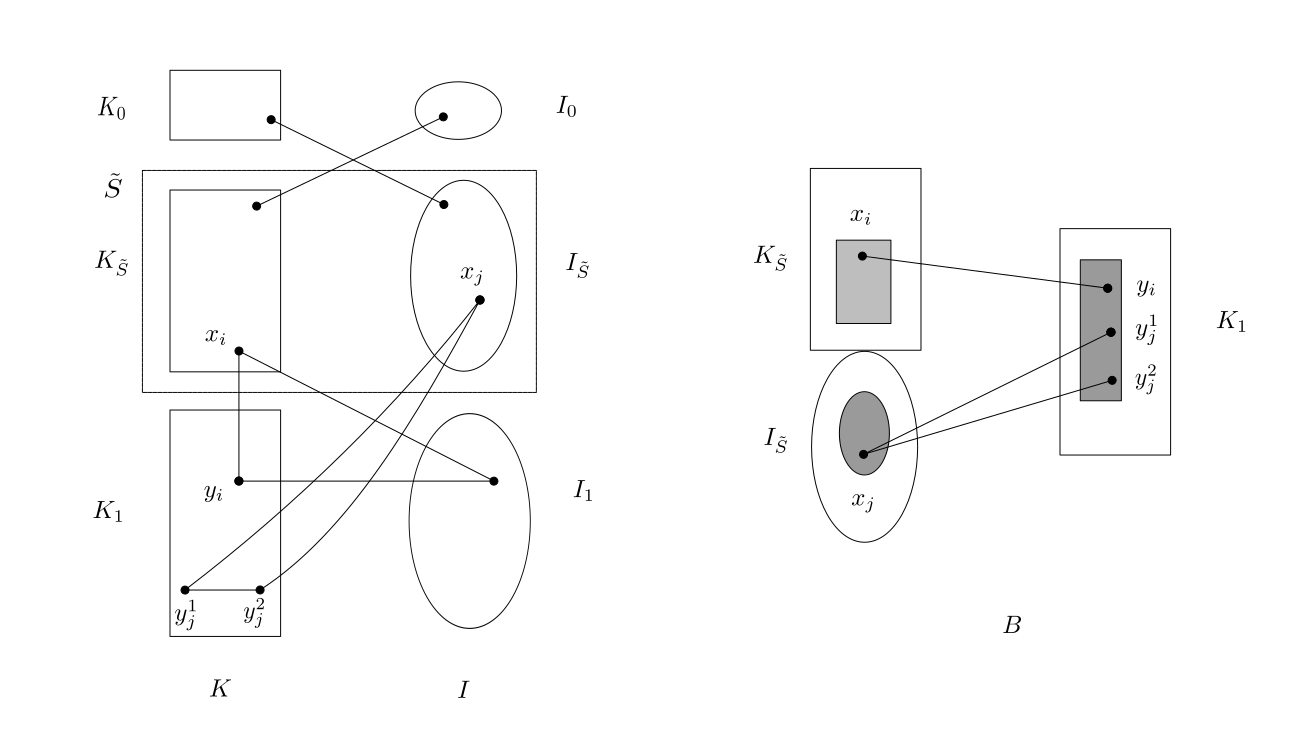}
\caption{The figure on the left shows the partition of $V(G)$ as
  described after
  \autoref{lemma:proof-rr-large-approx-solution-no}. On the right
  side, we have graph $B$ as described before
  \autoref{cor:Kone-expansion-bound}. The shaded regions in
  $\tilde{S}$ and $K_1$ represent sets $X, Y$, respectively, as
  defined in \autoref{cor:Kone-expansion-bound}.}
\label{fig:bound-K1}
\end{figure}

Let \(H\) be the bipartite graph obtained from
$G[I_{\tilde{S}}\cup K_0]$ by deleting all the edges in
\(G[K_{0}]\). Since---\autoref{obs:T-equals-I}---every
vertex in the set \(K_{0}\) has at least one neighbour in the set
\(I\) and since \((N(K_{0})\cap{}I)\subseteq{}I_{\tilde{S}}\) by
construction, we get that there are no isolated vertices in graph
\(H\). So if $|K_0| \geq{} 2|I_{\tilde{S}}|$ then
\autoref{lem:expansion-lemma} applies to graph \(H\) with
\(P\gets{}I_{\tilde{S}},Q\gets{}K_{0},t\gets{}2\) and we get

\begin{corollary}\label{cor:Kzero-expansion-bound}
  Let \((G;T;k)\) be an instance which is reduced with respect to
  \autoref{rr:large-approx-solution-no}, and let the
  sets \(K_{\tilde{S}},K_{0},K_{1},I_{\tilde{S}},I_{0},I_{1}\) be as described
  above. If \(|K_0|\geq{}2|I_S|\) then we can find, in polynomial
  time, non-empty vertex sets
  \(X\subseteq{}I_{\tilde{S}}\subseteq{}I,Y\subseteq{}K_{0}\subseteq{}K\)
  such that (i) \(X\) has a \(2\)-expansion \(M\) into \(Y\), and
  (ii) \(N_{G}(Y)=X\).
\end{corollary}

\begin{lemma}\label{lem:X-and-none-of-Y}
  Let \((G;T;k)\) be an instance which is reduced with respect to
  \autoref{rr:large-approx-solution-no}, and let the
  sets \(K_{\tilde{S}},K_{0},K_{1},I_{\tilde{S}},I_{0},I_{1}\) be
  as described above. Suppose \(|K_0|\geq{}2|I_{\tilde{S}}|\), and
  let
  \(X\subseteq{}I_{\tilde{S}}\subseteq{}I,Y\subseteq{}K_{0}\subseteq{}K,M\subseteq{}E(G[X\cup{}Y])\)
  be as guaranteed to exist by
  \autoref{cor:Kzero-expansion-bound}. If \(S\) is a
  \(T\)-THS of graph \(G\) of size at most \(k\) then
  \((S\setminus{}Y)\cup{}X\) is also a \(T\)-THS of \(G\) of size
  at most \(k\).
\end{lemma}
\begin{proof}
  Let \(X=\{x_{1},\dotsc,x_{|X|}\}\). For each vertex
  \(x_{i}\in{}X\) let \(x_{i}y_{i}^{1},x_{i}y_{i}^{2}\) be the two
  edges in \(M\) which are incident with \(x_{i}\). Then the set
  \(\{y_{1}^{1},y_{1}^{2},\dotsc,y_{|X|}^{1},y_{|X|}^{2}\}\subseteq{}Y\)
  has size exactly \(2|X|\), and the \(|X|\)-many sets
  \(\{\{x_{1},y_{1}^{1},y_{1}^{2}\},\dotsc,\{x_{|X|},y_{|X|}^{1},y_{|X|}^{2}\}\}\)
  form pairwise vertex-disjoint \(T\)-triangles in \(G\).

  The \(T\)-THS \(S\) of \(G\) of size at most \(k\) contains at
  least one vertex from each of the sets
  \(\{\{x_{i},y_{i}^{1},y_{i}^{2}\}\;;\;1\leq{}i\leq{}|X|\}\).
  Let \(\hat{S}=(S\setminus{}Y)\cup{}X\). Then we can get
  \(\hat{S}\) from \(S\) as follows.

  % Then
  % \(\hat{S}\subseteq((S'\setminus\{y_{1},\dotsc,y_{|X|}\})\cup\{x_{1},\dotsc,x_{|X|}\})\).
  
  \begin{itemize}
  \item For each triangle
    \(\{\{x_{i},y_{i}^{1},y_{i}^{2}\}\;;\;1\leq{}i\leq|X|\}\),
    \begin{itemize}
    \item if \(x_{i}\in{}S\) then set
      \(S\gets{}S\setminus\{y_{i}^{1},y_{i}^{2}\}\);
    \item if \(x_{i}\notin{}S\) then set
      \(S\gets{}(S\setminus\{y_{i}^{1},y_{i}^{2}\})\cup\{x_{i}\}\).
    \end{itemize}
  \item Delete all of
    \(Y\setminus\{y_{1}^{1},y_{1}^{2},\dotsc,y_{|X|}^{1},y_{|X|}^{2}\}\)
    from \(S\).
  \end{itemize}
  Thus to get \(\hat{S}\) from \(S\) we add at most as many
  vertices as we delete, and so it is the case that
  \(|\hat{S}|\leq{}|S|\leq{}k\).

  Consider the induced subgraphs \(H=G-S\) and
  \(\hat{H}=G-\hat{S}\) of \(G\). Of these \(H\) has no
  \(T\)-triangles. Every vertex of \(\hat{H}\) which is \emph{not}
  present in \(H\) belongs to the set \(Y\). So if \(\hat{H}\)
  contains a \(T\)-triangle then each such \(T\)-triangle must
  contain a vertex from \(Y\). Now from
  \autoref{cor:Kzero-expansion-bound} we get that \(N(Y)=X\)
  and by definition we have that no vertex in \(X\) is present in
  \(\hat{H}\). Thus each vertex in \(Y\) has degree zero in
  \(\hat{H}\), and so is not in any \(T\)-triangle in
  \(\hat{H}\). Hence there are no \(T\)-triangles in
  \(\hat{H}\). Thus \(\hat{S}=(S\setminus{}Y)\cup{}X\) is a
  \(T\)-THS of \(G'\) of size at most \(k\).
\end{proof}

\begin{reductionrule}\label{rr:for-K0}
  If $|K_0|\geq{}2|I_{\tilde{S}}|$ then find sets
  $X \subseteq I_{\tilde{S}}$ and $Y \subseteq K_{0} $ as
  described by \autoref{cor:Kzero-expansion-bound}.  Set
  \(G'\gets{}G-X,T'\gets{}T\setminus{}X,k'\gets{}k-|X|\). The
  reduced instance is \((G';T';k')\).
\end{reductionrule}

\begin{lemma}\label{lem:for-K0-rr-safe}
  \autoref{rr:for-K0} is safe.
\end{lemma}
\begin{proof}
  If \((G;T;k)\) is a \YES instance then we get from
  \autoref{lem:X-and-none-of-Y} that \(G\) has a \(T\)-THS \(S\)
  of size at most \(k\) which contains all of \(X\). By applying
  part (1) of
  \autoref{obs:deleting-vertices-bridges-fwd-safe} \(|X|\)
  times we get that \(G'=G-X,T'=T\setminus{}X,k'=k-|X|\) is a \YES
  instance.

  Now suppose \((G';T';k')\) is a \YES instance. Observe that we
  can get graph \(G\) from \(G'\) by adding, in turn, each vertex
  \(x\in{}X\) and some edges incident on \(x\), and also that
  every vertex that we add in this process is a terminal vertex in
  graph \(G\).  Hence by applying
  \autoref{obs:adding-solution-vertex-safe} \(|X|\) times
  we get that \((G;T=T'\cup{}X;k=k'+|X|)\) is a \YES instance.
\end{proof}

At this point we have the bounds \(|K_{\tilde{S}}|\leq{}2k\) and
\(|K_{0}|<2|I_{\tilde{S}}|=2k\). We now use a more involved
application of the Expansion Lemma to bound the size of the
remaining part $K_{1}$ of the clique side. The general idea is
that if \(K_{1}\) is at least twice as large as the approximate
solution \(\tilde{S}\) then the \(2\)-expansion which exists
between subsets of these two sets will yield a non-empty set of
``redundant'' vertices in \(K_{1}\).

%\todo{We need new figures which match the description below.}

Consider the bipartite graph \(B\) obtained from the induced
subgraph \(G[\tilde{S}\cup{}K_{1}]\) of \(G\) by (i) deleting all
the edges in the two induced subgraphs \(G[\tilde{S}]\) and
\(G[K_{1}]\), respectively, and (ii) deleting every edge
\(uv\;;\;u\in{}K_{\tilde{S}},v\in{}K_{1}\) \emph{if and only if}
there is \emph{no} vertex \(w\in{}I_{1}\) such that \(\{u,v,w\}\)
form a triangle in \(G\). Consider a vertex \(v\in{}K_{1}\). If
\(v\) has a neighbour \(w\in{}I_{\tilde{S}}\) then the edge \(vw\)
is present in graph \(B\) and so \(v\) is not isolated in
\(B\). Now suppose \(v\) has no neighbour in
\(I_{\tilde{S}}\). From the construction we know that \(v\) has no
neighbour in \(I_{0}\) either. Then from
\autoref{lem:neat-split-instance} and
\autoref{obs:T-equals-I} we get that there is a triangle
\(\{v,x,y\}\) in \(G\) where
\(x\in{}(I\setminus{}(I_{0}\cup{}I_{\tilde{S}}))=I_{1}\) and
\(y\in{}K\).  Now by construction vertex \(x\in{}I_{1}\) has no
neighbour in the set \(K_{0}\), and from
\autoref{obs:split-approx-partition-properties} we get
that \(x\) has no neighbour other than \(v\) in the set
\(K_{1}\). Thus we get that \(y\in{}K_{\tilde{S}}\), and hence
that the edge \(vy\) survives in graph \(B\). Hence \(v\) is not
isolated in \(B\) in this case either.  So if
\(|K_{1}| \geq{} 2|\tilde{S}|\) then
\autoref{lem:expansion-lemma} applies to the bipartite graph
\(B\) with \(P\gets{}\tilde{S},Q\gets{}K_{1},t\gets{}2\) and we
get

\begin{corollary}\label{cor:Kone-expansion-bound}
  Let \((G;T;k)\) be an instance which is reduced with respect to
  \autoref{rr:for-K0}, and let the sets
  \(K_{1},\tilde{S},K_{\tilde{S}},I_{\tilde{S}},K\) and the
  bipartite graph \(B\) be as described above. If
  \(|K_{1}|\geq{}2|\tilde{S}|\) then we can find, in polynomial
  time, non-empty vertex sets
  \(X\subseteq{}\tilde{S}=(K_{\tilde{S}}\cup{}I_{\tilde{S}}),Y\subseteq{}K_{1}\subseteq{}K\)
  such that (i) \(X\) has a \(2\)-expansion \(M\) into \(Y\), and
  (ii) \(N_{B}(Y)=X\).
\end{corollary}

\begin{lemma}\label{lem:X-and-none-of-Y-Kone-version}
  Let \((G;T;k)\) be an instance which is reduced with respect to
  \autoref{rr:large-approx-solution-no}, and let the
  sets \(K_{1},\tilde{S},K_{\tilde{S}},I_{\tilde{S}},K\) and the
  bipartite graph \(B\) be as described above. Suppose
  \(|K_{1}|\geq{}2|\tilde{S}|\), and let
  \(X\subseteq{}\tilde{S}=(K_{\tilde{S}}\cup{}I_{\tilde{S}}),Y\subseteq{}K_{1}\subseteq{}K,M\subseteq{}E(G[X\cup{}Y])\)
  be as guaranteed to exist by
  \autoref{cor:Kzero-expansion-bound}. If \(S\) is a
  \(T\)-THS of graph \(G\) of size at most \(k\) then
  \((S\setminus{}Y)\cup{}X\) is also a \(T\)-THS of \(G\) of size
  at most \(k\).
\end{lemma}
\begin{proof}
  Let \(X=\{x_{1},\dotsc,x_{|X|}\}\). Without loss of generality,
  let \(X\cap{}K_{\tilde{S}}=\{x_{1},\dotsc,x_{\ell}\}\) and
  \(X\cap{}I_{\tilde{S}}=\{x_{\ell+1},\dotsc,x_{|X|}\}\). Note
  that at most one of these two sets could be the empty set; this
  does not affect the remaining arguments.

  For each vertex \(x_{i}\in{}(X\cap{}K_{\tilde{S}})\) let
  \(x_{i}y_{i}\) be an edge in \(M\) which is incident with
  \(x_{i}\). Then we know by the construction that there is a
  vertex \(z_{i}\in{}I_{1}\) such that \(\{x_{i},y_{i},z_{i}\}\)
  is a \(T\)-triangle in \(G\). From
  \autoref{obs:split-approx-partition-properties} we get
  that the vertices \(z_{i}\in{}I_{1}\;;\;1\leq{}i\leq\ell\) are
  pairwise disjoint. Now for each vertex
  \(x_{j}\in{}(X\cap{}I_{\tilde{S}})\) let
  \(x_{j}y_{j}^{1},x_{j}y_{j}^{2}\) be the two edges in \(M\)
  which are incident with \(x_{j}\). Then we get that the vertices
  \(\{x_{j},y_{j}^{1},y_{j}^{2}\}\) form a \(T\)-triangle in
  \(G\). Putting these together we get that the \(|X|\)-many sets
  \(\{\{x_{1},y_{1},z_{1}\},\dotsc,\{x_{\ell},y_{\ell},z_{\ell}\},\{x_{\ell},y_{\ell}^{1},y_{\ell}^{2}\},\dotsc,\{x_{|X|},y_{|X|}^{1},y_{|X|}^{2}\}\}\)
  form pairwise vertex-disjoint \(T\)-triangles in \(G\).  The
  \(T\)-THS \(S\) of \(G\) of size at most \(k\) contains at least
  one vertex from each of these \(|X|\) sets.  Let
  \(\hat{S}=(S\setminus{}Y)\cup{}X\). Then we can get \(\hat{S}\)
  from \(S\) as follows.

  % Then
  % \(\hat{S}\subseteq((S'\setminus\{y_{1},\dotsc,y_{|X|}\})\cup\{x_{1},\dotsc,x_{|X|}\})\).
  
  \begin{itemize}
  \item For each triangle
    \(\{\{x_{i},y_{i},z_{i}\}\;;\;1\leq{}i\leq\ell\}\),
    \begin{itemize}
    \item if \(x_{i}\in{}S\) then set
      \(S\gets{}S\setminus\{y_{i}\}\);
  \item if \(x_{i}\notin{}S\) and \(y_{i}\in{}S\) then set
    \(S\gets{}(S\setminus\{y_{i}\})\cup\{x_{i}\}\).
    \end{itemize}
  \item For each triangle
    \(\{\{x_{j},y_{j}^{1},y_{j}^{2}\}\;;\;\ell+1\leq{}j\leq|X|\}\),
    \begin{itemize}
    \item if \(x_{j}\in{}S\) then set
      \(S\gets{}S\setminus\{y_{j}^{1},y_{j}^{2}\}\);
    \item if \(x_{j}\notin{}S\) then set
      \(S\gets{}(S\setminus\{y_{j}^{1},y_{j}^{2}\})\cup\{x_{j}\}\).
    \end{itemize}
  \item Delete all of
    \(Y\setminus\{y_{1}\dotsc,y_{\ell},y_{\ell+1}^{1},y_{\ell+1}^{2},\dotsc,y_{|X|}^{1},y_{|X|}^{2}\}\)
    from \(S\).
  \end{itemize}
  Thus to get \(\hat{S}\) from \(S\) we add at most as many
  vertices as we delete, and so it is the case that
  \(|\hat{S}|\leq{}|S|\leq{}k\).

  Consider the induced subgraphs \(H=G-S\) and
  \(\hat{H}=G-\hat{S}\) of \(G\). Of these \(H\) has no
  \(T\)-triangles. Every vertex of \(\hat{H}\) which is \emph{not}
  present in \(H\) belongs to the set \(Y\). So if \(\hat{H}\)
  contains a \(T\)-triangle then each such \(T\)-triangle must
  contain a vertex from \(Y\).

  Assume that there exists a $T$-triangle $\{y, z, u\}$ in $\hat{H}$, where vertices $y, z, u$ are in sets $Y, I$ and $K$, respectively. 
  As $Y$ is a subset of $K_1$, any vertex in $Y$ is not adjacent with a vertex in $I_0$. This implies that $z$ belongs to the set $I_{\tilde{S}} \cup I_1$.
  Now from \autoref{cor:Kone-expansion-bound} we get that
  \(N_{B}(Y)=X\) and since $\hat{S}$ contains $X$,
  no vertex in $X$ is present in $\hat{H}$.
  In other words, no vertex in set $N_G(y) \cap I_{\tilde{S}}$ is present in graph $\hat{H}$.  
  This implies that $z$ is in $I_1$. 
  Now, consider the vertex $u$ which is adjacent with $z$ and hence can not be in set $K_0$.
  Vertex $z$, which is in $I_1$, has exactly one neighbor in $K_1$ (\autoref{obs:split-approx-partition-properties}). Since both $y, u$ are adjacent with $z$ and $y$ is contained in $Y \subseteq K_1$, vertex $u$ is contained in set $K_{\tilde{S}}$.
  We now argue that vertex $u$ is adjacent with $y$ even in graph $B$. This follows from the fact that while constructing graph $B$, edge $yu$ is not deleted as there exists $z$ in $I_1$ which is adjacent with both $y$ and $u$.
  Hence vertex $u$ is contained in set $X$. 
  %Now from Corollary~\ref{cor:Kone-expansion-bound} we get that
  %\(N_{B}(Y)=X\) and from the construction of graph \(B\) we get
  %that for each \(y\in{}Y\),
  %\(((N_{G}(y)\cap{}I)\setminus{}N_{B}(y))=\{z\in{}I_{1}\;;\;\text{
  %  there is no vertex }x\in{}(K\setminus{}X)\;;\;\{x,y,z\}\text{ is a triangle in }G\}\).
  %\todo{This is what we want, but is
  %it true? Also: simplify this.}
    By definition we have that no
  vertex in \(X\) is present in \(\hat{H}\). Thus our assumption is wrong and  no vertex in
  \(Y\) is in any \(T\)-triangle in \(\hat{H}\).
  Hence there are
  no \(T\)-triangles in \(\hat{H}\), and
  \(\hat{S}=(S\setminus{}Y)\cup{}X\) is a \(T\)-THS of \(G'\) of
  size at most \(k\).
  %\todo{I am not very convinced by this
  %  argument. We should make it simpler as well.}
\end{proof}

\begin{reductionrule}\label{rr:for-K1}
  If $|K_{1}|\geq{}2|\tilde{S}|$ then find sets
  $X\subseteq{}\tilde{S}$ and $Y\subseteq{}K_{1}$ as described by
  \autoref{cor:Kone-expansion-bound}.  Set
  \(G'\gets{}G-X,T'\gets{}T\setminus{}X,k'\gets{}k-|X|\). The
  reduced instance is \((G';T';k')\).
\end{reductionrule}

\begin{lemma}\label{lem:for-K1-rr-safe}
  \autoref{rr:for-K1} is safe.
\end{lemma}
\begin{proof}
  If \((G;T;k)\) is a \YES instance then we get from
  \autoref{lem:X-and-none-of-Y-Kone-version} that \(G\) has a
  \(T\)-THS \(S\) of size at most \(k\) which contains all of
  \(X\). By applying part (1) of
  \autoref{obs:deleting-vertices-bridges-fwd-safe} \(|X|\)
  times we get that \(G'=G-X,T'=T\setminus{}X,k'=k-|X|\) is a \YES
  instance.

  Now suppose \((G';T';k')\) is a \YES instance. Observe that we
  can get graph \(G\) from \(G'\) by adding, in turn, each vertex
  \(x\in{}X\) and some edges incident on \(x\), and also that
  every vertex that we add in this process is a terminal vertex in
  graph \(G\).  Hence by applying
  \autoref{obs:adding-solution-vertex-safe} \(|X|\) times
  we get that \((G;T=T'\cup{}X;k=k'+|X|)\) is a \YES instance.
\end{proof}

%\todo{Start from here}

%\paragraph*{Proof of Theorem \ref{thm:split-kernel}}

Putting all these together, we get

\splitkernel*
\begin{proof}
  We describe such a kernelization algorithm. Given an instance
  \((G;T;k)\) of \SFVSS, our kernelization algorithm applies the
  reduction rules described in this section, exhaustively and in
  the given order, to get an instance \((G';T';k')\) to which none
  of the reduction rules applies.  That is, if at least one of the
  reduction rules applies to the current instance at any point,
  then the algorithm applies the \emph{first} such rule, and
  repeats the process with the reduced instance. The algorithm
  outputs the final instance \((G';T';k')\) as the kernel. Each
  reduction rule can be applied in polynomial time, and either (i)
  stops directly, or (ii) deletes at least one vertex or edge from
  the graph. Thus this entire procedure runs in polynomial time.

  % Consider an algorithm which apply reduction rules
  % described in this section on a given instance $(G; T; k)$ of
  % \SFVSS. The algorithm applies least indexed reduction rule
  % exhaustively before moving on to next reduction rule. 

  % If the above procedure outputs \(I_{\YES}\) (respectively,
  % \(I_{\NO}\)), then the kernelization algorithm outputs
  % \(I_{\YES}\) (respectively, \(I_{\NO}\)) and stops.  In the
  % remaining case, if \(|V(G')|=\Oh(k^{2}),|E(G')|=\Oh(k^{2})\),
  % and the clique side of $G'$ contains at most $10k$ vertices then
  % the algorithm outputs $(G'; T'; k)$ and stops; otherwise it
  % outputs $I_{\no}$ and stops.

  % Let $(G'; T'; k)$ be output of the algorithm on input
  % $(G; T; k)$. First part of the proof i.e. $(G; T; k)$ is a \yes
  % instance if and only if $(G'; T'; k)$ is a \yes instance follows
  % from the proof of safeness of each reduction rule.

  The correctness of this kernelization algorithm follows from the
  proofs of safeness of the various reduction rules.  We now argue
  the size bound. If $(G'; T'; k)$ is $I_{\yes}$ or $I_{\no}$ then
  this bound is trivially correct. Hence we assume that no
  reduction rule returns a trivial \yes or \no instance.  Since
  \autoref{rr:delete-isolates} is not applicable, there are no
  isolated vertices in the graph.  Let $(K', I')$ be the split
  partition of the split graph $G'$. Since
  \Autoref{rr:max-mat,rr:degK} are not applicable, every vertex in
  $K'$ is adjacent with at most $k$ vertices in $I'$. Since there
  are no isolated vertices in the graph, this implies
  $|I'| \le k \cdot |K'|$. Since
  \autoref{rr:large-approx-solution-no} did not return $I_{\no}$
  we have that the approximate solution, $\tilde{S}$, is of size
  at most $3k$. Analogous to the definitions after
  \autoref{lemma:proof-rr-large-approx-solution-no}, let
  \(K'_{\tilde{S}}=K' \cap \tilde{S},I'_{\tilde{S}}=I' \cap
  \tilde{S},K'_{0}=\set{u\in(K'\setminus{}K'_{\tilde{S}})\;;\;N(u)\cap{}I'\subseteq{}I'_{\tilde{S}}},K'_{1}=K'\setminus{}(K'_{\tilde{S}}\cup{}K'_{0})\). From
  \autoref{obs:split-approx-partition-properties} we get that
  $|K'_{\tilde{S}}| \le 2k$ and $|I'_{\tilde{S}}| \le k$
  hold. % Let
  % $K'_0, K'_1$ be the partition as defined after
  % \autoref{lemma:proof-rr-large-approx-solution-no}.
  Since \autoref{rr:for-K0} is not applicable, we get that
  $|K_0| < 2 |I'_{\tilde{S}}| < 2k$ holds. Similarly, since
  Reduction \autoref{rr:for-K1} is not applicable we get that
  $|K_1| < 2|\tilde{S}| < 6k$ holds. Since $K'_{\tilde{S}}, K_0$
  and $K_1$ form a partition of $K$, we get that the cardinality
  of $K$ is upper bounded by $10k$.
  % We conclude this proof by stating that the running time of
  % entire algorithm is polynomial in size of input. Every Reduction
  % rule either finds some structure in graph or applies
  % \autoref{lem:expansion-lemma},\autoref{lem:expansion-matching}
  % or \autoref{lem:extra-expansion-vertex}. Either of these process
  % can be done in polynomial time. In case of
  % \autoref{rr:large-approx-solution-no}, algorithm finds an
  % approximate solution in polynomial time as mentioned above the
  % statement.
  \qed
\end{proof}

%%% Local Variables:
%%% mode: latex
%%% TeX-master: "subsetFVS"
%%% End:

%% file: kernel_lowerbound.tex
\section{Kernel Lower Bound}
\label{sec:lowerbound}
In this section we show that the bound of
\autoref{thm:split-kernel} is essentially tight; we prove

\lowerbound*

Let $\Pi \subseteq \Sigma^*$ be any language.  A \emph{polynomial
  compression} for a parameterized problem
$Q \subseteq \Sigma^* \times \mathbb{N}$ is an algorithm
$\mathcal{C}$ that, given an instance $(I; k)$ of $Q$, runs in
time $\Oh((|I| + k)^c)$ and returns a string $y$ such that (i)
$|y| \le p(k)$ for some polynomial $p(\cdot)$, and (ii)
$y \in \Pi$ if and only if $(I; k) \in Q$. Here, $c$ is a
constant. If $|\Sigma| = 2$, the polynomial $p(\cdot)$ will be
called the bitsize of the compression.
%We define unparameterized version of a parameterized language $Q$ as a classical language $\tilde{Q} \subseteq \Sigma^*$ where parameter is appended in unary form. 
Dell and van Melkebeek \cite{Dell:lowerbound} established
following breakthrough result.
\begin{proposition} \label{prop:kernel-bound-VC} For any
  $\epsilon > 0$, the \textsc{Vertex Cover} problem parameterized
  by the solution size does not admit a polynomial compression
  with bitsize $\Oh(k^{2 - \epsilon})$, unless
  $\NP \subseteq \coNP/poly$.
\end{proposition}

%We prove a similar result for \SFVSS.

% \begin{theorem} For any $\epsilon > 0$, the \SFVSS problem parameterized by the solution size does not admit a polynomial compression with bitsize $\Oh(k^{2 - \epsilon})$, unless $\NP \subseteq \coNP/poly$. 

% \end{theorem}
%\lowerbound*
\begin{proof}[of \autoref{thm:lowerbound}] We prove a stronger
  statement, namely that \SFVSS parameterized by the solution size
  does not admit a polynomial \emph{compression} with bitsize
  $\Oh(k^{2 - \epsilon})$, to \emph{any} language, unless
  $\NP \subseteq \coNP/poly$.  Indeed, fix a language
  $\Pi \subseteq \Sigma^*$, and suppose there exists an algorithm
  $\mathcal{C}_{1}$ and a constant $\delta > 0$ such that given an
  instance $(I; k)$ of \SFVSS as input, algorithm $\mathcal{C}_1$
  outputs a string $y \in \Sigma$ of size $\Oh(k^{2 - \delta})$
  such that $(I; k)$ is a \yes instance of \SFVSS if and only if
  $y$ is a \yes instance of $\Pi$. We show how to design a
  polynomial compression for \textsc{Vertex Cover} using
  \(\mathcal{C}_{1}\), in a way which contradicts
  \autoref{prop:kernel-bound-VC}. For this we reuse a reduction
  from \textsc{Vertex Cover} to \SFVSS due to Fomin et al.:

  \noindent
  \textbf{Reduction:}~\cite[proof of
  Theorem~$2.1$]{fomin2014enumerating} Let $(G, k)$ be an instance
  of \textsc{Vertex Cover}, where $G$ is an arbitrary graph with
  $n$ vertices and $m$ edges. We construct a split graph $H$ with
  split partition $(K, I)$ as follows. $V(H) = K \dot\cup I$
  contains $n + m$ vertices: for each vertex $u \in V(G)$, there
  is a vertex $u \in K$, and for each edge $\{v,w\} \in E(G)$,
  there is a vertex $u_{vw} \in I$. The edge set $E(H)$ is defined
  so that vertices in $K$ are pairwise adjacent, and each vertex
  $u_{vw}$ of $I$ has exactly two neighbors: vertices $v$ and $w$
  in $K$. Consequently, $K$ is a clique and $I$ is an independent
  set.  Fomin et al. show that $(G; k)$ is a \yes instance of
  \textsc{Vertex Cover} if and only if $(H; k)$ is a \yes instance
  of \SFVSS.

  Our compression algorithm $\mathcal{C}_2$ for \textsc{Vertex
    Cover} works as follows. Given an instance $(G; k)$ of
  \textsc{Vertex Cover}, algorithm $\mathcal{C}_2$ applies the
  above reduction to obtain an equivalent instance $(H; k)$ of
  \SFVSS. On this instance, $\mathcal{C}_2$ runs the hypothetical
  compression algorithm $\mathcal{C}_1$ for \SFVSS, to produce a
  string $y$. By assumption, $(H; k)$ is a \yes instance of \SFVSS
  if and only if $y$ is a yes instance of $\Pi$. By
  \cite[Theorem~$2.1$]{fomin2014enumerating}, $(H; k)$ is a \yes
  instance of \SFVSS if and only if $(G; k)$ is a yes instance of
  $\Pi$. Hence $(G; k)$ is a \yes instance of \textsc{Vertex
    Cover} if and only if $y$ is a \yes instance of $\Sigma^*$,
  and the size of $y$ is at most $\Oh(k^{2 - \delta})$. Since this
  entire process can be completed in time polynimal in size of
  input $(G; k)$, this contradicts
  \autoref{prop:kernel-bound-VC}. Thus the compression algorithm
  \(\mathcal{C}_{1}\) cannot exist unless \caveat. Setting the
  language \(\Pi\) in the definition of \(\mathcal{C}_{1}\) to be
  \SFVSS itself, we get \autoref{thm:lowerbound}.
\end{proof}

%%% Local Variables:
%%% mode: latex
%%% TeX-master: "subsetFVS"
%%% End:

%% file: fpt-chordal.tex
% Version without a forbidden set
\section{An \FPT\ Algorithm  For \SFVSC}
\label{sec:fpt-algorithm}

In this section we prove \autoref{thm:chordalfpt}; we show that
the \SFVSC problem can be solved in \(\Oh(2^{k}(n+m))\) time where
\(n,m\) are the number of vertices and edges of \(G\),
respectively.  Our algorithm consists of the application of
\emph{reduction rules} and \emph{branching rules} to the input
instance. To bound the running time we assign a measure to each
instance.  Given an instance \(I\) as input, each branching rule
creates at least two new instances which have strictly smaller
measures than that of \(I\); the algorithm then solves these new
instances recursively and puts their solutions together to obtain
a solution for \(I\).  Consider an application of branching rule
\texttt{BR} to an instance \(I\) whose measure is $k$. Let $r>1$
be a positive integer, and let $t_i$ be a positive real number for
$i\in\{1, 2, \dots, r\}$. Suppose rule \texttt{BR}, when applied
to instance \(I\), creates $r$ new instances
$I_1, I_2, \dots, I_r$ with measures
$k - t_1, k - t_2, \dots, k - t_r$, respectively.  We say that the
branching rule \texttt{BR} is \emph{exhaustive} if $I$ is a \yes
instance if and only if at least one of $I_1, I_2, \dots, I_r$ is
a \yes instance.  We say that $(t_1, t_2, \dots, t_r)$ is the
\emph{branching vector} corresponding to rule \texttt{BR}. The
contribution of branching rule \texttt{BR} to the running time of
the algorithm is $\mathcal{O}^*(\alpha^k)$, where $\alpha$ is the
unique positive real root of
$x^{k} - x^{k - t_1} - x^{k - t_2} \cdots - x^{k - t_r} =
0$~\cite[Theorem~2.1]{fomin2010exact}. % We use $\mathcal{O}^*$
% notation to suppress functions which are polynomial in input
% size.  If algorithm uses more than one branching algorithm then
% the running time of the algorithm is dominated by running time
% of branching rule for which this $\alpha$ value is highest.

We first give an informal description of the algorithm. At a high
level, the algorithm proceeds by branching on the vertices of a
clique in the input graph. We use properties of chordal graphs to
ensure that we can always find a clique to branch on, and that
this can be done with branching vectors%\footnote{See
%  \cite{fomin2010exact} for definition of Branching Vectors.}
which keep the running time within \(\OhStar(2^{k})\).

%\todo{Define simplicial clique}
We apply reduction rules to ensure that every
simplicial clique  in the graph has
size at least three, and that every
vertex in the graph is part of a \(T\)-triangle.
%\todo{Explicitly define \(T\)-triangle in the prelims.}.
The latter condition implies that every simplicial clique contains at least one terminal vertex. 
% (ii) every non-terminal vertex has at least two terminals as
% neighbours\todo{Homogenize spelling.}, and (iii) every
% simplicial clique\todo{Define
% this.}  has at least three vertices, at least one of which is
% a terminal.
Let \(Q=\{s,a,b\}\) be a simplicial clique in the graph, where
\(s\) is a simplicial vertex. The algorithm branches by picking
either of the two vertices \(a,b\) into the solution. This is
safe\footnote{Precise arguments follow.} since \(s,a,b\) is the
only \(T\)-triangle containing \(s\). The branching vector is
\((1,1)\) and resolves to a running time of
\(\OhStar(2^{k})\). This branching rule is applied whenever there
is a simplicial clique of size three in the graph.

Let \(Q=\{t,a,b,c,d\}\subseteq{}V(G)\) be a---not necessarily
simplicial---clique of size five in \(G\) which contains a
terminal \(t\). If a solution does not contain \(t\), then it must
contain at least three of the vertices \(\{a,b,c,d\}\). This is
because each pair of these four vertices forms a \(T\)-triangle
with \(t\). This implies that the three-way branch
\(\{\{t\},\{a,b\},\{c,d\}\}\) is correct and
complete.\footnote{For a given instance $(G, T, k)$ each branch
  creates new (``smaller'') instances and solves them
  recursively. In this case the three branches would be
  $(G - \{t\}, T \setminus \{t\}, k - 1), (G - \{a, b\}, T
  \setminus \{a, b\}, k - 2)$ and
  $(G - \{a, b\}, T \setminus \{c, d\}, k - 2)$.} This branching
has the branching vector \((1,2,2)\) which resolves to a running
time of \(\OhStar(2^{k})\).  Note that if clique \(Q\) contains
more than five vertices (including a terminal \(t\)) then we can
apply this branching to a sub-clique
\(\{t,a,b,c,d\}\subsetneq{}Q\). If there is a simplicial clique of
size \emph{at least} five in the graph then the algorithm applies
this branching.

We now come to the case where every simplicial clique in the graph
has size exactly four. Simple branching rules now fail to give a
running time within the \(\OhStar(2^{k})\) bound. We argue if the
graph has more than eight vertices at this point then we can find
three simplicial cliques with certain properties. We design a
branching rule which applies to this structure and gives a total
running time of \(\OhStar(2^{k})\). 

% a terminal \(t\) and four cliques \(Q,Q_{1},Q_{2},Q_{3}\) where
% (i) \(Q=\{t,x,y,z\},Q_{1},Q_{2}\) are simplicial cliques of size
% four each and \(Q_{3}\) is a maximal clique of size at least
% three, and (ii) \(Q_{1},Q_{2},Q_{3}\) contain distinct terminals
% \(t_{1},t_{2},t_{3}\), respectively, and intersect \(Q\) in
% exactly \(\{x\},\{y\},\{z\}\), respectively. See the
% figure\todo{Add a schematic} for an illustration. We show that
% we can use this structure to design branching rules which take
% care of this case as well in time \(\OhStar(2^{k})\). We
% describe the algorithm in more detail in the rest of this
% section.

\subsection{Reduction Rules}
%We apply these reduction rules in the given order, repeatedly
%applying each rule till it does not cause any change, before going
%on to the next rule. 

The first two rules apply at the leaves of the branching tree.
\begin{reductionrule}\label{rr:trivial-no}
  Let $(G; T; k)$ be an instance of \SFVSC. If $k \leq 0$
  and there is a $T$-triangle in \(G\) then return a trivial \NO
  instance.
\end{reductionrule}

\begin{reductionrule}\label{rr:trivial-yes}
  Let $(G; T; k)$ be an instance of \SFVSC. If $k \geq 0$
  and there is \emph{no} $T$-triangle in \(G\) then return a
  trivial \YES instance.
\end{reductionrule}

We can
%, in polynomial\todo{Linear?} time,
get rid of connected components which are cliques.

\begin{reductionrule}\label{rr:no-more-clique-components}
  Let \((G=(V,E);T;k)\) be an instance of \SFVSC, and let
  \(Q\) be a connected component of \(G\) which is a
  clique. Delete \(Q\) from \(G\) to obtain the graph
  \(G'=G-Q\),
  %\todo{Add the notation \(G-X\) to prelims.}
  and let
  \(T'=T\setminus{}Q\). The reduced instance is computed as
  follows:
  \begin{enumerate}
  \item If \(|Q|\in\{1,2\}\), \emph{or} if \(Q\) contains no
    terminal, then the reduced instance is \((G'; T'; k)\).
  \item Else: If \(Q\) contains at most two non-terminals, then
    decrement \(k\) by \(|Q|-2\): the reduced instance is
    \((G'; T'; k-(|Q|-2))\).
  \item In the remaining case \(Q\) contains three or more
    non-terminals. Let \(t=|T\cap{}Q|\) be the number of
    \emph{terminals} in \(Q\). Decrement \(k\) by \(t\): the
    reduced instance is \((G'; T'; k-t)\). 
  \end{enumerate}
\end{reductionrule}

%\todo{Prove that this reduction rule is correct.}
\begin{lemma} \autoref{rr:no-more-clique-components} is safe.
\end{lemma}
\begin{proof} Since $G'$ is induced subgraph of $G$, any subset-FVS of $G$ is also a subset-FVS of $G'$. This implies forward direction in Case 1. For reverse direction in the same case, notice that no vertices in $Q$ is part of any $T$-triangle in $G$ and vertices in clique $Q$ are disjoint from vertices of $G'$. Hence if $S$ is subset-FVS of $G'$ then is also a subset-FVS of $G$. 

%  The proof of remaining cases is based on simple observation that in a clique we either delete all terminals or if at least one terminal is not deleted then we have to delete all but one vertex in that clique. We prove the safeness of reduction rules in this case by showing that there exists an optimum solution which contains vertices which.
  Consider clique $Q$ and an optimal subset-FVS $S$ in graph $G$.
  %Consider the case when cardinality of $Q \setminus S$ is three or more.
  In Case 2, clique $Q$ contains at most two non-terminal vertices. If cardinality of $Q \setminus S$ is three or more then vertices in $Q \setminus S$ forms a $T$-triangle which contradicts the fact that $S$ is a subset-FVS. Hence $Q \setminus S$ contains at most two vertices. By deleting any $|Q| - 2$ vertices we obtains a cycle-less graph in $Q$ and hence it is irrelevant which $|Q| - 2$ vertices are deleted. This proves the forward direction in Case 2.
  For reverse direction, note that adding an isolated edge in graph $G'$ does not change its subset-FVS. This observation along with the fact that for any subset $X$ of $V(G)$ if $S$ is a subset-FVS of $G - X$ of size $k - |X|$ then $S \cup X$ is a subset-FVS of $G$ of size $k$. 

  In Case 3, clique $Q$ contains at least three non-terminal vertices, say $y_1, y_2, y_3$. Any subset-FVS can omit exclude at most two terminals in clique $Q$. If $S$ does not contains a terminal, say $t$, then it must include all but one vertices in $Q$. Let another vertex excluded by $S$ in clique $Q$ be $x$.
  If $x$ is a terminal then $S$ includes $y_1, y_2, y_3$. In this case, $(S \setminus \{y_1, y_2, y_3\}) \cup \{t, x\}$ is a subset-FVS of strictly lesser cardinality then $S$. This contradicts the optimality of $S$. If $x$ is not a terminal then $S$ contains at least two vertices from set $\{y_1, y_2, y_3\}$. Without loss of generality, let those be $y_1$ and $y_2$. Again in this case, $(S \setminus \{y_1, y_2\}) \cup \{t\}$ is a subset-FVS of strictly lesser cardinality then $S$ which contradicts the optimality of $S$. Hence $S$ contains all the terminal vertices in clique $Q$. Since including all the terminals in $Q$ kills all $T$-cycles contained in it and $S$ is optimum, $S$ contains does not contain any non-terminal vertex in $Q$. Reverse direction in Case~$3$ is implied by correctness proof for Case~$1$ and the fact that for any subset $X$ of $V(G)$ if $S$ is a subset-FVS of $G - X$ of size $k - |X|$ then $S \cup X$ is a subset-FVS of $G$ of size $k$. 
\end{proof}

%In following reduction rule, we extended this idea to non-neighbors of terminals even when graph is not a clique. In second case, since there are at most two non-terminals, after deleting any subset of size strictly smaller than $|Q| - 2$ there is a $T$-triangle in $Q$. Hence any solution must include $|Q| - 2$ many vertices from clique $Q$. Since clique $Q$ is not adjacent with any other $T$-triangle, it does not matter which subset of $Q$ of cardinality $|Q| - 2$ is included in solution. Not including any terminal in solution forces that solution to include all but one vertex in it. In third case, since $Q$ contains at least three or more non-terminals, excluding a terminal  (two terminals) from a solution forces that solution to pick at least two (at least three) non-terminals. These non-terminal vertices in solution can be replaced with terminal vertices which are excluded. Hence there exists an optimum solution which contains all terminals in clique $Q$.

Next reduction rule  deletes any vertex which is not part of a \(T\)-triangle.

\begin{reductionrule}\label{rr:only-T-neighbours}
  Let \((G=(V,E);T;k)\) be an instance of \SFVSC, and let
  \(N\) be the set of all non terminal vertices in \(G\) which do not have any
  terminal from \(T\) as a neighbour. Delete \(N\) from \(G\) to
  obtain the graph \(G'=G-N\). The reduced instance is
  \((G';T;k)\).
\end{reductionrule}

\begin{lemma} \autoref{rr:only-T-neighbours} is safe.
\end{lemma}
\begin{proof} Since $G'$ is induced subgraph of $G$, any
  subset-FVS of $G$ is also a subset-FVS of $G'$. This implies the
  correctness of forward direction.  In reverse direction,
  consider subset-FVS $S'$ of $G'$. If $S'$ is not a subset-FVS of
  $G$ then there exists a $T$-cycle in $G \setminus S$. Since this
  $T$-cycle is not contains in $G'$, it must contain a vertex, say
  $v$, from set $N$. By \autoref{lem:triangles}, if there exists a
  $T$-cycle contains vertex $v$ then there is a $T$-triangle
  containing $v$. This implies $v$ has a neighbor in terminal set
  contradicting the fact that $v$ is in $N$.
\end{proof}

We can safely delete \emph{some} edges which are not part of any
\(T\)-triangle, to get a graph with no ``tiny'' maximal cliques.

\begin{reductionrule}\label{rr:no-cut-edges}
  Let \((G=(V,E);T;k)\) be an instance of \SFVSC, and let
  \(e=\{u,v\}\) be a bridge in \(G\). Delete the edge \(e\) to
  get the graph \(G'=(V,E\setminus\{e\})\). The reduced instance
  is \((G'; T; k)\). 
\end{reductionrule}

\begin{lemma} \autoref{rr:no-cut-edges} is safe.
\end{lemma}
\begin{proof} %Since a bridge $e$ is not a part of any cycle in $G$, there is unique path between its endpoints. Deleting edge $e$ does not create any cycles in $G$ and hence $G'$ is also a chordal graph.
  By \autoref{lem:bridge-deletion}, $G'$ is a chordal graph.
  Since $G'$ is a subgraph of $G$, any subset-FVS of $G$ is also a subset-FVS of $G'$. This proves the safeness of forward direction of reduction rule. In graph $G'$, vertices $u, v$ are different connected components. If $S'$ is a subset-FVS of $G'$ which does not contain either of $u$ or $v$ then these two vertices are in different connected components of $G' - S'$. Hence adding an edge $e = uv$ in graph $G' - S'$ does not add more cycle. If $S'$ contains either $u$ or $v$ then graphs $G' - S'$ and $G - S'$ are identical. Hence in either case, $S'$ is a subset-FVS of $G$.  
\end{proof}

\begin{lemma}\label{lem:no-small-cliques}
  Let \((G;T;k)\) be an instance of \SFVSC which is reduced with
  respect to
  \Autoref{rr:no-more-clique-components,rr:no-cut-edges}.  Then
  every maximal clique in \(G\) is of size at least three.
\end{lemma}
\begin{proof}
  \autoref{rr:no-more-clique-components} ensures that for any
  maximal clique \(Q\) in the reduced graph \(G\), some vertex in
  \(Q\) has a neighbour which is not in \(Q\). Thus every maximal
  clique in \(Q\) has size at least two.  Let \(Q=\{u,v\}\) be a
  maximal clique of size two in \(G\). Then there is a third
  vertex \(w\) in \(G\) such that either \(\{u,w\}\) or
  \(\{v,w\}\) is an edge in \(G\). If \emph{both} these edges are
  present in \(G\) then \(\{u,v,w\}\) is a clique, contradicting
  the maximality of \(Q\).  So let \(\{u,w\}\) be a non-edge in
  \(G\).

  \autoref{rr:no-more-clique-components} ensures that \(G\) has at
  least three vertices, and hence \autoref{rr:no-cut-edges} ensures
  that \(G\) has no cut vertex\footnote{Exercise 2.3.1 in ``Graph
    Theory with Applications'' by Bondy and Murty}. \(G\) is thus
  2-vertex-connected. In particular, there is a path from \(u\) to
  \(w\) which avoids the vertex \(v\); let \(P\) be a shortest
  such path. Since \(G\) is chordal
  %\todo{Mention in prelims that deleting vertices or cut edges does not make \(G\) non-chordal.}
  and since \(\{u,v,w\}\) is a path in \(G\), path
  \(P\) cannot be of length two or more. Thus \(P\) is a single
  edge \(\{u,w\}\), a contradiction.
\end{proof}

At this point, every vertex in the graph is part of at least one
\(T\)-triangle.
\begin{lemma}\label{lem:all-in-T-triangles}
  Let \((G;T;k)\) be an instance of \SFVSC which is reduced with
  respect to
  \Autoref{rr:no-more-clique-components,rr:no-cut-edges,rr:only-T-neighbours}.
  Then every vertex in \(G\) is part of a \(T\)-triangle.
\end{lemma}
\begin{proof}
  Let \(v\) be a vertex in \(G\) and let \(t\) be a terminal which
  is adjacent to \(v\).
  %\todo{Define adjacent in prelims.}
  Let \(Q\) be a maximal clique in \(G\) which contains the edge
  \(\{v,t\}\); such a clique exists, because the edge \(\{v,t\}\)
  is a clique by itself. Clique \(Q\) contains at least one other
  vertex \(u\), and \(\{u,v,t\}\) form a \(T\)-triangle which
  contains \(v\). 
\end{proof}

Consider a simplicial vertex $v$ in reduced graph $G$. Either
vertex $v$ is a terminal or by \autoref{lem:all-in-T-triangles},
it is adjacent with some terminal. In either case, simplicial
clique which containing $v$ contains a terminal. Combining this
fact with \autoref{lem:no-small-cliques}, we get following
Corollary.

\begin{corollary}\label{cor:simplicial-terminal}
  Let \((G;T;k)\) be an instance of \SFVSC which is reduced with
  respect to
  \Autoref{rr:no-more-clique-components,rr:only-T-neighbours,rr:no-cut-edges},
  and let \(Q\) be a simplicial clique in \(G\). Then \(Q\) has
  size at least three, and contains a terminal vertex.
\end{corollary}
%\todo{Add a proof. Mention that simplicial cliques are maximal.}

It remains to argue that these reduction rules can be applied in polynomial time. In the following Lemma we prove a stronger statement. We prove that not only each of this reduction rules can be applied in linear time, but we can apply all these rules exhaustively in linear time.  
  
\begin{lemma}\label{lemma:running-time-rr} Let \((G;T;k)\) be an
  instance of \SFVSC. We can exhaustively apply
  \Autoref{rr:no-more-clique-components,rr:only-T-neighbours,rr:no-cut-edges}
  on this instance in time $\mathcal{O}(n + m)$. Here $n, m$ are
  number of vertices and edges in input graph $G$.
\end{lemma}
\begin{proof}
  Let $\mathcal{C}$ be the set of all maximal clique of input graph $G$. The cardinality of set $\mathcal{C}$ is at most $n$ and we can compute this set in time $\Oh(n + m)$ (\cite[Theorem 4.17]{golumbic2004algorithmic}). While constructing set $\mathcal{C}$, we can mark all cliques in $C$ which contains terminal in it. We define \emph{strong neighbors} of $T$ as set of non-terminal vertices which are present in maximal clique containing a terminal and is of size at least three.
% In other words, strong-neighbors of $T$ is set $\{v |\ \exists\ t \in T \text{ such that } vt \in E(G) \text{ and it is not a bridge} \}$.
  Given a chordal graph $G$ and set $\mathcal{C}$, one can mark
  all strong neighbors of $T$ in time $\Oh(n + m)$. Moreover, all
  maximal cliques in $\mathcal{C}$ of size two are either bridges
  or isolated cliques.  Given graph $G$, algorithm computes set
  $\mathcal{C}$ and performs following steps. Step~$1:$ Mark all
  strong neighbors of $T$ and delete remaining vertices. Step~$2:$
  Find and delete all bridges in graph $G$. Step~$3:$ Delete
  isolated cliques in graph according to
  \autoref{rr:no-more-clique-components}.  One iteration of three
  steps can be performed in time $\Oh(n + m)$. Next, we argue that
  only one iteration is sufficient to get a non-reducible
  instance.

Consider a graph $G'$ obtained by above process. We first argue that every non-terminal vertex in $G'$ is strong neighbor (and hence neighbor) of $T$. Suppose not. Consider a non-terminal vertex $v$ in $G'$ which is not a strong neighbor of $T$. Since vertex $v$ was not deleted in Step~1, $v$ was a strong vertex in $G$. This implies, $v$ is part of maximal clique $Q$ of size at least $3$. Since we only delete bridge edges in Step~$2$, no edge in $Q$ has been deleted. As $v$ is not deleted in Step~3, $v$ is not a part of isolated clique. This implies $Q$ is not an isolated clique and hence no vertex in $Q$ has been deleted in Step~3. Hence $Q$ is present in graph $G'$. Since $Q$ contains a terminal, vertex $v$ and is of size at least there, this contradicts the fact that $v$ is not a strong neighborhood of $T$. All bridges in graph $G$ has been deleted in Step~$2$. Step~$3$ of the algorithm deletes an isolated clique which satisfies certain criteria. Hence if there is a bridge in $G'$ then the same bridge is present in $G$ at the end of Step~$2$. Hence $G'$ contains no bridge. Graph $G'$ also does not contain any isolated cliques as all such cliques were deleted in Step~$3$.    
\end{proof}

\subsection{Simple Branching Rules}
Our first branching rule ensures that every non-terminal vertex
has at least two terminals in its neighbourhood. So let \(v\) be a
non-terminal vertex which has exactly one terminal \(t\) in its
neighbourhood. Let \(x\) be a vertex---as guaranteed to exist by
\autoref{lem:all-in-T-triangles}---such that \(\{v,t,x\}\) is a
\(T\)-triangle containing \(v\). By our assumption \emph{every}
\(T\)-triangle which contains \(v\) also contains \(t\).  A
solution \(S\) as called as \emph{target solution} if
\(|S|\leq{}k\).  It follows that if \(S\) is a target solution
then \((S\setminus{}\{v\})\cup\{t\}\) is also a target
solution. This means that we may assume that if there exists a
target solution then there exists one which does not contain
\(v\). Hence it is safe to branch on any two neighbours of \(v\)
which form a \(T\)-triangle containing \(v\).

\begin{branchingrule}\label{br:non-terminals-two-terminal-nbrs}
  Let \((G;T;k)\) be an instance of \SFVSC, let \(v\) be a
  non-terminal vertex which has exactly one terminal neighbour
  \(t\), and let \(\{v,t,x\}\) be a \(T\)-triangle containing
  \(v\). Let \(G_{1}=G-\{t\},G_{2}=G-\{x\}\) and
  \(T_{1}=T\setminus\{t\},T_{2}=T\setminus\{x\}\). The new
  instances are: \((G_{1};T_{1};k-1),(G_{2};T_{2};k-1)\).
\end{branchingrule}

\begin{lemma} \autoref{br:non-terminals-two-terminal-nbrs} is
  exhaustive and it can be executed in time $\Oh(n + m)$.
\end{lemma}
\begin{proof} Consider an instance $(G; T; k)$ of \SFVSC and let
  $(G_{1};T_{1};k-1)$ and $(G_{2};T_{2};k-1)$ be two instances
  produced by \autoref{br:non-terminals-two-terminal-nbrs} when
  applied on $(G; T; k)$.  We argue that $(G;T;k)$ is a \yes
  instance if and only if either $(G_{1};T_{1};k-1)$ or
  $(G_{2};T_{2};k-1)$ is a \yes instance. Let $v, t, x$ be
  vertices in $G$ as specified in the statement of branching rule.

  $(\Rightarrow)$ Since $\{v, t, x\}$ is a $T$-triangle in $G$, any target solution of $(G; T; k)$ contains at least one vertex among them. Consider a target solution $S$ which contains $t$. Set $S \setminus \{t\}$ is a solution of $(G - \{t\}; T_1; k - 1)$. By applying similar argument in the case when $S$ contains $x$ implies that $(G - \{x\}; T_2; k - 1)$ is an \yes instance. We now consider the case when $S$ contains $v$. Note that any $T$-cycle passing through $v$ contains terminal $t$. Hence if $S$ is a target solution contains $v$ then $(S \setminus \{v\}) \cup \{t\}$ is also a target solution. This implies if $(G; T; k)$ is a \yes instance then either $(G_{1};T_{1};k-1)$ or $(G_{2};T_{2};k-1)$ is also a \yes instance.

  $(\Leftarrow)$  The proof of reverse direction follows from the fact that for any $U \subseteq V(G)$ if $S'$ is a target solution of $(G - U; T \setminus U; k - |U|)$ then $S' \cup U$ is a target solution of $(G; T; k)$.

  To apply this reduction rule, we need to find a vertex which is adjacent with exactly one terminal. If such vertex exists then it can be found in time $\Oh(n + m)$.
\end{proof}

The rest of our branching rules apply to simplicial cliques. As
noted above, this is more or less straightforward when we have a
simplicial clique of size either three, or at least five. Let
\(v\) be a simplicial vertex in graph \(G\) and \(\{v,a,b\}\) its
simplicial clique. By \autoref{cor:simplicial-terminal} at least
one of \(\{v,a,b\}\) is a terminal, and so every solution must
contain at least one of \(\{v,a,b\}\). Vertex \(v\) is not part of
any triangle in the graphs \(G-\{a\}\)
%\todo{Define this notation in prelims.}
and \(G-\{b\}\).  It follows that if a solution \(S\) contains
\(v\) then both \((S\setminus\{v\})\cup\{a\}\) and
\((S\setminus\{v\})\cup\{b\}\) are solutions of the same size or
smaller. Thus we get that there is an optimal solution which does
not contain \(v\), and contains at least one of \(\{a,b\}\). This
gives us a two-way branching rule for such cliques.

\begin{branchingrule}\label{br:branch-on-three-cliques}
  Let \((G;T;k)\) be an instance of \SFVSC, and let \(\{v,a,b\}\)
  be a simplicial clique with \(v\) being the simplicial
  vertex. Let \(G_{1}=G-\{v,a\},G_{2}=G-\{v,b\}\) and
  \(T_{1}=T\setminus\{v,a\},T_{2}=T\setminus\{v,b\}\). The new
  instances are: \((G_{1};T_{1};k-1),(G_{2};T_{2};k-1)\).
\end{branchingrule}

\begin{lemma} \autoref{br:branch-on-three-cliques} is exhaustive
  and it can be executed in time $\Oh(n + m)$.
  \end{lemma}
  \begin{proof} Consider an instance $(G; T; k)$ of \SFVSC and let $(G_{1};T_{1};k-1)$ and $(G_{2};T_{2};k-1)$ be two instance produced by \autoref{br:branch-on-three-cliques} when applied on $(G; T; k)$.
    Let $v$ be a simplicial vertex of degree two and $a, b$ be its neighbors.

    $(\Rightarrow)$ By \autoref{cor:simplicial-terminal} at least
    one of \(\{v,a,b\}\) is a terminal.  Since $\{v, a, b\}$ is a
    $T$-triangle in $G$, any target solution of $(G; T; k)$
    contains at least one vertex among them. Consider a target
    solution $S$ which contains $a$. Set $S \setminus \{a\}$ is a
    solution of $(G - \{a\}; T_1; k - 1)$. In graph $G - \{a\}$,
    vertex $v$ is adjacent with only
    $b$. \Autoref{rr:no-cut-edges,rr:no-more-clique-components}
    applied to $G - a$ will delete the edge incident of $v$ and
    then vertex $v$. Hence $(G - \{a\}; T_1; k - 1)$ is a \yes
    instance if and only if $(G - \{a, v\}; T_1; k - 1)$ is a \yes
    instance. By applying similar argument in the case when $S$
    contains $b$ implies that if a target solution $S$ contains
    $b$ then $(G - \{b, v\}; T_1; k - 1)$ is a \yes instance. We
    now consider the case when $S$ contains $v$. We claim that we
    can construct another solution which excludes $v$ and is of
    cardinality at most $S$. Since $v$ is part of exactly one
    $T$-triangles, namely, $\{v, a, b\}$, every $T$-triangle
    passing through $v$ can be hit by picking either $a$ or
    $b$. Hence, if a target solution $S$ contains $v$ then
    $(S \setminus \{v\}) \cup \{a\}$ is also a target
    solution. This implies if $(G; T; k)$ is a \yes instance then
    either $(G_{1};T_{1};k-1)$ or $(G_{2};T_{2};k-1)$ is also a
    \yes instance.

    $(\Leftarrow)$ If $S'$ is a target solution of $(G \setminus \{v,a\}; T \setminus \{v, a\}; k - 1)$ then $S'\cup \{a\}$ is a solution of $G - v$ of size at most $k$. Since vertex $v$ is adjacent with only $b$ in graph $G - (S' \cup \{a\})$, set $S' \cup \{a\}$ is a solution of $G$ as well. This implies that $(G; T; k)$ is a \yes instance. By applying similar argument in case of $(G \setminus \{v,b\}; T \setminus \{v, b\}; k - 1)$, we complete the proof of the reverse direction.

    To apply this reduction rule, algorithm needs to find simplicial vertex of degree two if one exists. It is easy to check all vertices of degree two whether or not their neighborhood is a clique. Hence one can find a simplicial clique to branch on or conclude that no such clique exits in time $\Oh(n + m)$.
  \end{proof}

Now let \(Q\) be a clique in \(G\) of size at least
five, and let \(v\) be its simplicial vertex. By
\autoref{cor:simplicial-terminal} \(Q\) contains a terminal,
say \(t\). We can thus---as argued at the beginning of this
section---safely branch in the following manner.

\begin{branchingrule}\label{br:branch-on-large-cliques}
  Let \((G;T;k)\) be an instance of \SFVSC, \(Q\) be a clique of size at least five in \(G\), and \(t\) a terminal in
  \(Q\).  Let \(\{a,b,c,d\}\) be four other vertices of \(Q\). Let
  \(G_{1}=G-\{t\},G_{2}=G-\{a,b\},G_{3}=G-\{c,d\} \) and
  \(T_{1}=T\setminus\{t\},T_{3}=T\setminus\{a,b\},T_{3}=T\setminus\{c,d\}\). The
  new instances are:
  \((G_{1};T_{1};k-1),(G_{2};T_{2};k-2),(G_{3};T_{3};k-2)\).
\end{branchingrule}

\begin{lemma} \autoref{br:branch-on-large-cliques} is exhaustive
  and it can be executed in time $\Oh(n + m)$.
  \end{lemma}
  \begin{proof}
    Consider an instance $(G; T; k)$ of \SFVSC and let
    $(G_{1};T_{1};k-1)$, $(G_{2};T_{2};k-2)$, and
    $(G_{2};T_{2};k-2)$ be three instances produced by
    \autoref{br:branch-on-large-cliques} when applied on
    $(G; T; k)$.  Let $t$ be a terminal in maximal clique of size
    at least five and vertices $a, b, c, d$ be any of its four
    neighbors.

    $(\Rightarrow)$ Consider a target solution $S$ for $(G, T, k)$. We consider following cases: $S$ includes the terminal $t$; $S$ excludes $t$ but contains \emph{both} $a, b$ and third case is when $S$ excludes $t$ and dose not contains \emph{both} $a, b$. Since $\{t, a, b\}$ is a $T$-triangle in $G$, solution $S$ can not exclude all of $a, b, t$. Hence in third case, $S$ includes at least one of $a, b$. If target solution $S$ contains $t$ then set $S \setminus \{t\}$ is a solution of $(G - \{t\}, T_1, k - 1)$. Similarly, if target solution $S$ contains $a, b$ then set $S \setminus \{a, b\}$ is a solution of $(G - \{a, b\}, T_2, k - 2)$. Consider third case mentioned above. Without loss of generality, assume that $S$ excludes $b$. This implies $S$ excludes both $t$ and $b$. Since both $\{t, b, c\}$ and $\{t, b, d\}$ are $T$-triangles, $S$ must include both $c, d$. By applying similar argument as in previous case, set $S \setminus \{c, d\}$ is a solution of $(G - \{c, d\}, T_2, k - 2)$. This implies if $(G, T, k)$ is a \yes instance then either $(G_{1};T_{1};k-1)$ or $(G_{2};T_{2};k-2)$ or $(G_{3};T_{3};k-2)$ is also a \yes instance.

    $(\Leftarrow)$ The proof of reverse direction follows from the fact that for any $U \subseteq V(G)$ if $S'$ is a target solution of $(G - U; T \setminus U; k - |U|)$ then $S' \cup U$ is a target solution of $(G; T; k)$.

    To apply the branching rule, we need to find a maximal clique of size at least five which contains a terminal if one exists. One can enumerate all maximal cliques in given chordal graph in time $\Oh(n + m)$ and hence can find desired clique or conclude that no such clique exists.
  \end{proof}

\begin{observation}\label{obs:simplicial-sizes-exactly-four}
  Let \((G;T;k)\) be an instance which is reduced with respect to
  all the reduction rules in this section, and to which neither of
  the above branching rules applies. Then every simplicial clique
  in \(G\) is of size exactly four.
\end{observation}

We now show that we can ensure, within the \(\OhStar(2^{k})\) time
bound, that every simplicial vertex in \(G\) is a terminal as
well.\footnote{The converse may not hold.} For this we need a
structural result. %\todo{Phrase this better.}

\begin{figure}[t]
\centering
%\begin{subfigure}{.5\textwidth}
\subfloat[Refer to \autoref{lem:unique-triangle-witness}]{
\label{fig:lemma-sub1}
%  \centering
  \includegraphics[scale=0.2]{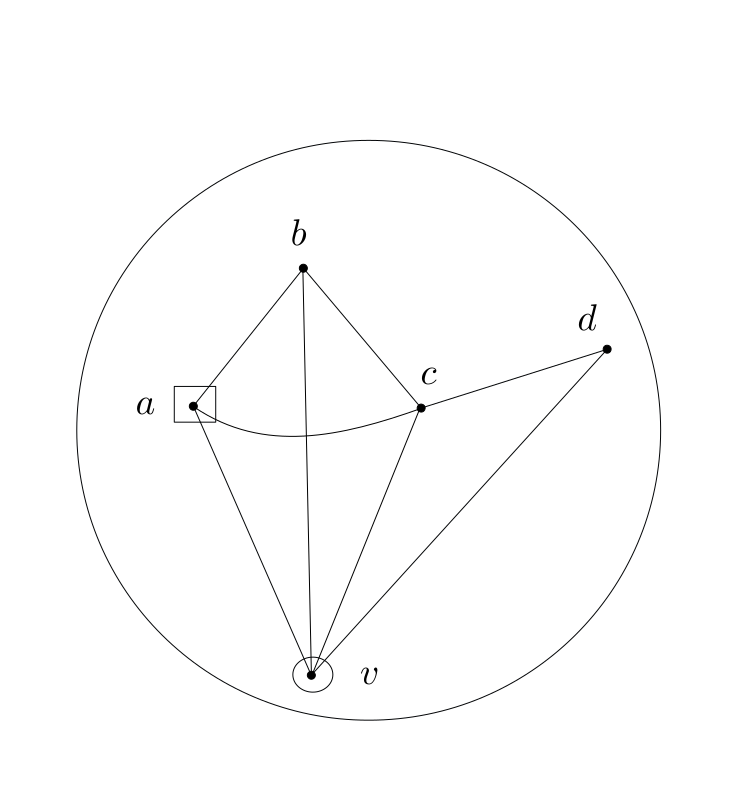}}
%  \caption{Refer to Lemma~\ref{lem:unique-triangle-witness}}
%  \label{fig:lemma-sub1}
%\end{subfigure}%
%\begin{subfigure}{.5\textwidth}
%  \centering
\subfloat[Refer to \autoref{lem:exclude-non-terminal}]{
  \label{fig:lemma-sub2}
  \includegraphics[scale=0.2]{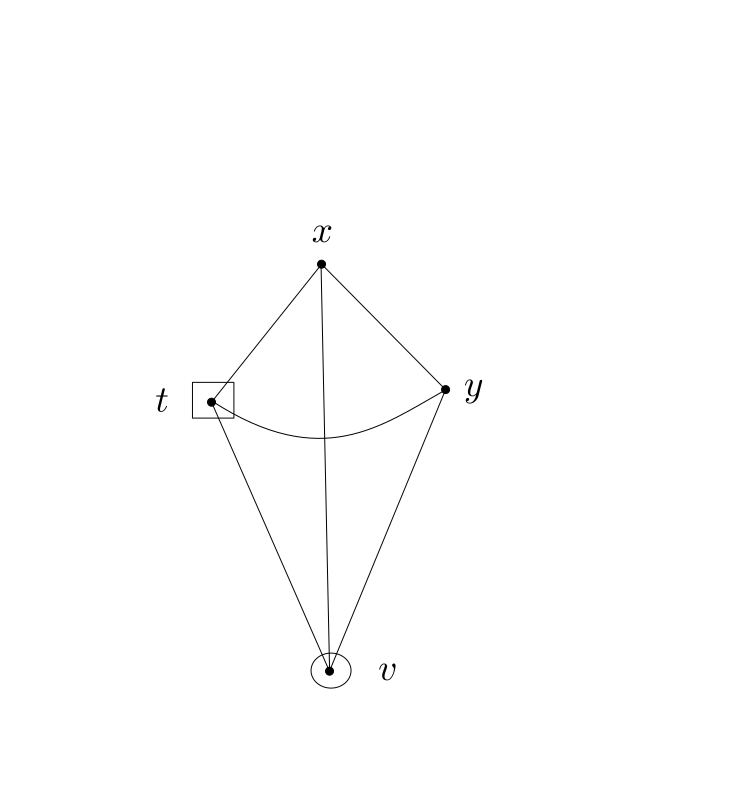}}
%  \caption{Refer to Lemma~\ref{lem:exclude-non-terminal}}
%  \label{fig:lemma-sub2}
%\end{subfigure}
\caption{In both figures, square around the vertex denotes that it
  is a terminal while circle denotes that the vertex is included
  in solution $S$. In \autoref{fig:lemma-sub1}, outer circle
  represents the maximal clique $Q$. Not all the vertices and
  edges in $Q$ are shown in the figure.}
\label{fig:lemma}
\end{figure}

\begin{lemma}\label{lem:unique-triangle-witness}
  Let \((G;T;k)\) be an instance of \SFVSC, and let \(S\) be a
  \emph{minimal} solution of \((G;T;k)\).  Let \(v\) be a
  \emph{non-terminal} vertex in \(S\) for which there exists a
  clique \(Q\) in \(G\) such that every \(T\)-triangle which
  contains \(v\) is also contained in \(Q\). Then there is a
  \emph{unique} \(T\)-triangle \(\triangle_{v}\) such that \(v\)
  is the only vertex in \(S\) which is present in
  \(\triangle_{v}\).
%  \todo{Move this before first branching rule?}
\end{lemma}
\begin{proof}
  If every \(T\)-triangle in which \(v\) is present also contains
  another vertex from \(S\) then \(S\setminus\{v\}\) is a solution
  as well, contradicting the minimality of \(S\). So there is at
  least one \(T\)-triangle \(\triangle_{v}\) such that \(v\)
  is the only vertex in \(S\) which is present in
  \(\triangle_{v}\).  %\todo{Add an illustration.-1}

  Now suppose there are two \(T\)-triangles
  \(\triangle_{1}=\{v,a,b\},\triangle_{1}=\{v,c,d\}\) such that
  \(v\) is the only vertex in \(S\) which is present in each of
  \(\triangle_{1},\triangle_{2}\) (See
  \autoref{fig:lemma-sub1}). Now triangles
  \(\triangle_{1},\triangle_{2}\) are distinct, each contains a
  terminal vertex, and \(v\) is not a terminal vertex. So the
  collection \((v,a,b,c,d)\) contains at least four distinct
  vertices, say \(\{v,a,b,c\}\), of which at least one vertex, say
  \(a\), is a terminal vertex. Since both these triangles are part
  of the clique \(Q\), the four vertices \(\{v,a,b,c\}\) form a
  clique. Deleting \(v\) from this clique leaves the triangle
  \(\{a,b,c\}\) which (i) contains the terminal \(a\), and (ii)
  contains no vertex from the set \(S\). Thus the purported
  solution \(S\) does not intersect the \(T\)-triangle
  \(\{a,b,c\}\), a contradiction.
\end{proof}

Since every neighbour of a simplicial vertex is contained in its
simplicial clique, we get

\begin{corollary}\label{cor:simplicial-nonterminal-unique-triangle-witness}
  Let \((G;T;k)\) be an instance of \SFVSC, and let \(S\) be a
  \emph{minimal} solution of \((G;T;k)\).  Let \(v\) be a
  \emph{non-terminal, simplicial} vertex in \(S\). Then there is a
  \emph{unique} \(T\)-triangle \(\triangle_{v}\) such that \(v\)
  is the only vertex in \(S\) which is present in
  \(\triangle_{v}\).
\end{corollary}

This implies that we can safely get rid of \emph{non-terminal,
  simplicial} vertices.

\begin{lemma}\label{lem:exclude-non-terminal}
  Let \((G;T;k)\) be an instance of \SFVSC which is reduced with
  respect to the reduction and branching rules described so far. Let \(v\) be non-terminal vertex which is simplicial in
  \(G\). If \((G;T;k)\) has a target solution then it has a
  target solution which does not contain \(v\).
\end{lemma}
\begin{proof}
  By \autoref{obs:simplicial-sizes-exactly-four}, every
  simplicial clique in \(G\) has size exactly four.  Let \(S\) be
  a target solution of \((G;T;k)\) which contains \(v\), and let
  \(Q=\{v,t,x,y\}\) be the simplicial clique of \(v\) where \(t\)
  is a terminal vertex.

  From
  \autoref{cor:simplicial-nonterminal-unique-triangle-witness} we
  get that there is a unique \(T\)-triangle \(\triangle_{v}\) such
  that \(v\) is the only vertex from \(S\) which is in
  \(\triangle_{v}\). Without loss of generality let
  \(\triangle_{v}\) be \(\{v,t,x\}\) (See
  \autoref{fig:lemma-sub2}). Then neither \(t\) nor \(x\) is in
  \(S\). Now \(\{t,x,y\}\) is a \(T\)-triangle as well, so we get
  that vertex \(y\) is in \(S\). If we delete the vertex set
  \(S\setminus\{v\}\) from \(G\), then every remaining
  \(T\)-triangle must intersect \(v\). But the only such triangle
  is \(\triangle_{v}=\{v, t, x\}\), and so we get that
  \((S\setminus{}\{v\})\cup\{t\}\) is a solution as well.
%  \todo{Add a picture.-2}
\end{proof}

Let \(v\) be a non-terminal vertex which is simplicial, and let
\(Q=\{v,t,x,y\}\) be the simplicial clique of \(v\) where \(t\) is
a terminal vertex.  If \((G;T;k)\) is a \yes instance then we get
from \autoref{lem:exclude-non-terminal} that it has a target
solution which (i) does not contain \(v\), and (ii) contains
either \(t\) or both of \(x,y\). This implies that we can safely
delete \(v\) and branch on \(t,\{x,y\}\).

\begin{branchingrule}\label{br:branch-on-non-terminal-simplicials}
  Let \((G;T;k)\) be an instance of \SFVSC, and let
  \(\{v,t,x,y\}\) be a simplicial clique whose simplicial vertex
  \(v\) is a \emph{non-terminal}, and \(t\) is a terminal.  Let
  \(G_{1}=G-\{t\}, G_{2}=G-\{v,x,y\}\) and
  \(T_{1}=T\setminus\{t\},T_{2}=T\setminus\{x,y\}\). The new
  instances are: \((G_{1};T_{1};k-1),(G_{2};T_{2};k-2)\).
\end{branchingrule}

\begin{lemma} \autoref{br:branch-on-non-terminal-simplicials} is
  exhaustive and it can be executed in time $\Oh(n + m)$.
  \end{lemma}
  \begin{proof}
    Consider an instance $(G; T; k)$ of \SFVSC and let
    $(G_{1};T_{1};k-1)$ and $(G_{2};T_{2};k-2)$ be two instances
    produced by \autoref{br:branch-on-non-terminal-simplicials}
    when applied on $(G; T; k)$.  Let $v$ be a non-terminal
    simplicial vertex with neighbors $t, x, y$, where $t$ is a
    terminal.

    $(\Rightarrow)$ Since $\{v, t, x, y\}$ is a clique in $G$, any
    target solution of $(G, T, k)$ contains at least one vertex
    among them.  By \autoref{lem:exclude-non-terminal}, there
    exists a target solution which does not contain \(v\) and
    contains either \(t\) or both of \(x,y\). If target solution
    $S$ contains $t$ then set $S \setminus \{t\}$ is a solution of
    $(G - \{t\}; T_1; k - 1)$. If solution $S$ excludes $t$ then
    it includes both $x, y$. In this case, $S \setminus \{x, y\}$
    is a solution of $(G - \{x, y\}, T_2, k - 2)$. Notice that in
    graph $G - \{x, y\}$, vertex $v$ is adjacent with exactly one
    vertex. \Autoref{rr:no-cut-edges,rr:no-more-clique-components}
    applied to $G - \{x, y\}$ will delete the edge incident on $v$
    and then vertex $v$. Hence $(G - \{v, x, y \}; T_1; k - 1)$ is
    a \yes instance if and only if $(G_2; T_2; k - 2)$ is a \yes
    instance.  This implies if $(G, T, k)$ is a \yes instance then
    either $(G_{1};T_{1};k-1)$ or $(G_{2};T_{2};k-2)$ is also a
    \yes instance.

    $(\Leftarrow)$ We use the fact that for all $U \subseteq V(G)$ if $S'$ is a target solution of $(G - U; T \setminus U; k - |U|)$ then $S' \cup U$ is a target solution of $(G; T; k)$. In case $(G_{1};T_{1};k-1)$ is a \yes instance, reverse direction follows immediately. If $(G_{2};T_{2};k-2)$ is a \yes instance then we first note that adding a new vertex of degree one in $G_2$ does not change its subset-FVS. This completes the proof of the reverse direction of the lemma.

    To apply this reduction rule, algorithm needs to find simplicial vertex of degree four if one exists. It is easy to check all vertices of degree four whether or not their neighborhood is a clique. Hence one can find a simplicial clique to branch on or conclude that no such clique exits in time $\Oh(n + m)$.
  \end{proof}

At this point we have that every simplicial vertex is a
terminal. We now ensure that there is exactly one terminal in any
simplicial clique. We show first that if \(t\) is a simplicial
terminal vertex with exactly three neighbours \(x,y,z\), then we
can safely assume that a target solution which contains \(t\) does
not contain any of \(\{x,y,z\}\), and vice versa.

\begin{lemma}\label{lem:terminal-private-nbrhood}
  Let \((G;T;k)\) be an instance of \SFVSC. Let \(t\) be a
  terminal vertex which is simplicial, and let \(C=\{t,x,y,z\}\)
  be its simplicial clique. If \((G;T;k)\) has a target solution
  which contains two vertices from the clique \(C\), then it has
  a target solution which does not contain \(t\).
\end{lemma}
\begin{proof}
  Let \(S\) be a target solution which contains two vertices
  from \(C\). If \(S\) does not contain \(t\) then there is
  nothing more to prove. So let \(t\in{}S\).  We assume, without
  loss of generality, that \(x\) is another vertex from \(C\)
  which is in \(S\), so that \(\{t,x\}\subseteq{}S\). Let \(H\) be
  the graph obtained from \(G\) by deleting all of \(S\). Then
  \(H\) contains no \(T\)-triangle.

  Now consider the set \(S'=(S\setminus\{t\})\cup\{y\}\), which is
  not larger than \(S\). Let \(H'\) be the graph obtained by
  deleting \(S'\) from \(G\). Thus \(H'\) is the graph obtained
  from the graph \(H\) by (i) adding back \(t\) and all the edges
  \(\{\{t,s\}\in{}E(G)\;;\;s\in{}V(H)\}\), and (ii) deleting \(y\)
  from the resulting graph. Since \(t\) is the only vertex
  \emph{added} to \(H\) in this process, we get that every
  \(T\)-triangle in \(H'\) must contain vertex \(t\).

  Now since \(S'\) contains both of \(\{x,y\}\), we get that
  vertex \(t\) has degree at most one in graph \(H'\). Thus \(t\)
  it is not part of \emph{any} triangle in \(H'\), which implies
  that \(H'\) contains no \(T\)-triangle. Thus \(S'\) is a
  target solution which does not contain \(t\).
\end{proof}

\begin{corollary}\label{lem:simplicial-terminal-vertex-selection}
  Let \((G;T;k)\) be an instance of \SFVSC. Let \(t\) be a
  terminal vertex which is simplicial, and let
  \(C=\{t,x,y,z\}\) be its simplicial clique. Then the
  following are all equivalent:
  \begin{enumerate}
  \item \((G;T;k)\) has a target solution.
  \item \((G;T;k)\) has a target solution which
  \begin{itemize}
    \item does not contain \(t\), OR
    \item contains \(t\), and excludes all of \(\{x,y,z\}\).
  \end{itemize}
\item \((G;T;k)\) has a target solution which, for each
  \(w\in\{x, y, z\}\),
  \begin{itemize}
    \item does not contain \(w\), OR
    \item contains \(w\), and excludes \(t\).
  \end{itemize}
\end{enumerate}
\end{corollary} 
\begin{proof}
  We show that (\textbf{1}) and (\textbf{2}) are equivalent, and
  that (\textbf{1}) and (\textbf{3}) are equivalent. Observe that
  the directions \((\mathbf{2})\implies(\mathbf{1})\) and
  \((\mathbf{3})\implies(\mathbf{1})\) are trivially true; we now
  argue that the forward directions hold in each case.
  
  \begin{description}
  \item [\((\mathbf{1}) \implies (\mathbf{2})\)]. If \((G;T;k)\)
    has a target solution which does not contain \(t\) then there
    is nothing more to prove. So let it be the case that
    \emph{every} such solution contains \(t\). If there is such a
    solution which excludes all of \(\{x,y,z\}\) then there is
    nothing more to prove. So let it be the case that every such
    solution contains at least one of \(\{x,y,z\}\), as well. Thus
    every target solution contains two vertices from clique \(C\).
    Now using (\textbf{1}) we get that \((G;T;k)\) has a target
    solution which contains two vertices from clique
    \(C\). \autoref{lem:simplicial-terminal-vertex-selection}
    implies that \((G;T;k)\) has a target solution which does not
    contain \(t\), which contradicts our assumption.

  \item [\((\mathbf{1}) \implies (\mathbf{3})\)]. We present the
    case for \(w=x\); the other two cases follow by symmetric
    arguments. If \((G;T;k)\) has a target solution which does not
    contain \(x\) then there is nothing more to prove. So let it
    be the case that \emph{every} such solution contains \(x\). If
    there is such a solution which excludes \(t\) then there is
    nothing more to prove. So let it be the case that every such
    solution contains \(t\). Thus every target solution of
    \((G;T;k)\) contains both of \(\{t,x\}\). Now using
    (\textbf{1}) we get that \((G;T;k)\) has a target solution
    which contains two vertices---\(t\) and \(x\)---from clique
    \(C\). \autoref{lem:simplicial-terminal-vertex-selection}
    implies that \((G;T;k)\) has a target solution which does not
    contain \(t\), which contradicts our assumption.%\qedhere
  \end{description}
\end{proof}

From this we get
\begin{corollary}\label{cor:t-and-no-nbrs}
  If \(C=\{t,x,y,z\}\) is a simplicial clique of the given type
  then it is safe to assume the following:
  \begin{itemize}
  \item if there is a target solution which contains \(t\), then
    there is such a solution which contains none of the vertices
    \(\{x,y,z\}\).
  \item if there is a target solution which contains at least one
    of \(\{x,y,z\}\), then there is such a solution which does not
    contain \(t\).
  \end{itemize}
\end{corollary}

\begin{definition}\label{def:selective-solution}
  Let \((G;T;k)\) be an instance of \SFVSC to which none of the
  previous reduction or branching rules applies. Let
  \(C=\{t,x,y,z\}\) be a simplicial clique in \(G\) with \(t\)
  being its unique terminal and simplicial vertex. A subset
  \(S\subseteq{}V(G)\) of vertices of \(G\) is a \emph{selective
    solution} of the pair \(((G;T;k), C)\) if the following hold:

  \begin{itemize}
  \item \(G-S\) has no \(T\)-triangles,
  \item \(|S|\leq{}k\),
  \item If \(S\) contains \(t\) then \(S\) contains none of
    \(\{x,y,z\}\), and,
  \item If \(S\) contains one of \(\{x,y,z\}\) then \(S\) does not
    contain \(t\).
  \end{itemize}
\end{definition}

A selective solution is thus a target solution which satisfies the
conditions of \autoref{cor:t-and-no-nbrs}. From
\autoref{cor:t-and-no-nbrs} we get that there exists a target
solution if and only if there exists a selective solution. From
now on our algorithm will search \emph{only} for selective
solutions.

Now let \(\{t,x,y,z\}\) be a simplicial clique with two terminals,
say \(t,x\), where \(t\) is a simplicial vertex. If vertex \(t\)
is present in a selective solution \(S\) then \(S\) does not
contain any of \(\{x,y,z\}\). But then \(\{x,y,z\}\) forms a
\(T\)-triangle, which contradicts the fact that \(S\) is a
solution. Thus vertex \(t\) is not present in \emph{any} selective
solution.  Now since \(t\) is a terminal vertex, every selective
solution \(S\) must pick at least two vertices from \(\{x,y,z\}\):
if \(S\) does not contain \(y,z\), for instance, then
\(\{t,y,z\}\) forms a \(T\)-triangle in the graph obtained after
deleting \(S\), a contradiction. These considerations lead to the
next branching rule.

\begin{branchingrule}\label{br:simplicial-two-terminals-case}
  Let \((G;T;k)\) be an instance of \SFVSC, and let
  \(\{t,x,y,z\}\) be a simplicial clique where \(t\) and \(x\) are
  terminal vertices, and \(t\) is a simplicial vertex. Let
  \(G_{1}=G-\{x,y\},G_{2}=G-\{y,z\},G_{3}=G-\{x,z\}\) and
  \(T_{1}=T\setminus\{x,y\},T_{2}=T\setminus\{y,z\},T_{3}=T\setminus\{x,z\}\). The
  new instances are:
  \((G_{1};T_{1};k-2),(G_{2};T_{2};k-2),(G_{3};T_{3};k-2)\).
\end{branchingrule}

\begin{lemma} \autoref{br:simplicial-two-terminals-case} is
  exhaustive and it can be executed in time $\Oh(n + m)$.
  \end{lemma}
  \begin{proof}

    Consider an instance $(G; T; k)$ of \SFVSC and let
    $(G_{1};T_{1};k-2), (G_{2};T_{2};k-2)$, and
    $(G_{3};T_{3};k-2)$ be three instances produced by
    \autoref{br:simplicial-two-terminals-case} when applied on
    $(G; T; k)$.  Let $t$ be a terminal simplicial vertex with
    neighbors $x, y,$ and $z$, where $x$ is a terminal.

    $(\Rightarrow)$ From \autoref{cor:t-and-no-nbrs}, there exists
    a selective solution, say $S$. By definition of selective
    solution, if $S$ contains $t$ then $S$ contains none of
    $x, y, z$. Since $x$ is a terminal, $x, y, z$ forms a $T$
    triangle which is not intersected by $S$ which is a
    contradiction. Hence $S$ must contain at least one of
    $x, y, z$ which implies $S$ does not contain $t$. But $t$ is a
    terminal vertex and hence $S$ must contain at least two
    vertices from $x, y, z$. If $S$ contains $\{x, y\}$ then set
    $S \setminus \{x, y\}$ is a solution for
    $(G - \{x, y\}; T \setminus \{x, y\}; k - 2)$. By applying
    similar argument when $S$ contains $\{x, z\}$ and $\{y, z\}$,
    derive that if $(G; T; k)$ is a \yes instance then at least
    one of $(G_{1};T_{1};k-2), (G_{2};T_{2};k-2)$ or
    $(G_{3};T_{3};k-2)$ is a \yes instance.
    
    $(\Leftarrow)$ We use the fact that for any $U \subseteq V(G)$ if  $S'$ is a target solution of $(G - U; T \setminus U; k - |U|)$ then $S' \cup U$ is a target solution of $(G; T; k)$ to prove the reverse direction.

    To apply this reduction rule, algorithm needs to find terminal simplicial vertex of degree four which is adjacent with a terminal if one exists. It is easy to check all vertices of degree four whether or not their neighborhood is a clique and contains a terminal. Hence one can find a simplicial clique to branch on or conclude that no such clique exits in time $\Oh(n + m)$.
  \end{proof}

  Once \autoref{br:simplicial-two-terminals-case} has been applied
  exhaustively, no simplicial clique in the graph contains two or
  more terminal vertices.  Let \(\{t,x,y,z\}\) be a simplicial
  clique with \(t\) as its unique terminal (and simplicial)
  vertex, and let \(t'\neq{}t\) be a terminal which shares two
  common neighbours---say, \(x\) and \(y\)---with \(t\). Then we
  can branch on one of these common neighbours, say \(x\), as
  follows.  If \(x\) is in a selective solution \(S\) then \(t\)
  is not in \(S\).  At least one of \(\{y,z\}\) must be in \(S\),
  or else the \(T\)-triangle \(\{t,y,z\}\) will remain after
  deleting \(S\). If \(x\) is not picked in \(S\) then we branch
  on the vertex \(t\): if \(t\) is picked in \(S\) then \(x,y,z\)
  are not in \(S\). This forces the terminal \(t'\) to be in
  \(S\), since otherwise the \(T\)-triangle \(\{t',x,y\}\) will
  remain after deleting \(S\). In the remaining case neither of
  \(x,t\) is picked in \(S\), and this forces both of \(y,z\) into
  \(S\). Thus we get
   
\begin{branchingrule}\label{br:two-terminals-two-common-nbrs}
  Let \((G;T;k)\) be an instance of \SFVSC to which none of the
  previous reduction or branching rules applies, let
  \(\{t,x,y,z\}\) be a simplicial clique of \(G\) with terminal
  \(t\), and let \(t'\neq{}t\) be a terminal which has \(x,y\) as
  neighbours.  Let
  \(G_{1}=G-\{x,y\},G_{2}=G-\{y,z\},G_{3}=G-\{x,z\},G_{4}=G-\{t,t'\}\)
  and \(T_{1}=T_{2}=T_{3}=T,T_{4}=T\setminus{}\{t,t'\}\). The new
  instances are:
  \((G_{1};T_{1};k-2),(G_{2};T_{2};k-2),(G_{3};T_{3};k-2),(G_{4};T_{4};k-2)\).
\end{branchingrule}

\begin{lemma} \autoref{br:two-terminals-two-common-nbrs} is
  exhaustive and it can be executed in time $\Oh(n + m)$ on
  instance which is not reducible by any of previously mentioned
  reduction or branching rule.
  \end{lemma}
  \begin{proof} Consider an instance $(G; T; k)$ of \SFVSC and let
    $(G_{1};T_{1};k-2), (G_{2};T_{2};k-2)$, $(G_{3};T_{3};k-2),$
    and $(G_{4};T_{4};k-2)$ be four instances produced by
    \autoref{br:two-terminals-two-common-nbrs} when applied on
    $(G; T; k)$.  Let $t$ be a terminal simplicial vertex with
    neighbors $x, y,$ and $z$. Let $t'$ be another terminal which
    is adjacent with both $x$ and $y$.

    $(\Rightarrow)$ From \autoref{cor:t-and-no-nbrs}, there exists
    a selective solution, say $S$. Since $\{t, x, y, z\}$ is a
    clique containing terminal, $S$ contains at least one vertex
    of this clique. By definition of selective solution, if $S$
    contains $t$ then $S$ contains none of $x, y, z$. Since
    $\{t', x, y\}$ is a $T$-cycle, if $S$ excludes both $\{x, y\}$
    then it must include $t'$ to hit this cycle. This implies any
    selective solution which contains terminal $t$ also contains
    terminal $t'$. Consider the case when selective solution does
    not contain $t$. Since $t$ is a terminal, $S$ must contain at
    least two vertices from $\{x, y, z\}$. If $S$ contains
    $\{x, y\}$ then set $S \setminus \{x, y\}$ is a solution for
    $(G - \{x, y\}, T \setminus \{x, y\}, k - 2)$. By applying
    similar argument when $S$ contains $(x, z)$, $(y, z)$, and
    $(t, t')$, we derive that if $(G; T; k)$ is a \yes instance
    then at least one of
    $(G_{1};T_{1};k-2), (G_{2};T_{2};k-2), (G_{3};T_{3};k-2)$, or
    $(G_{4};T_{4};k-2)$ is a \yes instance.
    
    $(\Leftarrow)$ We use the fact that for any $U \subseteq V(G)$ if $S'$ is a target solution of $(G - U; T \setminus U; k - |U|)$ then $S' \cup U$ is a target solution of $(G; T; k)$ to prove the reverse direction.

    To apply this reduction rule, we run following pre-processing:
    For a given graph $G$, construct an array of size $m$. Each
    entry in the array corresponds to an edge and can store two
    terminals. Given a chordal graph $G$, compute set of all
    maximal cliques in graph. For every maximal clique $Q$, if it
    contains terminal $t'$ then for every edge in $Q$ which is not
    incident on $t'$, add terminal $t'$ to the entry in an array
    corresponding to that edge.  If there are already two
    terminals in an array corresponding to some particular edge
    then we do not add new terminals. Since the input instance
    $(G; T; k)$ is not reducible by
    \autoref{br:branch-on-large-cliques}, every maximal clique
    which contains a terminal is of size at most four. Hence
    algorithm spend constant amount of time at each maximal clique
    which contains a terminal.  Since there are at most $n$ all
    maximal cliques all of which can be computed in $\Oh(n + m)$
    time, the overall time required for this pre-processing is
    $\Oh(n + m)$. Notice that if terminal $t'$ is adjacent with
    $x, y$ and $xy$ is an edge then $t', x, y$ are contained in a
    maximal clique containing terminal $t'$. Hence for every edge
    whose end-points are adjacent with at least two terminals, two
    terminals will be stored in the array.  Once this pre-process
    is complete, algorithm check for every terminal of degree four
    whether or not it is simplicial. Suppose terminal $t$ is
    terminal and is adjacent with $x, y, z$. For every edge in set
    $\{xy, yz, zx\}$, algorithm checks whether there exists
    another terminal $t'$ which is adjacent with end-point of the
    edge using the array constructed in the pre-processing
    step. This step can be completed in constant time. At most one
    of terminals stored for $xy$ can be $t$ and hence at least one
    another terminal, if exists, will be stored in the entry
    corresponding to edge $xy$.  Hence one can find a simplicial
    clique to branch on or conclude that no such clique exits in
    time $\Oh(n + m)$.
  \end{proof}

At this point we have

\begin{observation}\label{obs:graph-after-simple-stuff}
  Let \((G;T;k)\) be an instance which is reduced with respect to
  \Autoref{rr:trivial-no,rr:trivial-yes,rr:no-more-clique-components,rr:only-T-neighbours,rr:no-cut-edges},
  and
  \Autoref{br:non-terminals-two-terminal-nbrs,br:branch-on-three-cliques,br:branch-on-large-cliques,br:branch-on-non-terminal-simplicials,br:simplicial-two-terminals-case,br:two-terminals-two-common-nbrs}. Then
  \(G\) has the following properties:
  \begin{enumerate}
  \item\label{obs:maximal-cliques-at-least-three} Every maximal
    clique is of size at least three.
  \item\label{obs:at-least-two-terminal-nbrs} Every non-terminal
    vertex has at least two terminals as neighbours.
  \item\label{obs:simplicials-are-terminals} Every simplicial vertex is a terminal.
  \item\label{obs:simplicials-size-four-unique-vertex} Every
    simplicial clique \(C\) has size exactly four and
    contains exactly one simplicial vertex, which is also the only
    terminal vertex in \(C\).
  \item\label{obs:non-terminal-pairs-unique-common-terminal-nbr}
    Each pair of non-terminals in a simplicial clique \(C\) has
    exactly one common neighbour (namely, the simplicial vertex in
    \(C\)) from among the terminal vertices.
  \end{enumerate}
\end{observation}

\subsection{Dealing With Simplicial Cliques of Size Four}
\label{sec:size-four-branching}
We now derive some structural properties of a ``reduced'' instance
\((G;T;k)\) of \SFVSC as described in
\autoref{obs:graph-after-simple-stuff}, and use these to
handle simplicial cliques of size exactly four within the required
time bound of \(\OhStar(2^{k})\).  Let \(\mathcal{T}_{G}\) be a
clique tree of graph \(G\).  Then from
\autoref{fac:clique-tree-properties} we get that every leaf of
\(\mathcal{T}_{G}\) is a simplicial clique of size exactly
four. If \(\mathcal{T}_{G}\) contains at most two nodes then both
of them are leaves, and so they are simplicial cliques of size
four each. In this case graph \(G\) has at most eight vertices and
we can solve this instance in constant time. So we assume, without
loss of generality, that \(\mathcal{T}_{G}\) has at least three
nodes. Then \(\mathcal{T}_{G}\) has at least one node which is not
a leaf; let \(C_{r}\) be such a node. We root the tree
\(\mathcal{T}_{G}\) at node \(C_{r}\).

From now on we assume that \(C_{\ell}=\{t, x, y, z\}\) is a leaf
node of \(\mathcal{T}_{G}\) which is at the maximum distance from
the root \(C_{r}\) with \(t\) being the unique terminal in
\(C_{\ell}\), and that \(C_{p}\) is the unique neighbour (``parent'')
of \(C_{\ell}\) in \(\mathcal{T}_{G}\).

\begin{claimnr}\label{claim:x-y-z-in-Cp}
  \(\{x,y,z\}\subseteq{}C_{p}\), and \(C_{p}\) does not contain a terminal vertex.
\end{claimnr}
\begin{proof}
  From \autoref{obs:graph-after-simple-stuff} we know that \(t\)
  is the only simplicial vertex in \(C_{\ell}\).  For the sake of
  contradiction, assume that \(x\) is not in \(C_{p}\). Then from
  the intersection property of clique trees we get that \(x\) is
  not present in any maximal clique apart from $C_{\ell}$. Thus
  (\autoref{fac:simplicial-vertex-characterization}) \(x\) is a
  simplicial vertex, a contradiction. Repeating this argument, we
  get \(\{x,y,z\}\subseteq{}C_{p}\).

  If \(C_{p}\) contains a terminal vertex \(t'\) then \(t'\) has
  all of \(\{x,y,z\}\) as neighbours, and since \(t\) is
  simplicial we get that \(t'\neq{}t\). This contradicts
  part~\ref{obs:non-terminal-pairs-unique-common-terminal-nbr} of
  \autoref{obs:graph-after-simple-stuff}.
\end{proof}

Parts~\ref{obs:at-least-two-terminal-nbrs} and
\ref{obs:non-terminal-pairs-unique-common-terminal-nbr} of
\autoref{obs:graph-after-simple-stuff} together imply that
each of \(x,y,z\) must have at least one terminal other than \(t\)
as a neighbour, and that these terminals must be pairwise
distinct: each such terminal is adjacent to exactly one of
\(x,y,z\). Let \(t_{x},t_{y},t_{z}\) be three such terminal
vertices which are adjacent to \(x,y,z\), respectively. From
\autoref{claim:x-y-z-in-Cp} we know that none of
\(t_{x},t_{y},t_{z}\) is in \(C_{p}\), and it is trivially the
case that none of these vertices is in \(C_{\ell}\) either.

Let \(C_{x}',C_{y}',C_{z}'\) be three maximal cliques that contain
the edges \(\{x,t_{x}\},\{y,t_{y}\},\{z,t_{z}\}\),
respectively. Then we have that clique \(C_{x}'\) does not contain
vertex \(y\), and similarly for the other pairs.  Thus the maximal
cliques \(C_{x}',C_{y}',C_{z}'\) intersect with \(C_{\ell}\)
exactly at \(\{x\},\{y\},\{z\}\), respectively. Now let
\(C_{x},C_{y},C_{z}\) be maximal cliques \emph{closest to
  \(C_{p}\) in the clique tree \(\mathcal{T}_{G}\)} such that
\(t_{i}\in{}C_{i}\) and \(C_{i}\cap{}C_{\ell}=\{i\}\) for
\(i\in\{x,y,z\}\).

\begin{figure}[t]
\centering
%\begin{subfigure}{.5\textwidth}
\subfloat[]{
%  \centering
  \label{fig:sub1}
  \includegraphics[scale=0.3]{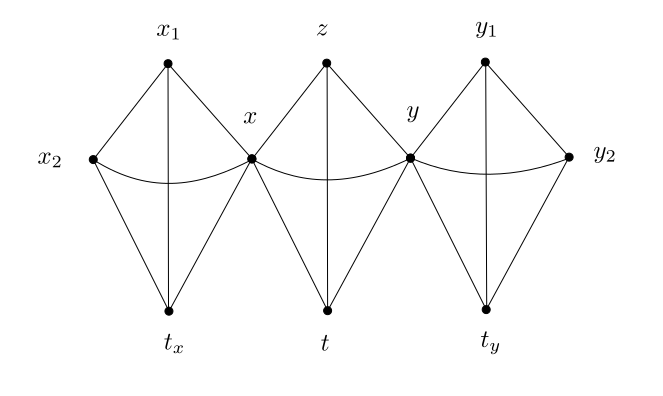}}
  %\label{fig:sub1}
%\end{subfigure}%
%\begin{subfigure}{.5\textwidth}
\subfloat[]{
  \label{fig:sub2}
%  \centering
  \includegraphics[scale=0.3]{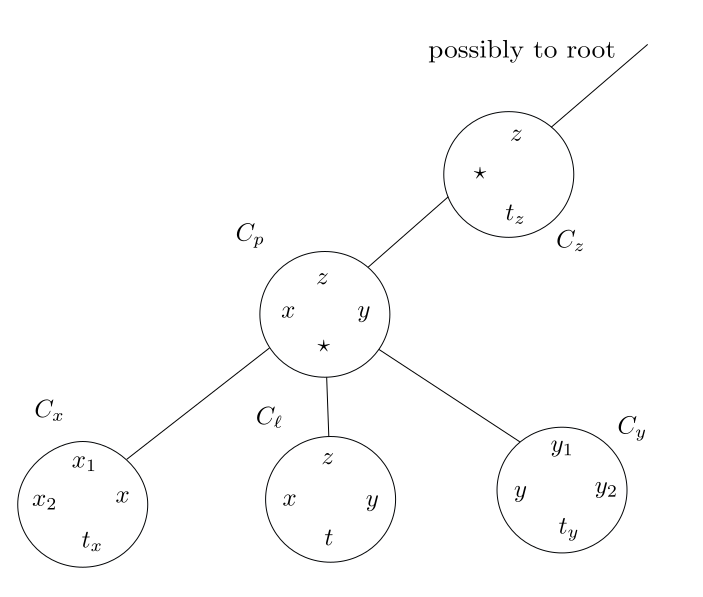}}
%  \label{fig:sub2}
%\end{subfigure}
\caption{Left side figure shows vertices and edges in the input graph. Right side shows the clique tree of given graph. Every circle is a maximal clique in input graph. $(\star)$ denotes that all vertices in maximal clique are not shown in the figure.}
\label{fig:clique-tree}
\end{figure}

%\todo{Add a picture.}
\begin{claimnr}\label{claim:three-nbrs}
  At least two among the cliques \(\{C_{x},C_{y},C_{z}\}\) are
  leaf nodes of the clique tree \(\mathcal{T}_{G}\). 
\end{claimnr}
\begin{proof}
  From item~\ref{fac:subtree} of
  \autoref{fac:clique-tree-properties} and from our assumption
  about distances, we get that all of \(C_{x},C_{y},C_{z}\) must
  be adjacent to \(C_{p}\) in the clique tree
  \(\mathcal{T}_{G}\). Since \(C_{p}\) can have at most one parent
  node in \(\mathcal{T}_{G}\), we get that at least two among
  \(C_{x},C_{y},C_{z}\), say \(C_{x}\) and \(C_{y}\), are child
  nodes of \(C_{p}\).  The child node \(C_{\ell}\) of \(C_{p}\)
  is, by assumption, a leaf which is farthest from the root node
  \(C_{r}\). It follows that neither \(C_{x}\) nor \(C_{y}\) can
  be internal nodes in \(\mathcal{T}_{G}\), or else there would be
  a leaf which is farther from \(C_{r}\) than is
  \(C_{\ell}\). Thus both \(C_{x}\) and \(C_{y}\) are leaf nodes
  of \(\mathcal{T}_{G}\).
\end{proof}

From now on we assume that \(C_{x}\) and \(C_{y}\) are leaf nodes
of \(\mathcal{T}_{G}\). By
\autoref{obs:leaves-are-simplicial} we have that \(C_{x}\)
and \(C_{y}\) are simplicial, and by
part~\ref{obs:simplicials-size-four-unique-vertex} of
\autoref{obs:graph-after-simple-stuff} we get that these
two cliques contain exactly four vertices each, and that
\(t_{x},t_{y}\) are the only terminals in \(C_{x},C_{y}\),
respectively. Recall that \(C_{\ell}=\{t,x,y,z\}\). Let
\(C_{x}=\{t_{x},x,x_{1},x_{2}\}\) and
\(C_{y}=\{t_{y},y,y_{1},y_{2}\}\).
\begin{table}[]
\centering
\begin{tabular}{lllll}
  \toprule
       Branch     & Vertices picked & New graph \(G_{i}\) & New terminal set \(T_{i}\)  & New parameter \(k_{i}\) \\ \midrule
  \(B_{1}\) &  \(\{t,t_{x},t_{y}\}\) & \(G-\{t,t_{x},t_{y}\}\)  & \(T\setminus\{t,t_{x},t_{y}\}\) & \(k-3\) \\
  \(B_{2}\) &  \(\{t,t_{x},y_{1},y_{2}\}\) & \(G-\{t,t_{x},y_{1},y_{2}\}\) & \(T\setminus\{t,t_{x}\}\) & \(k-4\)  \\
  \(B_{3}\) &  \(\{t,x_{1},x_{2},t_{y}\}\) & \(G-\{t,,x_{1},x_{2},t_{y}\}\) & \(T\setminus\{t,t_{y}\}\) & \(k-4\)  \\ \cline{1-5}
  \multicolumn{1}{|l}{\(B_{4}(i)\)} &  \(\{t,x_{1},x_{2},y_{1},y_{2}\}\) & \(G-\{t,x_{1},x_{2},y_{1},y_{2}\}\) & \(T\setminus\{t\}\) & \multicolumn{1}{l|}{\((k-5)\)}  \\
  \multicolumn{1}{|l}{\(B_{4}(ii)\)} &  \(\{t,x_{1},x_{2},y_{1}\}\) & \(G-\{t,x_{1},x_{2},y_{1}\}\) & \(T\setminus\{t\}\) & \multicolumn{1}{l|}{\((k-4)\)}  \\ \cline{1-5}
  \(B_{5}\) &  \(\{x\}\) & \(G-\{x\}\) & \(T\) & \(k-1\)  \\
  \(B_{6}\) &  \(\{y,z,t_{x}\}\) & \(G-\{y,z,t_{x}\}\) & \(T\setminus\{t_{x}\}\) & \(k-3\)  \\
  \(B_{7}\) &  \(\{y,z,x_{1}, x_{2}\}\) & \(G-\{y,z,x_{1},x_{2}\}\) & \(T\) & \(k-4\)  \\
 \bottomrule
\end{tabular}
\caption{Overview of \autoref{br:mega-branching-rule}
  showing the vertices picked in the solution and the resulting
  graph, terminal set, and parameter in each branch. Exactly
  \emph{one} of the two versions of rule \(B_{4}\) are applicable
  to any one instance, depending on whether the sets \(C_{x}\) and
  \(C_{y}\) share a vertex. See also the branching tree in
  \autoref{fig:branching-tree}.}
\label{tab:branching-tree}
\end{table}

\begin{figure}
  \includegraphics[scale=0.5]{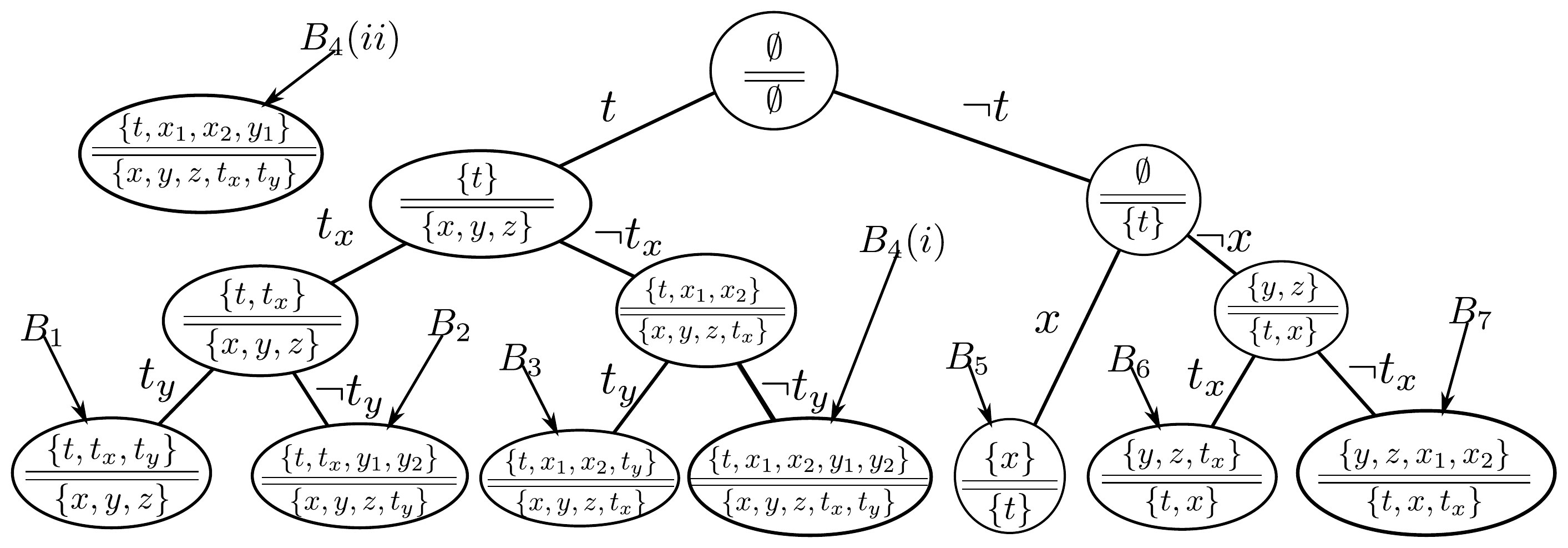}
  \caption{The seven branches of
    \autoref{br:mega-branching-rule}. An edge label denotes the
    choice we make of picking (e.g: ``\(t\)'') or excluding (e.g:
    ``\(\neg{}t\)'') a vertex in/from the solution.  The set of
    vertices that the rule has picked in the solution thus far
    appears above the double line ``\({=\joinrel=}\)'' in each
    node, and the set \emph{excluded} from the solution appears
    below \({=\joinrel=}\). The seven leaf nodes correspond to the
    seven branches \(B_{1},\cdots,B_{7}\) of the branching
    rule. The branch \(B_{4}\) takes \emph{one} of two forms
    depending on the instance. If \(C_{x}\cap{}C_{y}=\emptyset\)
    then option \(B_{4}(i)\) applies, and if
    \(C_{x}\cap{}C_{y}\neq\emptyset\) then option \(B_{4}(ii)\)
    applies.}
  \label{fig:branching-tree}
\end{figure}

We now branch on vertex \(t\) (See \autoref{fig:branching-tree}
and \autoref{tab:branching-tree}). If \(t\) is in the solution
then none of \(\{x,y,z\}\) is in the solution. Since \(x\) is not
in the solution, we have that either \(t_{x}\) is in the solution,
or both of \(\{x_{1},x_{2}\}\) are in the solution. Similarly we
get that either \(t_{y}\) is in the solution, or both of
\(\{y_{1},y_{2}\}\) are in the solution. If \(t\) is \emph{not} in
the solution then either (i) \(x\) is in the solution or (ii)
\(x\) is not in the solution and both of \(\{y,z\}\) are in the
solution. In the second case, since \(x\) is not in the solution,
we get that either \(t_{x}\) is in the solution, or both of
\(\{x_{1},x_{2}\}\) are in the solution. We summarize these seven
branches in \autoref{tab:branching-tree}. Note that \(C_{x}\)
and \(C_{y}\) could share---at most---one vertex, in which case
the set \(\{x_{1},x_{2}, y_{1},y_{2}\}\) will contain three
vertices, not four.  Rule \(B_{4}(i)\) applies when
\(|C_{x}\cap{}C_{y}|=0\) and rule \(B_{4}(ii)\) applies when
\(|C_{x}\cap{}C_{y}|=1\). In stating rule \(B_{4}(ii)\) we have
assumed that vertex \(x_{1}\) is common to both cliques;
\(C_{x}=\{t_{x},x,x_{1},x_{2}\}\) and
\(C_{y}=\{t_{y},y,x_{1},y_{1}\}\).  On every instance we apply
\emph{one} of the two variants \(B_{4}(i)\) and \(B_{4}(ii)\) of
rule \(B_{4}\), as appropriate.

\begin{branchingrule}\label{br:mega-branching-rule}
  Let \((G;T;k)\) be an instance of \SFVSC and let
  \(C_{\ell}=\{t,x,y,z\},C_{x}=\{x,t_{x},x_{1},x_{2}\},C_{y}=\{y,t_{y},y_{1},y_{2}\}\)
  be leaf nodes of the clique tree \(\mathcal{T}_{G}\) of
  \(G\). For each \(i\;;\;1\leq{}i\leq{}7\), let graph \(G_{i}\),
  terminal set \(T_{i}\), and parameter \(k_{i}\) be as described
  in \autoref{tab:branching-tree}, with the proviso that rule
  \(B_{4}(i)\) applies if and only if \(|C_{x}\cap{}C_{y}|=0\) and
  rule \(B_{4}(ii)\) applies if and only if
  \(|C_{x}\cap{}C_{y}|=1\).  The new instances are:
  \(\{(G_{i};T_{i};k_{i})\}\;;\;1\leq{}i\leq{}7\).
\end{branchingrule}

\begin{lemma} \autoref{br:mega-branching-rule} is exhaustive and
  it can be executed in time $\Oh(n + m)$ on instance which is not
  reducible by any of previously mentioned reduction or branching
  rule.
  \end{lemma}
  \begin{proof}
    Consider an instance $(G; T; k)$ of \SFVSC and cliques
    $C_{\ell}, C_x, C_y$ as described in the statement of
    branching rule. Let $(G_i; T_i; k_i)$ for each
    \(i\;;\;1\leq{}i\leq{}7\), be seven instances produced by
    \autoref{br:mega-branching-rule} when applied on
    $(G; T; k)$ as described in \autoref{tab:branching-tree}.

    $(\Rightarrow)$ From \autoref{cor:t-and-no-nbrs}, there exists
    a selective solution, say $S$. Since $\{t, x, y, z\}$ is a
    clique containing terminal, $S$ contains at least one vertex
    of this clique. By definition of selective solution, if $S$
    contains $t$ then $S$ contains none of $\{x, y, z\}$. We first
    see the implication of the fact that $S$ does not contain $x$
    in clique $C_x$.  Since $x$ is not a part of solution,
    $S \cap C_x$ either contains $t_x$ or both $x_1, x_2$. Hence
    if $S$ contains $t$ then it either contains $t_x$ or both
    $x_1,x_2$. We apply similar argument with respect to $y$. This
    implies $S$ either contains $t_y$ or both $y_1, y_2$. Hence we
    can conclude that if $S$ contains $t$ then at least one of
    sets
    $\{t_x, t_y\}, \{t_x, y_1, y_2\}, \{x_1, x_2, t_y\}, \{x_1,
    x_2, y_1, y_2 \}$ is contained in $S$. Note that set
    $\{x_1, x_2, y_1, y_2 \}$ may not contain four distinct
    elements.

    Consider the case when selective solution does not contain
    $t$. Since $t$ is a terminal, $S$ must contain at least two
    vertices from $x, y, z$. Instead of analyzing three cases, viz
    whether $\{x, y\}, \{y, z\}$ or $\{z, x\}$, we analyze two
    cases depending on whether or not $x$ is contained in $S$. We
    do this to minimize the complexity in stating this branching
    rule.  If $x$ is not contained in $S$ then it contains
    $\{y, z\}$. Moreover, $S$ either contains $t_x$ or both
    $x_1, x_2$. This implies if $S$ does not contain $t$ then $S$
    contains at least one of sets
    $\{x\}, \{y, z, t_x\}, \{y, z, t_x\}$. We have established
    that if $G$ is a \yes instance then by
    \autoref{cor:t-and-no-nbrs}, there exits a selective
    solution. This selective solution contains at least one of the
    seven sets mentioned above. Hence if $(G; T; k)$ is a \yes
    instance then at least one of the seven instances specified in
    \autoref{tab:branching-tree} is a \yes instance.
    
    $(\Leftarrow)$ We use the fact that for every $U \subseteq$ if $S'$ is a target solution of $(G - U; T \setminus U; k - |U|)$ then $S' \cup U$ is a target solution of $(G; T; k)$ to prove reverse direction.

    To execute branching rule, we need an algorithm which finds
    cliques $C_{\ell}, C_x, C_y$ which satisfies the mentioned
    property or conclude that no such cliques exists. We assume
    graph $G$ is a connected graph. If $G$ is not connected, we
    can process each of its connected component separately.  We
    first compute clique tree for of graph $G$ in time
    $\Oh(n + m)$ (\autoref{fac:clique-tree-properties}).  Once
    clique tree is computed, we can arbitrary root it at any
    internal node.  If clique tree does not have an internal node
    than it has at most two maximal clique each of which contains
    at most four vertices. In this case, problem can be solved in
    constant time.  Fix any node which is at farthest distance
    from the root as $C_{\ell}$. Let $C_{\ell} = \{t, x, y,
    z\}$. We now argue that there exists cliques $C_x, C_y$ as
    desired by branching rule assuming the input instance is not
    reducible by any reduction and branching rules mentioned
    before this branching rule.

    Let $C_p$ be the parent of $C_{\ell}$ in this rooted tree. By
    \autoref{claim:x-y-z-in-Cp}, $\{x, y, z\}$ is present in
    $C_p$. Parts~\ref{obs:at-least-two-terminal-nbrs} and
    \ref{obs:non-terminal-pairs-unique-common-terminal-nbr} of
    \autoref{obs:graph-after-simple-stuff} together imply
    that each of $x, y, z$ must have at least one terminal other
    than \(t\) as a neighbour, and that these terminals must be
    pairwise distinct. This implies that $C_p$ has at least three
    neighbors apart from $C_{\ell}$ in the clique tree. Let
    $C_x, C_y$ be any two neighbors of $C_p$ which are not in the
    path from $C_p$ to the root. This implies that distance
    between root and $C_{\ell}$ is same as distance between the
    root and $C_x$ or $C_y$. Since $C_{\ell}$ is the leaf which is
    farthest from the root, $C_x, C_y$ both are leaves in the
    clique tree. By
    part~\ref{obs:simplicials-size-four-unique-vertex} of
    \autoref{obs:graph-after-simple-stuff} we get that all
    simplicial cliques contain exactly four vertices including one
    terminal each. This proves that any leaf which is farthest
    from the root can be used as simplicial clique $C_{\ell}$ in
    branching process and concludes the proof of lemma.
  \end{proof}

We have mentioned all branching rules and now are in a position to conclude main result of this section.

\begin{theorem} There exists an algorithm which given an instance $(G, T, k)$ of \SubsetFVSChordal\ runs in time $\mathcal{O}(2^k (n + m))$ and decides whether input is \yes\ of \no\ instance. Here $n, m$ are number of vertices and edges in input graph $G$. 
\end{theorem}
\begin{proof} 
  The algorithm applies
  \Autoref{rr:no-more-clique-components,rr:only-T-neighbours,rr:no-cut-edges}
  exhaustively (i) to the input graph, and (ii) after every
  application of a branching rule. By
  \autoref{lemma:running-time-rr}, this can be done in time
  $\Oh(n + m)$. We assume that graph \(G\) is reduced with respect
  to these rules at the start of any branching rule.

  The algorithm applies least indexed applicable Branching Rule
  mentioned in this section. The correctness of branching steps
  and running time to execute branching rule follows from the
  arguments given for each case. We now claim that we have taken
  care of all possible cases.  Application of reduction rules
  implies that every maximal clique is of size at least tree and
  every non-terminal vertex is adjacent with some terminal.  If
  there exists a non-terminal vertex which is adjacent with
  exactly one terminal vertex then
  \autoref{br:non-terminals-two-terminal-nbrs} is applicable.  If
  there exists a maximal clique of size of five or more which
  contains a terminal then \autoref{br:branch-on-large-cliques} is
  applicable. Hence we can safely consider chordal graphs in which
  every maximal clique which contains a terminal is of size at
  most four. Every chordal graph has a simplicial vertex and hence
  simplicial clique. If simplicial clique is of size three then
  \autoref{br:branch-on-three-cliques} is applicable. We are now
  in case that every simplicial clique which contains a terminal
  is of size four. Notice that since
  \autoref{rr:only-T-neighbours} is not applicable, every
  simplicial clique contains a terminal. If this terminal is not a
  simplicial vertex then
  \autoref{br:branch-on-non-terminal-simplicials} is
  applicable. We are left with the case when every simplicial
  vertex is terminal. If there exists another terminal in
  simplicial clique then
  \autoref{br:simplicial-two-terminals-case} is applicable. If not
  then every simplicial clique contains exactly one terminal which
  is also a simplicial vertex. In a simplicial clique if there
  exists a pair of non-terminals which has more than one terminals
  as neighbors then \autoref{br:two-terminals-two-common-nbrs} is
  applicable. Now consider a graph which is not reducible by any
  of the branching rule. This graph satisfies all the properties
  mentioned in \autoref{obs:graph-after-simple-stuff}. Hence
  \autoref{br:mega-branching-rule} is applicable on this
  instance. This completes all the possible cases.

  We analyze the branching factor for each branching
  rule. Branching vectors for
  \Autoref{br:non-terminals-two-terminal-nbrs,br:branch-on-three-cliques,br:branch-on-large-cliques,br:branch-on-non-terminal-simplicials,br:simplicial-two-terminals-case,br:two-terminals-two-common-nbrs}
  are $(1, 1); (1, 1); (1, 2, 2); (1, 2); (2, 2, 2);$ and
  $(2, 2, 2, 2)$, respectively. Branching vector for
  \autoref{br:mega-branching-rule} is $(3, 4, 4, 5, 4, 1, 3, 4)$
  or $(3, 4, 4, 4, 4, 1, 3, 4)$ depending on whether branch
  $B_4(i)$ or $B_4(ii)$ is being used. For all these branching
  vectors, the branching factor is at most $2$. Hence the entire
  branching algorithm can be executed in time $\Oh(2^k(n + m))$.
\end{proof}

%%% Local Variables:
%%% mode: latex
%%% TeX-master: "subsetFVS"
%%% End:

%% file: conclusion.tex
\section{Conclusion}
\label{sec:conclusion}

In this article we studied \SFVSS and presented a kernel of size
$\mathcal{O}(k^2)$ with $\mathcal{O}(k)$ and $\mathcal{O}(k^2)$
vertices on the clique and independent set sides
respectively. Though this bound on the total size of the kernel is
optimal under standard complexity-theoretic assumptions, it is an
interesting open question if we can bound the number of vertices
on the independent side by $\mathcal{O}(k^{2 - \epsilon})$ for
some positive constant $\epsilon$. Another natural question is to
obtain a quadratic-size kernel for the more general \SFVSC
problem. We presented an \FPT\ algorithm running in time
$\mathcal{O}^*(2^k)$ which solves \textsc{Subset FVS} when input
graph is chordal. Under ETH, sub-exponential \FPT\ algorithms for
this problem are ruled out. Is it possible to obtain an algorithm
with a smaller base in the running time?  It would also be
interesting to take up other implicit hitting set problems from
graph theory and obtain better kernel and \FPT\ results than the
ones guranteed by \textsc{Hitting Set}.

%We have seen that both these problems can easily reduced to \textsc{$3$-Hitting Set}. These results are better than the one guaranted by the implied by results on \textsc{$3$-Hitting Set} problem. In case kernel for Split graphs, partition of split graphs into clique and ind-set helped in applying expansion lemmas. Where as in case of \FPT\ algorithm on Chordal graphs, clique-tree of corresponding graph helped us to provide vital branching point to speed up the algorithm. 

%%% Local Variables:
%%% mode: latex
%%% TeX-master: "subsetFVS"
%%% End: